\documentclass[a4paper]{article}
\usepackage[pdftex]{hyperref}
\usepackage{amsmath}
\usepackage[english]{babel}
\usepackage{latexsym}
\usepackage{amssymb}
\usepackage{amscd}
\usepackage{amsgen,amstext,amsbsy,amsopn}
\usepackage{math rsfs}

\usepackage{amsthm,epsfig,graphicx,graphics}
\usepackage[latin1]{inputenc}
\usepackage{xspace}
\usepackage{amsxtra}

\usepackage{color}

\usepackage{amstext}
\usepackage{amsxtra}

\usepackage{latexsym}

\usepackage{amscd}

\usepackage{amsgen,amstext,amsbsy,amsopn}
\usepackage{math rsfs}

\numberwithin{equation}{section}
\newcommand{\bdm}{\begin{displaymath}}
\newcommand{\edm}{\end{displaymath}}
\newcommand{\bdn}{\begin{eqnarray}}
\newcommand{\edn}{\end{eqnarray}}
\newcommand{\bay}{\begin{array}{c}}
\newcommand{\eay}{\end{array}}
\newcommand{\ben}{\begin{enumerate}}
\newcommand{\een}{\end{enumerate}}
\newcommand{\beq}{\begin{equation}}
\newcommand{\eeq}{\end{equation}}

\newcommand{\lf}{\left}
\newcommand{\ri}{\right}

\newcommand{\rv}{\vec{r}}

\newcommand{\diff}{\mathrm{d}}
\newcommand{\eps}{\varepsilon}

\newcommand{\avi}{a_i}

\newcommand{\gpf}{\mathcal{E}^{\mathrm{GP}}}
\newcommand{\gpe}{E^{\mathrm{GP}}}
\newcommand{\gpm}{\Psi^{\mathrm{GP}}}
\newcommand{\chem}{\mu^{\mathrm{GP}}}

\newcommand{\hgpf}{\hat{\mathcal{E}}^{\mathrm{GP}}}
\newcommand{\hgpe}{\hat{E}^{\mathrm{GP}}}
\newcommand{\hchem}{\hat{\mu}^{\mathrm{GP}}}

\newcommand{\Fg}{\mathcal{F}}

\newcommand{\dg}{\mbox{deg}}
\newcommand{\curl}{\mbox{curl}}

\newcommand{\tff}{\mathcal{E}^{\mathrm{TF}}}
\newcommand{\tfd}{\mathcal{D}^{\mathrm{TF}}}
\newcommand{\tfe}{E^{\mathrm{TF}}}
\newcommand{\tfm}{\rho^{\mathrm{TF}}}
\newcommand{\rtf}{R_{\mathrm{h}}}
\newcommand{\rt}{R_<}

\newcommand{\tfchem}{\mu^{\mathrm{TF}}}

\newcommand{\ann}{\mathcal{A}}

\newcommand{\optphtf}{\omega^{\mathrm{TF}}}

\newcommand{\gaintf}{H^{\mathrm{TF}}}

\newcommand{\costtf}{F^{\mathrm{TF}}}

\newcommand{\half}{\hbox{$\frac12$}}

\newcommand{\Z}{\mathbb{Z}}

\newcommand{\NN}{\mathcal{N}}

\newcommand{\CC}{\mathcal{C}}
\newcommand{\D}{\mathcal{D}}
\newcommand{\F}{\mathcal{F}}
\newcommand{\E}{\mathcal{E}}
\newcommand{\A}{\mathcal{A}}
\newcommand{\B}{\mathcal{B}}

\newcommand{\OO}{\mathcal{O}}

\newcommand{\al}{\alpha}
\newcommand{\alt}{\tilde{\alpha}}
\newcommand{\ep}{\varepsilon}

\newcommand{\Om}{\Omega}
\newcommand{\om}{\omega}

\newcommand{\dd}{\partial}

\newcommand{\one}{\mathrm{Id}}

\newcommand{\tfa}{\A ^{\mathrm{TF}}}

\newcommand{\gpfa}{\mathcal{E}^{\mathrm{GP}}_{\A}}

\newcommand{\DirC}{\delta_*}
\newcommand{\xiBL}{\xi_{\mathrm{BL}}}
\newcommand{\rtt}{\tilde{R}}
\newcommand{\rhob}{\rho}

\newcommand{\rb}{\bar{R}}
\newcommand{\rd}{R_>}

\newcommand{\barh}{\bar{h}}
\newcommand{\Fbin}{\bar{F}_{\mathrm{in}}}
\newcommand{\Fbout}{\bar{F}_{\mathrm{out}}}
\newcommand{\Fb}{\bar{F}}

\newcommand{\xt}{\tilde{x}}
\newcommand{\yt}{\tilde{y}}
\newcommand{\Gt}{\tilde{G}}
\newcommand{\rhot}{\tilde{\rho}}
\newcommand{\Kt}{\tilde{K}}

\newcommand{\at}{\tilde{\A}}
\newcommand{\ab}{\A_{\mathrm{bulk}}}
\newcommand{\Rb}{R_{\mathrm{bulk}}}

\newcommand{\Gb}{\Gamma}
\newcommand{\Cet}{\CC_{R_*}}

\newcommand{\jt}{\tilde{j}}
\newcommand{\mut}{\tilde{\mu}}

\newcommand{\muc}{\check{\mu}}

\newcommand{\ac}{\ab}
\newcommand{\hche}{\check{h}}
\newcommand{\Iche}{\check{I}}

\newcommand{\musta}{\mu_{*}}

\newcommand{\xiin}{\xi_{\mathrm{in}}}

\newcommand{\chiin}{\chi_{\mathrm{in}}}
\newcommand{\chiout}{\chi_{\mathrm{out}}}

\newcommand{\Jin}{J_{\mathrm{in}}}
\newcommand{\Jout}{J_{\mathrm{out}}}

\newcommand{\jm}{j_{\mathrm{mod}}}
\newcommand{\mum}{\mu_{\mathrm{mod}}}

\newcommand{\rgaintf}{\tilde{H}^{\mathrm{TF}}}

\newcommand{\Omc}{\Om_{\mathrm{c}1}}
\newcommand{\Omcc}{\Om_{\mathrm{c}2}}
\newcommand{\Omccc}{\Om_{\mathrm{c}3}}

\newcommand{\Hc}{H_{\mathrm{c}1}}
\newcommand{\Hcc}{H_{\mathrm{c}2}}
\newcommand{\Hccc}{H_{\mathrm{c}3}}

\newcommand{\diam}{\mathrm{diam}}

\newtheorem{teo}{Theorem}[section]
\newtheorem{lem}{Lemma}[section]
\newtheorem{pro}{Proposition}[section]

\newtheorem{defi}{Definition}[section]

\newcounter{remark}[section]
\newenvironment{rem}{\stepcounter{remark} \vspace{0,1cm} \noindent \textbf{Remark \thesection.\theremark.}}{\vspace{0,2cm}}

\begin{document}

\markboth{\scriptsize{Vortex Rings in rotating BECs}}{\scriptsize{Vortex Rings in rotating BECs}}

\title{Vortex Rings in Fast Rotating Bose-Einstein Condensates}

\author{N. Rougerie	\\	\normalsize\it	Universit\'{e} de Cergy-Pontoise et CNRS\\ \normalsize\it Laboratoire AGM \\ \normalsize\it 2, avenue Adolphe Chauvin, 95302 Cergy-Pontoise Cedex FRANCE \\ \normalsize\it \hspace{-.5 cm}}

\date{April, 2011}

\maketitle

\begin{abstract} 

When Bose-Einstein condensates are rotated sufficiently fast, a giant vortex phase appears, that is the condensate becomes annular with no vortices in the bulk but a macroscopic phase circulation around the central hole. In a former paper [\textsc{M. Correggi, N. Rougerie, J. Yngvason}, {\it Communications in Mathematical Physics} \textbf{303}, 451-308 (2011)] we have studied this phenomenon by minimizing the two dimensional Gross-Pitaevskii energy on the unit disc. In particular we computed an upper bound to the critical speed for the transition to the giant vortex phase. In this paper we confirm that this upper bound is optimal by proving that if the rotation speed is taken slightly below the threshold there are vortices in the condensate. We prove that they gather along a particular circle on which they are uniformly distributed. This is done by providing new upper and lower bounds to the GP energy. 
 
	\vspace{0,2cm}

	\vspace{0,2cm}
	
\end{abstract}

\tableofcontents

\section{Introduction}\label{sec:Intro}

A Bose gas trapped in a magnetic potential exhibits the remarkable property that, at sufficiently low temperatures,  a macroscopic fraction of the atoms are in the same quantum state. This phenomenon is referred to as \emph{Bose Einstein Condensation} and has been first observed experimentally by the Jila and MIT groups in 1995 (2001 Nobel prize in physics attributed to Cornell, Wieman and Ketterle).
Particularly interesting experiments consist in creating a Bose-Einstein condensate (BEC) by cooling atomic gases in a magnetic trap, and setting the trap into rotation. Indeed, a BEC is a superfluid and responds to rotation by the nucleation of quantized vortices. There is a rich literature on this subject, both in physics and in mathematics. We refer the reader to \cite{A,Fe1} for extensive lists of references.\\
When dealing with rotating BECs, a strong distinction has to be made according to the type of trapping potential that is being used. If the potential is harmonic \footnote{which, in this context, always means `quadratic'}, there exists a critical speed at which the centrifugal forces overcome the trapping force and drive the gas out of the trap. By contrast, in the experiments of \cite{Exp1,Exp2}, a blue-detuned laser is applied to the condensate, resulting in a stronger confinement that prevents such a behavior. To illustrate the difference, let us recall briefly what is observed or expected when rotating a BEC at increasing speeds.\\

If the trapping potential is harmonic, as in most experiments, the gas set into rotation nucleates more and more vortices as the rotation speed is increased. Eventually, when vortices become densely packed in the condensate, they form a hexagonal lattice, called Abrikosov lattice by analogy with the physics of type II superconductors (see \cite[Chapter 5]{A} for a mathematical study of this phenomenon). If the rotation is very close to the limiting speed set by the confinement it is believed that the gas will enter a regime where the atoms are strongly correlated. Such a phenomenon has not yet been observed, see \cite{Co} for a review of the physics literature on this subject and \cite{LS} for recent mathematical results.\\

If the trapping potential increases faster than the square of the distance from the center of the trap, it is theoretically possible to explore regimes of arbitrary rotation rates, as first noted in \cite{Fe2}. For a potential of this type, variational arguments have been proposed in \cite{FB} to support the following picture~: three successive phase transitions should occur when increasing the rotation speed, at which the ground state and its vortex structure change drastically. For slow rotation speeds, there are no vortices and the condensate is at rest in the rotating frame. When reaching a first critical speed $\Omc$, vortices start to appear in the condensate and organize themselves in order to minimize their repulsion. This leads to a vortex-lattice state where a dense hexagonal lattice of vortices is observed. There is then a second critical speed $\Omcc$ where the centrifugal force dips a hole in the center of the condensate. This results in a vortex-lattice-plus-hole state where the condensate is annular and still supports a dense lattice of vortices. When the third critical speed $\Omccc$, that will be our main interest in this paper, is reached, vortices are expected to retreat from the annular bulk of the BEC. There should however remain a macroscopic phase circulation (circular superflow) in the annulus. These phenomena have motivated numerous theoretical and numerical studies  \cite{FJS,FZ,KB,KF,KTU}. In particular, it is observed numerically (see e.g. \cite{FJS}) that in the transition regime between the vortex-lattice-plus-hole state and the giant vortex state, an other phase appears which displays a single vortex circle around the central hole.\\
In this paper we aim at providing an estimate of the third critical speed in the limit where the coupling constant measuring the inter-atomic interactions tends to infinity (Thomas-Fermi regime). More precisely we build on the results of \cite{CRY} where a rigorous upper bound to the critical speed has been obtained and aim at proving the corresponding lower bound. We also want to prove rigorously that when the rotation speed is close but below the critical speed, a circle of vortices is present in the annular bulk of the condensate.\\

Our setting is the same as in \cite{CDY1,CY,CRY} : we consider the Gross-Pitaevskii energy for a Bose-Einstein condensate in a two-dimensional `flat' trap, i.e. a potential equal to $0$ in the unit disc $\B$ and $+\infty$ outside, leading to a problem posed on the unit disc. Formally this is the limit as $s\rightarrow +\infty$ of a homogeneous trap of the form $V(r)=r^s$ (see \cite{CDY2}). In the rotating frame, the energy reads 
\begin{equation}\label{fonctionelleGP}
	\gpf[\Psi] := \int_{\B}  \left| \left( \nabla - i \Om x^{\perp} \right) \Psi \right|^2 - \Omega^2 r^2 |\Psi|^2 + \eps^{-2} |\Psi|^4,
\end{equation}
and it should be minimized under the mass constraint
\begin{equation}\label{masse}
\int_{\B} |\Psi | ^2 = 1.
\end{equation}
The terms in the energy have the following interpretation. The first one is the kinetic energy, including the contribution of Coriolis forces due to the transformation to the rotating frame. The second term takes into account the centrifugal force while the third one models the interactions between the atoms. We denote by $\gpe$ and $\gpm$ respectively the ground-state energy and a ground state \footnote{its existence follows by standard techniques, see e.g. \cite[Section 2]{IM1}} (a priori not unique) of (\ref{fonctionelleGP})
\begin{equation}\label{fondamental}
 \gpe := \inf _ {\int_{\B} |\Psi | ^2 = 1} \gpf [\Psi] = \gpf [\gpm].
\end{equation}
Any GP minimizer satisfies the Euler-Lagrange equation
\begin{equation}
	\label{equationGP}
	- \left( \nabla -i \Om x^{\perp} \right) ^2 \gpm - \Om ^2 r^2 \gpm+ 2 \eps^{-2} \lf| \gpm \ri|^2 \gpm = \chem \gpm,
\end{equation}
where $\chem$ is the chemical potential associated with the mass constraint (\ref{masse}).\\ 
To evaluate the critical speed for the appearance of the giant vortex we consider the asymptotic behavior of $\Psi^{\rm GP}$ and $\gpe$ as $\eps\to 0$. For $|\log \ep |\ll \Om \ll \frac{1}{\ep ^2 |\log \ep|}$ it has been proved \cite{CY,CPRY} that many vortices are present and evenly distributed in the bulk of the condensate. Also the critical speed for the condensate to develop a hole (second critical speed in the above terminology) is estimated. It behaves as  $\Om \sim 2 \left( \sqrt \pi \ep \right)^{-1}$.\\ 
In a recent work \cite{CRY} (see also \cite{R} where a related problem is treated) we have proved that if 
\begin{equation}\label{ordreOmega}
\Om = \frac{\Om_0}{\ep ^2 |\log \ep|} 
\end{equation}
where $\Om_0$ is a fixed constant chosen above the critical value $2(3\pi )^{-1}$
%\begin{equation}
%	\label{OmegaCritique}
%	\Omega_c = \frac{1}{3\pi}
%\end{equation}
then the support of any GP minimizer is essentially vortex free. We have thus confirmed that a transition occurs in the regime where $\Om$ scales as (\ref{ordreOmega}) from a vortex lattice plus hole state (as described in \cite{CY}) to a pure giant vortex state. We have also provided an upper bound to the critical speed at which the transition is expected to occur : 
\begin{equation}\label{seuil}
\Omccc \leq \frac{2}{3\pi \ep ^2 |\log \ep|}(1+o(1)) .
\end{equation}
In the present paper we aim at providing the corresponding lower bound. Namely, we want to show that if 
\begin{equation}\label{seuil2}
 \Om < \frac{2}{3 \pi \ep ^2 |\log \ep|} 
\end{equation}
then there are individual vortices in the bulk of the condensate.\\
More precisely we consider rotation speeds of the form
\begin{equation}\label{regime}
 \Om = \frac{2}{3 \pi \ep ^2 |\log \ep|} - \frac{\Om_1}{\ep^{2}|\log \ep|} 
\end{equation}
where $\Om_1 = \Om_1 (\ep)$ satisfies
\begin{eqnarray}\label{Omega1}
  \Om_1 &>& 0 
\\ |\Om _ 1| &\ll& 1 \label{Omega1bis}
\\ |\Om _ 1| &\gg& \frac{\log |\log \ep |}{|\log \ep|} \label{Omega1ter} .
\end{eqnarray}
The first two assumptions on $\Om_1$ are natural : we want to study slightly subcritical speeds, namely we want to be below the threshold $2(3 \pi \ep ^2 |\log \ep|) ^{-1}$ by a relatively small amount. The third assumption is technical. We do believe that the results we are going to present stay true without this assumption but the proofs require some improvements to apply to a regime where $\Om_1$ is allowed to be extremely small.\\

Stated loosely, our result is that in the regime (\ref{regime}), with $\Om_1$ satisfying (\ref{Omega1}), (\ref{Omega1bis}) and (\ref{Omega1ter}), there are vortices in the annular bulk of the condensate. They are uniformly distributed on a particular circle of radius $ R_*$ and we are able to estimate their number, which is of order $\Om_1\ep ^{-1}$.\\ 
If one assumes that there is a uniquely determined third critical speed separating the giant vortex phase from a phase containing vortices in the bulk, these results complete the proof started in \cite{CRY} that the third critical speed is asymptotically equal to $2\left( 3 \pi \ep ^2 |\log \ep| \right)^{-1}$ in the limit of small $\ep$ :
\begin{equation}\label{Omtrans}
\Omccc = \frac{2}{3\pi \ep ^2 |\log \ep|}(1+o(1)).
\end{equation}
In fact, there exists actually a uniquely determined third critical speed. This can be seen from the fact that in our analysis we use the asumptions on $\Om_1$ only when proving that the vortices are distributed along a circle. The existence of vortices in the bulk follows from our analysis under the weaker assumption that $\Om = \al \left( \ep^2 |\log \ep|\right) ^{-1}$ with $\al < 2 (3\pi) ^{-1}$.\\
On the other hand, determining the distribution of vortices for rotation speeds of the form $\Om = \al \left( \ep^2 |\log \ep|\right) ^{-1}$ with $\al < 2 (3\pi) ^{-1}$ requires additional ingredients. As this will be the subject of a future work we do not elaborate more on this point here and leave aside for the present contribution the proof of the existence of a unique $\Omccc$. 

\subsection{Statement of the Main Results}\label{sousec:Results}

We first recall some notation that was introduced in \cite{CRY}. The following Thomas-Fermi energy functional will be of importance in our analysis because it gives the leading order of the GP energy when $\ep$ becomes small with the above scaling of the rotation speed : 
\begin{equation}\label{fonctionelleTF}
	\tff[\rho] : = \int_{\B}  \left( - \Omega^2 r^2 \rho + \eps^{-2} \rho^2 \right).
\end{equation}	
Here $\rho$ is a matter density, it thus plays the role of $| \Psi | ^2$. In particular it is positive. We minimize $\tff$ with respect to all positive $\rho$'s normalized in $L^1$
\[
\tfe = \inf \left\lbrace \tff [\rho], \: \rho \geq 0,\: \int_{\B} \rho = 1 \right\rbrace
\]
and find a unique ground-state given by the radial density
\begin{equation}\label{minimiseurTF}
	\tfm(r) = \frac{1}{2} \lf[ \eps^2 \tfchem + \eps^2 \Omega^2 r^2 \ri]_+ = \frac{\eps^2 \Omega^2}{2} \lf[ r^2 - \rtf^2 \ri]_+,
\end{equation}		
where the chemical potential is fixed by normalizing $ \tfm $ in $ L^1(\B) $, i.e.,
\begin{equation}
	\label{tfchem}
	\tfchem = \tfe + \eps^{-2} \lf\| \tfm \ri\|^2_2.
\end{equation}
Note that the TF minimizer is a compactly supported function, since it vanishes outside $ \tfa $, i.e., for $ r \leq \rtf $, where
\begin{equation}\label{anneau}
	\rtf = \sqrt{1 - \frac{2}{\sqrt{\pi} \eps \Omega}}, \quad \tfa = \left\lbrace \rtf \leq r \leq 1\right\rbrace .
\end{equation}
The corresponding ground state energy can be explicitly evaluated and is given by
\begin{equation}
	\label{TFe}
	\tfe = - \Omega^2 \lf( 1 - \frac{4}{3 \sqrt{\pi} \eps \Omega} \ri). 
\end{equation}
We stress that the annulus $ \tfa $ has a shrinking width of order $ \eps |\log\eps| $ and that the leading order term in the ground state energy asymptotics is $ - \Omega^2 $, which is due to the convergence of $ \tfm $ to a delta function supported at the boundary of the trap.\\

When $\ep \rightarrow 0$ with the scaling (\ref{ordreOmega}), the leading order of the energy is given by the TF energy and the GP density  $|\gpm|^2$ is very close to $\tfm$. In particular, the mass of $\gpm$ becomes exponentially small in the central hole $\B \setminus \tfa$ and the properties of the ground state and ground state energy are well approximated by a functional restricted to the TF annulus. For technical reasons however it is necessary to consider a slightly larger annulus
\beq
	\label{the annulus}
	\A : = \lf\{ \rv  \: : \: \rt \leq r = \left|\:\vec{r}\:\right| \leq 1 \ri\},
\eeq 
with the choice \footnote{the exponent $8/7$ in formula \eqref{the inner radius} is choosed for definiteness : any exponent between $1$ and $7/6$ would do, see the discussion at the beginning of \cite[Section 3.1]{CRY}}
\beq
	\label{the inner radius}
	\rt : = \rtf - \eps^{8/7}.
\eeq 

In order to give a precise statement, we need to be more specific. Namely we need to take into account the leading order (i.e. $\tfe$) and the sub-leading order of the GP energy via a simplified functional : For any $a \in \Z$ and any \textbf{real-valued} wave-function $\phi$ we introduce the `giant vortex energy':
\begin{equation}\label{energieVG}
\hgpf_{\A,a} [\phi]:= \gpfa [\phi \: e^{i\left( \left[\Om\right] - a \right)\theta }] = \int_{\A} \left| \nabla \phi \right| ^2 +\frac{1}{\ep ^2} |\phi |^4 +\left( \frac{\left( \left[\Om\right] - a \right) ^2}{r^2}-2\Om \left( \left[\Om\right] - a \right) \right)|\phi | ^2
\end{equation} 
where $[\Om]$ stands for the integer part of $\Om$. Let us denote $g_{\A,a}$ and $\hgpe_{\A,a}$ the ground state (unique up to a sign that we fix by requiring that $g_{\A,a}>0$) and ground state energy of this functional, see \cite[Proposition 2.3]{CRY}. We recall the following result (Proposition 3.2 in \cite{CRY})

\begin{pro}[\textbf{\cite{CRY} Properties of the optimal phase and associated density}]\label{pro:optimalphase}\mbox{}\\
There exists at least one $\om \in \Z$ minimizing $\hgpe_{ \A, a}$ with respect to $a$. Moreover we have
\beq\label{est omega}
\omega=\frac2{3\sqrt\pi \eps} \left(1+ \OO(|\log\eps|^{-1})\right)
\eeq
and
\begin{equation}\label{compatibility}
\int_{\A} g_{\om} ^2 \left( \Om -  \left( [\Om] - \om \right) r ^{-2}\right) = \OO (1)
\end{equation}
where $g_{\om}$ is the unique positive normalized minimizer of $\hgpf_{\A,\om}$ .
\end{pro}
Note that the above proposition was proved in \cite{CRY} with a slightly different assumption on $\Om$ but the results stays true with no modification of the proof under asumptions \eqref{regime} to \eqref{Omega1ter}.\\

The key to the results of \cite{CRY} in the super-critical case is that the asymptotics of the ground state and ground state energy of (\ref{fonctionelleGP}) are described with a very good precision by $g_{\A,\om}e^{i\left( \left[\Om\right] - \om \right)\theta }$ (we will denote $g_{\A,\om}$ by $g$ for short) and $\hgpe_{\A,\om}$ respectively.\\ 
In particular, a minimizer of (\ref{fonctionelleGP}) contains a large vorticity similar to that generated by a giant vortex located at the origin, which is apparent in the phase factor $e^{i\left( \left[\Om\right] - \om \right)\theta }$. To identify individual vortices in the annulus where the mass is concentrated we thus have to look at a reduced function as in Section 5 of \cite{ABM}. More precisely we extract the giant-vortex density and phase factor by introducing
\begin{equation}\label{u}
 u : = \frac{\gpm}{g \: e^{i\left([\Om] - \om \right) \theta} }
\end{equation}
which is defined only on $\A$.\\

Our main results deal with the vortices of $u$ and are stated, as is usual, in terms of \emph{vorticity measures} of $u$.
There actually exist two main paths to the definition of an appropriate vorticity measure. One is to use an `intrinsic' vorticity measure as, e.g., in \cite{ABM,SS}. This measure is defined as the $\curl$ of the superfluid current :
\begin{equation}\label{mesurevorticite}
 \mu := \curl (iu, \nabla u) = \frac{1}{2} \curl \left( i \left(u \nabla \bar{u} - \bar{u} \nabla u  \right)\right).
\end{equation}
The other one, used for example in \cite{CY,CPRY}, is to explicitly identify a collection of balls \footnote{$B(a_j,r_j)$ denotes the ball of center $a_j$ and radius $r_j$} $B_j = B (a_j,r_j),j\in J$ on the boundary of which the degree $d_j$ of $u$ is well defined and set 
\begin{equation}\label{mesurevorticite exp}
\mu_e := 2\pi \sum_{j\in J} d_j \delta_{a_j}
\end{equation}
with $\delta_{a_j}$ the Dirac mass at $a_j$. Actually, a classical technique (`Jacobian estimates' due to Jerrard and Soner \cite{JS}) allows to relate these two notions of vorticity, see Section \ref{sousec:Sketch}. In this paper we will prove results about both the `intrinsic vorticity measure' (Theorem \ref{theo:vorticity}) and an `explicit vorticity measure' (Theorem \ref{theo:vorticity exp} that includes the complete definition of $\mu_e$) : both vorticity measures are close to a Delta function concentrated on some particular circle of radius $R_*$ (see Appendix A for the definition of $R_*$).\\

We have to deal with a technical point before stating our main theorems. In \cite{CRY} we had identified the `short-range' energetic cost (by opposition to a `long-range' cost that we will identify in the present paper) of a vortex of degree $1$ at some point $\vec{r}$. It is given in terms of the following \emph{cost function}
\beq \label{fonctioncout}
H(r) :=\frac{1}{2} g^2(r)|\log\eps|+F(r)
\eeq
where $F$ is defined as 
\begin{equation}\label{F}
F(r):= 2 \int_{\rt} ^r \: g ^2  (s) \left(\Om  s- \left([\Om] - \om \right)\frac{1}{s}\right) ds .
\end{equation}
The dependence on $g$ makes it very difficult to identify vortices close to the inner boundary of $\A$. Indeed, $\tfm$ vanishes on $\partial B_{R_h}$, so $g$ is very small there (compared for example with its value on $\partial \B$) and the energetic cost of a vortex too close to $\partial B_{R_h}$ can not be taken into account with our method. To avoid this difficulty we limit ourselves to a statement about the asymptotics of $\mu$ on a smaller domain $\ab$ where the bulk of the mass is concentrated and the density is large enough:
\begin{equation}\label{defiAb}
\ab := \left\{ \vec{r} \: | \: \Rb \leq r \leq 1\right\}, \quad \Rb := R_h + \ep |\log \ep| \Om_1 ^{1/2}.
\end{equation}
Note that, it follows from the analysis in \cite[Section 2.1]{CRY} (the relevant result is recalled in Proposition \ref{cry pro:GP exp small} below)
\[
\int_{\ab} |\gpm | ^2  = 1 - \OO (\Om _ 1 ^{1/2}).
\]
This justifies the notation : $\ab$ indeed contains the bulk of the mass when $\Om_1 \ll 1$. The radius $R_*$ appearing in the theorems below is an approximate minimizer of the cost function $H$ (see Appendix A).\\

Before stating the vorticity asymptotics we need to anticipate on the energy asymptotics of Theorem \ref{theo:energy} and introduce the following notions. For any Radon measure $\nu$ supported in $\A$ we define $h_{\nu}$ as the unique solution to the elliptic problem 
\begin{equation}\label{defihnu}
	 \begin{cases}
		-\nabla \left( \frac{1}{g^2} \nabla h_{\nu}\right) = \nu \mbox{ in } \A \\
		h_{\nu} = 0 \mbox{ on } \partial \A.
	\end{cases} 
\end{equation}
Note that $g^2$ is bounded below on $\A$ because it is the ground state of a one-dimensional Schrödinger operator (see Lemma \ref{lem:rhob} below). There is thus no difficulty in defining $h_{\nu}$ as above. Next we introduce an `electrostatic energy' associated to $\nu$
\begin{equation}\label{defiInu}
 I(\nu) := \int_{\A} \frac{1}{g^2} \left| \nabla h_{\nu} \right| ^2.
\end{equation}
Then
\begin{equation}\label{defiIstar}
 I_* := \inf_{\nu \in \D_* ,\int \nu = 1} I(\nu) = I(\DirC )
\end{equation}
where the infimum is taken over the set $\D_*$ of positive Radon measures with support on the circle of radius $R_*$.
The fact that $\DirC$ is a solution (actually, the unique solution) of the above minimization problem will be proved in Proposition \ref{pro:electrostatic} below. The energy (\ref{defiInu}) is similar to that of a charge distribution $\nu$ in a shell with inhomogeneous conductivity described by $g^{-2}$.\\

In the following theorems, we use the norm
\begin{equation}\label{normeg}
\Vert \nu \Vert_g := \sup_{\phi \in C ^1 _c (\ab)} \frac{\left| \int_{\ab} \nu \phi \right|}{\left(\int_{\ab} \frac{1}{g ^2} |\nabla \phi| ^2 \right)^{1/2} + \ep |\log \ep| \Vert \nabla \phi \Vert_{L ^{\infty}(\ab)}},
\end{equation}
where $\nu$ is a Radon measure, to estimate $\mu$ and $\mu_e$. 

\begin{teo}[\textbf{Asymptotics for the intrinsic vorticity}]\label{theo:vorticity}\mbox{}\\
Let $\Om_1$, $u$ and $\mu$ be defined as above. Let $R_*$ be the radius defined in (\ref{AdefiRstar}), which satisfies 
\begin{equation}\label{defiRstar}
R_*  = \sqrt{R_h ^2  + \left(\eps \Omega \right)^{-1} \left(\frac{1}{\sqrt{\pi}} + \OO (\Om_1 ) \right)}.
\end{equation}
We denote $\DirC$ the normalized arclength measure on the circle of radius $R_*$.\\
There holds
\begin{equation}\label{resultvorticity}
 \left \Vert \mu  + \frac{H(R_*)}{2 I_*}  \DirC  \right\Vert_g \ll \frac{\Om_1}{\ep} \leq C \left\Vert \frac{H(R_*)}{2 I_*}  \DirC \right\Vert_g
\end{equation}
in the limit $\ep \rightarrow 0$, where $H$ is defined as in (\ref{fonctioncout}) and $I_*$ as in \eqref{defiIstar}.
\end{teo}

\begin{rem}\label{rem:com vortic}
\begin{enumerate}
\item The vorticity measure $\mu$ is really the intrinsic vorticity quantity associated with $u$. Indeed, we expect that $|u|\sim 1$ (which implies $\curl (iu,\nabla u) \sim 0$) outside of the vortex cores, so formally, for any domain which boundary $\sigma$ does not intersect any vortex core,
\[
2\pi \deg \{u, \sigma \} \sim \int_{\sigma} (iu,\dd_{\tau} u) 
\]
is given by the integral of $\mu$ over the domain (using Stokes' formula). Thus, if the vortex cores are small, we should have $\mu \sim 2 \pi \sum_j d_j \delta_{a_j}$ where the $d_j$ are the degrees of the vortices and $a_j$ their locations. This is stated rigorously in \eqref{mue-mui} below.\\
Note that  by definition the vortices of $u$ are identical to those of $\gpm$. 

\item Let us explain a bit why (\ref{resultvorticity}) indeed captures the leading order of the vorticity. The radius $R_*$ is such that 
\[
0 > H(R_*) \sim - C \frac{\Om_1}{\ep}.
\] 
Taking a positive test function (or rather a sequence of test functions) in $C^1 _c (\ab)$ we have
\[
  -  \frac{H(R_*)}{2 I_*} \int_{\ab}  \phi \DirC = -  \frac{ H(R_*)}{ 4\pi I_*} \int_{r= R_*} \phi (r,\theta) d\theta.
\]
If in addition $\phi$ is radial then
\[
  - \frac{H(R_*)}{2 I_*} \int_{\ab}   \phi  \DirC = -  \frac{ H(R_*)}{ 2 I_*} \phi (R_*).
\]
Choosing $\phi$ such that $\max \phi = \phi (R _*)$  and recalling that the thickness of the annulus $\ab$ is of order $\ep |\log \ep|$, one can obviously construct $\phi$ such that 
\[
\phi(R_*) \geq C \Vert \nabla \phi \Vert_{L^{\infty}(\ab)} \ep |\log \ep|.
\]
If in addition the support of $\phi$ is included in a region where $g^2\geq C(\ep |\log \ep|)^{-1}$ we have
\[
\left( \int_{\ab} \frac{1}{g^2} |\nabla \phi| ^2 \right)^{1/2} \leq C \ep |\log \ep| \Vert \nabla \phi \Vert_{L^{\infty}(\ab)} \leq C \Vert \phi \Vert_{L^{\infty}(\ab)}. 
\]
Thus
\begin{multline}\label{commentvorticity1}
 - \frac{H(R_*)}{2 I_*} \int_{\ab}  \phi \DirC  \geq - C  \ep |\log \ep |\frac{H(R_*)}{ I_*}  \Vert \nabla \phi \Vert_{L^{\infty}(\ab)} 
\\ > C \Om_1 |\log \ep| \Vert \nabla \phi \Vert_{L^{\infty}(\ab)} > C \frac{\Om_1}{\ep} \left( \int_{\ab} \frac{1}{g^2} |\nabla \phi| ^2 \right)^{1/2} 
\end{multline}
which proves 
\[
\frac{\Om_1}{\ep} \leq C \left\Vert \frac{H(R_*)}{2 I_*}  \DirC \right\Vert_g.
\]

\item Let us comment on the norm (\ref{normeg}) that we use in the theorem. In the denominator of the definition (\ref{normeg}), the first term is rather natural in view of the energy estimates below. Indeed, the quantity 
\[
\left( \int_{\ab} g^{-2} |\nabla \phi | ^2\right)^{1/2}
\]
 defines a norm on $C^1_c (\ab)$ that is associated with the problem (\ref{defiIstar}). The second term appears when regularizing $\mu$ in the course of the proof.\\  
Of course if the above norm could be controlled by $\ep |\log \ep| \Vert \nabla \phi \Vert_{L^{\infty}(\ab)} $ for any test function, then (\ref{resultvorticity}) would be equivalent to
\begin{equation}\label{commentvorticity2}
 \left\Vert \mu + \frac{H(R_*)}{2 I_*} \DirC \right\Vert_{(C_c ^1 (\ab)) ^*} \ll \Om_1 |\log \ep|,
\end{equation}
whereas the $(C_c ^1 (\ab)) ^*$  norm of both terms in the above right-hand side are $\gtrsim \Om_1 |\log \ep|$ as demonstrated by (\ref{commentvorticity1}). However, such a control is not possible because $g^2$ is not uniformly bounded below by $C(\ep |\log \ep|)^{-1}$ on $\ab$ : close to the inner boundary of $\ab$ it is of order $\Om_1 ^{1/2} (\ep |\log \ep|)^{-1}$. The norm $\Vert . \Vert_g$ thus takes into account the effect of the inhomogeneous density background.

\item If we define a change of coordinates mapping $\ab$ to a fixed annulus of radius $1$ and width say $1/2$ and denote $\tilde{R}_*$ the image of $R_*$ under this change of coordinates and $\tilde{\mu}$ the push-forward of $\mu$ we will have
\begin{equation}\label{ResultRescale}
- 2 I_* \frac{\tilde{\mu}}{H(R_*)} \rightharpoonup \tilde{\DirC}
\end{equation}
in the weak sense of measures. Here $\tilde{\DirC}$ is the normalized arc-length measure on the circle of radius $\tilde{R}_*$. We do not want to be too precise here in order not to mislead the reader : such a change of coordinates will not be used in the paper. The reason is that, applied to the energy functional, it would lead to a highly anisotropic (and $\ep$-dependent) kinetic energy term because of the shrinking width of the original annulus $\ab$. The convergence (\ref{ResultRescale}) is useful however to identify the optimal number of vortices : $I_*$ is of the order of a constant, $-H(R_*)$ of the order of $\Om_1 \ep ^{-1}$ so from (\ref{ResultRescale}) the number of vortices is to leading order equal to $-H(R_*) (2 I_*) ^{-1} \propto \Om_1 \ep ^{-1}$, as we announced. 
\end{enumerate}
 
\end{rem}

\begin{teo}[\textbf{Asymptotics for the explicit vorticity}]\label{theo:vorticity exp}\mbox{}\\
With the assumptions and notation of Theorem \ref{theo:vorticity}, there exists a finite collection of disjoint balls 
\[
B_k = B(a_k,r_k), k\in K
\]
such that $B_k \subset \ab$ and $|u| > 1-o(1)$ in the limit $\ep \to 0$ on $\dd B_k$ for any $k\in K$. Defining 
\begin{equation}\label{mesurevorticite exp 2}
\mu_e := 2\pi \sum_{k\in K} d_k \delta_{a_k},
\end{equation}
with
\[
d_k = \deg (u,\dd B_k),
\]
there holds 
\begin{equation}
\left\Vert \mu_e - \mu \right\Vert_{(C^1_c (\ab)) ^*} \leq C \Om_1 ^{3/2} |\log \ep| \ll  \left\Vert \mu \right\Vert_{(C^1_c (\ab)) ^*} \label{mue-mui}
\end{equation}
and
\begin{equation}
 \left\Vert \mu_e + \frac{H(R_*)}{2 I_*} \DirC \right\Vert_g \ll \frac{\Om_1}{\ep} \leq C \left\Vert \frac{H(R_*)}{2 I_*}  \DirC \right\Vert_g \label{asympt mue}.
\end{equation}
\end{teo}

\begin{rem}\label{rem:com vortic exp}
\begin{enumerate}
\item The `vortex balls' in the above theorem satisfy some properties that we now explain. In the parameter regime we consider, there is a `bad' region on which we do not have enough information to say something meaningfull about the vortices of $u$. What saves the day is that we have a control on the area of this region, saying that the complementary `good' region covers the major part of the annulus.\\
After removing fom the `good' region a boundary layer on which the matter density is too small, we can cover the possible zeros of $u$ in the remaining part by the balls in the theorem, and the area covered by this balls is much smaller than that of the annulus where the bulk of the mass of the condensate resides.
\item We actually prove explicitly that most of the balls carry a positive degree and that the degree carried by the balls outside some region close to the circle of radius $R_*$ is small (see Equations \eqref{borne inf vorticite}, \eqref{vorticite moins}, \eqref{vorticite plus}, \eqref{propD} and \eqref{calculD2}). All this is included in the statement \eqref{asympt mue}. We note that \eqref{asympt mue} also implies that the vortex balls must be in some sense evenly ditributed along the circle $\Cet$, but it seems difficult to prove it without using the intrinsic vorticity measure : the electrostatic energy \eqref{defiInu} that appears naturally in this context is not well-defined for a measure of the form \eqref{mesurevorticite exp 2}. The extra regularity of the intrinsic vorticity measure is thus used in a crucial way and there does not seem to be an easy way to prove \eqref{asympt mue} without using \eqref{mue-mui}, i.e. the Jacobian estimate.\\
Note also that the discussion in Item 2 of Remark \ref{rem:com vortic} implies
\[
\left\Vert \mu \right\Vert_{(C^1_c (\ab)) ^*} \geq C \Om_1 |\log \ep|,
\]
thus justifying the second inequality in \eqref{mue-mui}.
\end{enumerate}

\end{rem}

We now present our results on the ground state energy $\gpe$. The problem (\ref{defiIstar}) naturally appears in our analysis because it is linked in a crucial way to the asymptotics of the energy $\gpe$, as demonstrated in the 

\begin{teo}[\textbf{Energy asymptotics}]\label{theo:energy}\mbox{}\\
Let $\Om_1$ be as above and $\ep$ be small enough. There holds   
\begin{equation}\label{resultenergy}
 \gpe = \hgpe _{\A,\om} - \frac{H(R_*) ^2}{4 I_*} (1+o(1))
\end{equation}
in the limit $\ep \rightarrow 0$.
\end{teo}
The term $-H(R_*) ^2(4 I_*) ^{-1}$ in \eqref{resultenergy}, of order $\Om_1 ^2 \ep ^{-2}$, is a small correction compared to $\hgpe_{\A,\om}$.  We refer to \cite[Remark 1.4]{CRY} for a discussion of the different contributions to the latter.\\ 

The way the value $\frac{H(R_*) ^2}{4 I_*}$ appears is through the following minimization problem :
\begin{equation}\label{energyrenorm}
- \frac{H(R_*) ^2}{4 I_*} = \inf_{\nu \in \D_*} \left( \int_{\A} \frac{1}{g^2} \left| \nabla h_{\nu} \right| ^2 + H(R_*) \int_{\A} \nu \right).  
\end{equation}
The above functional of $\nu$ describes the energy of a vortex distribution corresponding to a vorticity measure $\nu$ concentrated on $\Cet$.
The right-hand side can be thought of as a renormalized energy (in analogy with \cite{BBH,SS}). The first term represents the interaction of the vortices, which is computed through the potential $h_{\nu}$ that they create (other formulae, including a particularly useful Green representation will be used in the sequel). The second term represents the energy gain of having a vortex partly compensating the rotation field. The unique minimizer of (\ref{energyrenorm}) is given by $-(2I_*) ^{-1} H(R_*) \DirC $, which explains why the vorticity measure of $u$ has to be close, at least in the sense of Theorem \ref{theo:vorticity}, to this particular measure. Note the close analogy between our renormalized energy and that obtained in \cite[Theorem 1.1]{ABM}, the difference being essentially apparent in the weight $g^{-2}$.\\

We recall that superfluids or superconductors in simply-connected geometries generically exhibit vortex concentration around isolated points (see \cite{A,SS} and references therein). Here, although the original domain $\B$ is simply connected, the strong centrifugal forces impose an annular form to the condensate. As a consequence, \emph{vortex concentration along a curve} occurs for this model at rotation rates $\Om \propto \ep ^{-2} |\log \ep| ^{-1}$. In \cite{ABM}, a model case for vortex concentration along a curve in the regime $\Om \propto |\log \ep|$ is considered and the limiting vorticity measure is identified via $\Gamma$-convergence. \\
Other physical situations where concentration along a curve occurs include superconductors with normal inclusions \cite{AB1,AB2}, the case described in \cite{AAB} of a BEC whose trapping potential impose an annular shape in the regime $\Om \propto |\log \ep|$, and that of two superconductors with different physical properties in contact along a circle \cite{Kac}. In those three cases the question of the distribution of the vortices along the curve that we solve here is left open. In particular our result is as far as we know the only one in the literature where the limiting vorticity measure is computed for a case where there is vortex concentration along a curve against an inhomogeneous background. Indeed, the inhomogeneous density profile $g^2$ is a new feature compared to the situation in \cite{ABM}. Our method could be used in the context of \cite{AAB} where a similar inhomogeneity prevents from using directly the analysis of \cite{ABM}.\\
A formula such as (\ref{Omtrans}) seems to be absent from the physics literature : in \cite{FB} the critical speed is estimated by comparing the energy of the giant vortex state to that of the vortex-lattice-plus-hole state. One then argues that the critical speed is that at which the former is smaller than the latter. This is too rough an estimate, for it does not take into account the fact that a circle of vortices appears at the transition. As a consequence, the formula we find by rigorous analysis differs from that given in \cite{FB} \footnote{the latter, given in formula (20) of \cite{FB}, is found to be $\frac{1}{9\pi \ep ^2 |\log \ep|}$ in our units}. We note that for a problem in a slightly different setting (namely, for a condensate trapped in a harmonic plus quartic trap) the papers \cite{FZ,KF} propose methods to numerically compute the critical speed. This is done by comparing the giant vortex energy to the energy of a condensate containing a single ring of vortices. Varying the parameters defining the latter state (size of the vortex cores, number of vortices ...) to find the best possible choice yields the critical speed. It is again defined as the one above which the giant vortex energy is the smaller of the two. This method does not provide an explicit formula such as (\ref{Omtrans}) and of course neither does a direct numerical minimization of the GP energy.\\ 

The response of type II superconductors to imposed external magnetic fields and that of superfluids to rotation of the container bear some striking similarities. This analogy between superfluidity (usually described by the Gross-Pitaevskii theory) and superconductivity (usually described by the Ginzburg-Landau theory) has been well-known to physicists for tens of years. More recently it turned out that mathematical tools originally developed for the GL theory could be successfully used in GP theory. Also there is an analogy between our (somewhat informal) terminology about critical speeds and that of critical fields in GL theory. In particular, the analogy between the first critical speed and the field $\Hc$ is well-known and of great use in the papers \cite{AAB,IM1,IM2}.\\
We want to emphasize however that the Gross-Pitaevskii theory in the regime we consider largely deviates from the Ginzburg-Landau theory. This is due to the presence of a mass constraint (which reflects the fact that a BEC has no normal state to relax to, contrary to a superconductor) and of centrifugal forces (which could be interpreted as electric fields in GL theory). In particular, the second and third critical speeds we have informally defined have little to do with the second and third critical fields in GL theory (to our knowledge this has been first noticed in \cite{FB}). The discrepancy between GP theory and GL theory arises when the centrifugal forces can no longer be neglected, namely close to the second critical speed. We refer to \cite{CDY1,CDY2,CY} for this aspect of the theory.\\
As for the third critical speed, that is our main concern here, it turns out that it bears more similarities with $\Hc$ than with $\Hcc$ or $\Hccc$. In fact, the transition happening there can be seen as a $\Hc$ type transition but \emph{backwards}, which is apparent in the fact that vortices \emph{disappear} from the bulk of the condensate when the rotation speed is increased. As a consequence, many mathematical tools characteristic of the study of type II superconductors (strictly) between $\Hc$ and $\Hcc$ will be of great use in this paper, as they have been in \cite{CRY}. We refer to Section 1.3 of that paper for a heuristic explanation of this surprising fact. 

\subsection{Formal Derivations}\label{sousec:Sketch}

The main intuitions in the proofs of our results are :
\begin{itemize}
\item The rotation field being along the azimuthal vector $\vec{e}_{\theta}$ with positive amplitude, a vortex of negative degree can only create energy. All the vortices should then be of positive degree.
\item $\mu$ should necessarily have its support on $\Cet$ because the cost function (\ref{fonctioncout}) is minimum there. Indeed, the cost function gives the short range energetic cost (self-interacting kinetic energy plus interaction with the rotation field) of a vortex situated at a particular location.
\item It is well-known that two vortices of same degree should repel each other while two vortices of opposite degree should attract each other. The vortices, all of positive degrees and located on $\Cet$ should organize in order to minimize their repulsion. This effect leads to the uniform distribution on the circle. 
\end{itemize}

Let us explain how the renormalized energy (\ref{energyrenorm}) appears in our setting. A large part of the analysis, in particular the 
introduction of the cost function (\ref{fonctioncout}),  has been carried out in \cite{CRY}. It is reminiscent of the method of \cite{AAB} 
although important new difficulties occur due to the different parameter regime we consider.\\
A first step is to extract from $\gpe$ the contribution of the giant vortex profile and phase factor. This uses a classical method of energy 
decoupling originating in \cite{LM} and an exponential decay property for $\gpm$ in $\B \setminus \A$ (see \cite[Propositions 2.2 and 3.1]{CRY}). 
We obtain \footnote{The notation $\OO (\ep ^{\infty})$ refers to a quantity decreasing faster than any power of $\ep$, e.g. exponentially fast.}
\begin{equation}\label{sketch:decouple}
\gpe = \hgpe_{\A,\om} + \E [u] +\OO (\ep ^{\infty}).
\end{equation} 
The reduced energy functional $\E$ is defined as follows
\beq\label{sketch:eomega0} 
\mathcal E[v] :=\int_{\mathcal \A} g^2|\nabla v|^2-2 g^2\vec B\cdot (iv,\nabla v)+\frac {g^4}{\eps^2}(1-|v|^2)^2 
\eeq
where
\beq
	\label{vec B 1}
	\vec B(r) : =B(r) \vec e_\theta = \lf( \Omega r - \lf( [\Omega] - \omega \ri) r^{-1} \ri)  \vec e_\theta
\eeq 
and we have used the notation 
\[
(iv,\nabla v) : =\half i(  v\nabla \bar v - \bar v \nabla  v).
\]
The energy $\E[u]$ effectively takes into account the energy added to the giant vortex contribution $\hgpe_{\A,\om}$ when individual vortices are present in the annular bulk of the condensate. Note in particular that in the case where $v\equiv 1$ the energy is $0$. We are thus interested in estimates for $\E[u]$. A strictly negative value of this energy will indicate the presence of vortices. We note that this energy functional is very similar to a functional appearing in \cite{AAB} where an annular condensate at slow rotation speeds is considered (see also \cite{IM1,IM2}). The major difference is that the domain $\A$ depends on $\ep$ in a crucial way : its width tends to zero proportionally to $\ep |\log \ep|$ when $\ep \rightarrow 0$. We refer to \cite[Section 1.3]{CRY} for a more detailed discussion on this point.\\
The study of the energy $\E[u]$ starts with the construction of vortex balls : we isolate the set where $u$ could possibly vanish and cover it 
with a collection of disjoint balls with relatively small radii. The growth and merging method, introduced by Jerrard and Sandier independently (see \cite[Chapter 4]{SS}) yields lower bounds of the form 
\begin{equation}\label{sketch:Ekin1}
\int_{B_i} g ^2 |\nabla u | ^2 \geq \pi |d_i| g^2 (a_i) |\log \ep| 
\end{equation}    
where $B_i=B_i (a_i,r_i)$ is a generic ball in the collection and $d_i$ the degree of $u$ around the ball $B_i$. We neglect remainder terms for the sake of clarity. We stress that strictly speaking we do not cover all possible zeros of $u$ by this method : there is first a layer close to the inner boundary of $\A$ where the density $g^2$ is too small to use this method. More importantly, the shrinking width of $\A$ makes it difficult to obtain the estimates needed for the implementation of the method. There is thus a region where we have basically no information on the vortices of $u$. We will neglect this fact in this sketch since it produces only remainder terms in the energy.\\
The next step is to use a Jacobian Estimate (see \cite[Chapter 6]{SS}) to obtain
\begin{equation}\label{sketch:JE}
\mu \sim 2\pi \sum_i d_i \delta_{a_i}
\end{equation}
where $\delta_{a_i}$ is the Dirac mass at $a_i$. We then note that the function $F$ in (\ref{F}) is constructed to satisfy
\[
\nabla ^{\perp} (F) = 2 g^2 \vec{B}, \: F(R_<) = 0.
\]
Thus, integrating by parts the second term in (\ref{sketch:eomega0}) we have
\begin{equation}\label{sketch:Erot}
- \int_{\mathcal \A} 2 g^2\vec B\cdot (iv,\nabla v) = \int_{\A} F \mu - \int_{\dd \B} F (iu, \dd _{\tau } u ).
\end{equation}
The boundary term has been shown to be negligible in \cite{CRY}. At this stage we thus have essentially, gathering (\ref{sketch:Ekin1}), (\ref{sketch:Erot}) and (\ref{sketch:JE}),
\begin{equation}\label{sketch:borneinf1}
\int_{\mathcal \A} g^2|\nabla u|^2-2 g^2\vec B\cdot (iu,\nabla u) \geq \sum_i 2 \pi \left( \frac{1}{2}|d_i| g^2 (a_i) |\log \ep| + d_i F(a_i) \right)+ \int_{\A \setminus \cup_i B_i} g^2 |\nabla u| ^2. 
\end{equation} 
In \cite{CRY} the first two terms in the above right-hand side were enough to conclude. Indeed, $F$ is negative in the bulk, thus we can already read on the above formula that negative degree vortices can only increase the energy. All vortices should then be of positive degree and we can write
\begin{equation}\label{sketch:borneinf2}
\int_{\mathcal \A} g^2|\nabla u|^2-2g^2\vec B\cdot (iu,\nabla u) \geq \sum_i 2 \pi d_i H(a_i)+ \int_{\A \setminus \cup_i B_i} g^2 |\nabla u| ^2
\end{equation}
where $H$ is the (radial) cost function (\ref{fonctioncout}). Above the critical speed, $H$ is positive everywhere in the bulk, thus no vortices can be present in the condensate. By contrast, when the rotation speed is sub-critical, the cost function has a negative minimum at $R_*$ (actually, close to $R_*$, see Appendix A for the precise definition of $R_*$). Positive degree vortices are thus favorable in this regime, and they will gather close to the circle $\Cet$.\\
We have now to understand what mechanism limits the nucleation of vortices and leads to their uniform distribution on $\Cet$. This is where the third term in the right-hand side (\ref{sketch:borneinf1}) comes into play. In \cite{CRY} this term was neglected, but in the present context it is crucial to bound it from below.\\
Let us introduce the \emph{superfluid current} of $u$
\begin{equation}\label{sketch:supercour}
j := (iu,\nabla u). 
\end{equation} 
Note that by definition $\mu = \curl (j)$. The idea is now that the kinetic (first) term in (\ref{sketch:eomega0}) is essentially due to the presence of this supercurrent. The energetic contribution of the current inside vortex balls is taken into account in (\ref{sketch:Ekin1}). The next step is to estimate the energetic contribution of the part of the current that is located outside vortex balls (this is the long-range energetic cost of the vortices we were alluding to before). To this end we define a modified superfluid current 
\begin{equation}\label{sketch:jmod}
\jm := \begin{cases} j \mbox{ in } \A \setminus \cup_i B_i \\
	0 \mbox{ in } \cup_i B_i.
\end{cases}
\end{equation}
and the associated vorticity
\begin{equation}\label{sketch:mumod}
\mum := \curl (\jm).
\end{equation}
It turns out that (in a sense that we keep vague in this sketch)
\begin{equation}\label{sketch:mum}
\mum \sim \mu.
\end{equation}
Then, combining (\ref{sketch:JE}) and (\ref{sketch:mum}) we can rewrite (\ref{sketch:borneinf2}) in the form
\begin{equation}\label{sketch:borneinf3}
\int_{\mathcal \A} g^2|\nabla u|^2-2 g^2\vec B\cdot (iu,\nabla u) \geq \int_{\A} H \mum+ \int_{\A \setminus \cup_i B_i} g^2 |\nabla u| ^2.
\end{equation}
Let us now estimate the last term in the above expression. By definition of the vortex balls, outside the set $\cup_i B_i$ we have $|u|\sim 1$, thus
\[
\int_{\A \setminus \cup_i B_i} g^2 |\nabla u| ^2 \sim \int_{\A\setminus \cup_i B_i } g^2 |j | ^2 = \int_{\A} g^2 |\jm | ^2.
\] 
Recalling the equation (\ref{defihnu}) satisfied by $h_{\mum}$ we have 
\begin{equation}\label{sketch:equaj}
\curl \left( \jm + \frac{1}{g ^2} \nabla ^{\perp} h_{\mum} \right) = 0\mbox{ in } \A. 
\end{equation}
If we were working in a simply connected domain this would imply that there exists a function $f$ such that
\begin{equation}\label{sketch:jegal}
\jm = - \frac{1}{g ^2} \nabla ^{\perp} h_{\mum}+ \nabla f .
\end{equation}
The argument has to be modified because we work on an annulus, but since this does not produce significant modifications of the energy, we assume that (\ref{sketch:jegal}) holds. Then, using $h_{\mum} = 0$ on $\dd \A$, we obtain
\begin{equation}\label{sketch:Ekin2}  
\int_{\A \setminus \cup_i B_i} g^2 |\nabla u| ^2 \sim \int_{\A} g^2 |\jm | ^2 \geq \int_{\A } \frac{1}{g^2} |\nabla h_{\mum}| ^2.
\end{equation}
Gathering (\ref{sketch:borneinf3}) and (\ref{sketch:Ekin2}) we obtain, up to the (many !) remainder terms we have neglected
\begin{equation}\label{sketch:borneinf4}
\int_{\mathcal \A} g^2|\nabla u|^2-2 g^2\vec B\cdot (iu,\nabla u) \geq \int_{\A} \frac{1}{g^2} \left| \nabla h_{\mum} \right| ^2 +  \int_{\A} H \mum. 
\end{equation} 
The right-hand side of the above equation is equal to the renormalized energy of $\mum$ defined in (\ref{energyrenorm}) if this measure has its support on $\Cet$. Intuitively this is justified because the vortices of $u$ will want to be close to $\Cet$ where the cost function is minimum. Mathematically, showing that the right-hand side of (\ref{sketch:borneinf3}) is bounded below by the infimum in (\ref{energyrenorm}) uses essentially two ingredients. One is the rigorous version of the Jacobian Estimate making precise the property (\ref{sketch:JE}), the other is the information that the potential associated (via (\ref{defihnu})) to the measure minimizing (\ref{defiIstar}) is constant on the support of the measure. This is a classical fact in potential theory. In our context this means that $h_{\DirC} = h_*$ is constant on $\Cet$. We refer to Section \ref{sousec:preuve borne inf} for details on the way we use this information to complete the lower bound, and accept for the present sketch that $\mum$ has support on $\Cet$ in the limit $\ep \to 0$ so that (\ref{sketch:borneinf4}) yields the renormalized energy (\ref{energyrenorm}) of $\mum$ as a lower bound to $\E[u]$. It is then intuitively clear, if we believe not to have lost too much information on the way, that $\mum$ should `almost' minimize the renormalized energy. Thus it should be close to $-H(R_*) ( 2 I_* ) ^{-1} \DirC$ and by (\ref{sketch:mum}) so should $\mu$.\\
What remains to be proved is that `we have not lost too much information on the way', i.e. we want an upper bound to the energy matching our lower bound. We thus construct a test function displaying the optimal number of vortices evenly distributed on the circle $\Cet$. To evaluate the energy of such a test function we adapt the well-established technique used e.g in \cite{AAB,ASS,AB1,AB2,ABM,SS}. It involves in particular a Green representation of the electrostatic energy (\ref{defiInu}). In our case the computation is significantly complicated by the fact that the matter density $g^2$ vanishes at the inner boundary of $\A$ and more importantly by the particular geometry of the annulus (fixed radius and shrinking width). We have to rely on the periodicity (in the angular variable) of our test function to obtain the estimates allowing to conclude that the appropriate upper bound holds true. \\

The paper is organized as follows. We first construct our trial function and prove our upper bound in Section \ref{sec:upper bound}. Section \ref{sec:lower bound} is then devoted to the proof of the lower bound. We will collect many estimates on the way that will allow us to conclude the proof of Theorems \ref{theo:vorticity} and \ref{theo:vorticity exp} in Section \ref{sec:vorticity}. Appendix A is concerned with the analysis of the cost function and contains mainly technical computations whereas Appendix B gathers several useful results from \cite{CRY} that we recall for the convenience of the reader.

\section{Energy Upper Bound}\label{sec:upper bound}

In this section we provide the upper bound part of the energy asymptotics of Theorem \ref{theo:energy} by constructing a trial function for $\gpf$ which displays the expected vortex structure.\\
The result we are aiming at is the following :  

\begin{pro}[\textbf{Upper bound to the energy}]\label{pro:upperbound}\mbox{}\\
Let $\NN = \NN (\ep)$ be a number satisfying
\begin{equation}\label{sup vort tot}
\NN = - \frac{H(R_*)}{4 \pi I_*} (1+o(1))  
\end{equation}
in the limit $\ep \to 0$, where $I_*$ is defined in (\ref{defiIstar}) and $H$ is the cost function (\ref{fonctioncout}). We have, for $\ep$ small enough
\begin{multline}\label{bornesupenergie}
 \gpe \leq \hgpe_{\A,\om} +4 \pi ^2 (\NN) ^2  (1+o(1)) I_*  + 2\pi \NN H(R_*) + \OO \left( \frac{\NN  \log |\log \ep| }{\ep |\log \ep|}\right) 
\\=  \hgpe_{\A,\om}  - \frac{H(R_*) ^2}{4 I_*} (1+o(1)).  
\end{multline}
\end{pro}

\subsection{The trial function}

We start by dividing $\A$ into $N\propto \frac{1}{\ep |\log \ep|}$ identical cells. The cell $\A_i$ is defined as 
\begin{equation}\label{cellulei}
 \A_i = \A \cap \left\lbrace \theta_i \leq \theta \leq \theta_{i+1} \right\rbrace
\end{equation}
with 
\[
\theta_i = (i-1)\frac{2\pi}{N},\quad i=1\ldots N.
\] 
and the convention that $\theta_{N+1} = \theta_1$. We have that $\theta_{i+1} - \theta_i$ is independent from $i$ (that is the cells are identical) and proportional to $\ep |\log \ep|$ (which is the thickness of the annulus $\A$). We construct a structure of vortices periodic in the angular variable, of period $\theta_{i+1} - \theta_i = \frac{2\pi}{N}$.\\
\begin{figure}
	\begin{center}
 	\includegraphics[width=250pt,height=250pt]{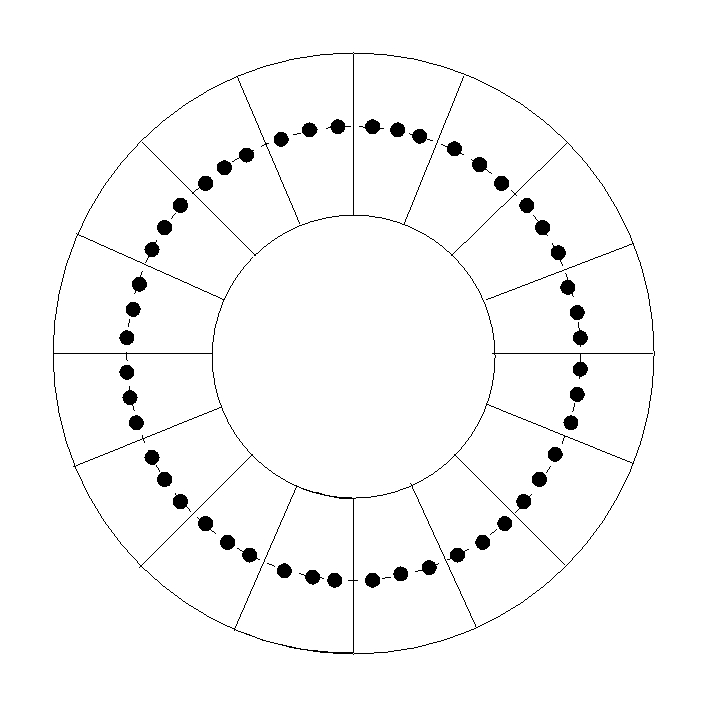}
 	% merging3.pdf 1179666x1179666 pixel, 300dpi, 9987.84x9987.84 cm, bb=0 0 542 542
	\caption{Vortex configuration in the trial function.}
	\end{center}
\end{figure}
We introduce a parameter $t$ which will be the size of the vortex cores and $M$ points $p_{i,1}, \ldots, p_{i,M}$ in each cell, which will be the locations of the vortices. We require that for any $i,j$
\begin{equation}\label{corecell}
 B(p_{i,j},t) \subset \A_i
\end{equation}
and that the collection of points $(p_{i,j})_{i=1\ldots N,\: j= 1 \ldots M}$ is evenly distributed on the circle of radius $R_*$ (which we denote $\CC_{R_*}$). The distance between two adjacent vortices is thus $2R_* \sin( \pi / MN )\sim 2\pi R_* (MN) ^{-1}$. The number of vortices will be fixed ($I_*$ is of the order of a constant, see \eqref{bornesIstar}, $H(R_*)$ is negative and proportional to $\Om_1 \ep ^{-1}$, see \eqref{borne cout TF 1} and \eqref{H-HTF}) so as to satisfy, in the limit $\ep \rightarrow 0$,
\begin{equation}\label{Nbvortexsup}
\NN := MN = - \frac{H(R_*)}{4 \pi I_*} (1+o(1)) \propto \frac{\Om_1}{\ep}. 
\end{equation}
The natural vorticity measure associated with our ring of (degree one) vortices would be a sum of Dirac masses $2\pi \sum \delta_{p_{i,j}}$, which would lead to an infinite energy. We thus introduce a regularized vorticity measure
\begin{equation}\label{vorticitetest}
f:= \sum _{i,j}  \frac{2}{t^2}  \one _{B(p_{i,j} ,t)}
\end{equation}
where $\one _{B(p_{i,j} ,t)}$ is the characteristic function of the ball $B(p_{i,j} ,t)$. The measure $f$ is normalized to have total mass $2\pi MN$ (i.e. $2\pi$ times the number of vortices). We require in our construction that this function be invariant by some reflections of the angular variable. More precisely, for any $i$, we require
\begin{eqnarray}\label{vortexsym1}
f(r,\theta) &=& f \left(r, 2 \theta_i  - \theta \right) \\
f(r,\theta) &=& f \left(r, \theta_i + \theta_{i+1} - \theta \right) \label{vortexsym2} . 
\end{eqnarray}
Clearly this conditions imply the periodicity of $f$ :
\begin{equation}\label{vorticitetestcell}
 f_{i} (r,\theta) := f_{| \A_i} (r,\theta) = f_{| \A_j} (r,\theta+\theta_j - \theta_i ).
\end{equation}
Note that the freedom that we have in the choice of the number of vortices (in the $o(1)$ term of formula \eqref{NBvortexsup}) is used to ensure that it is possible to construct a vortex configuration having the symmetries we want. A typical distribution of balls guaranteeing these symmetries is one that is even (in the angular variable) in each cell with the same distances of the two extremal balls to the radial boundaries of the cell. \\

\bigskip

Our trial function is given by
\begin{equation}\label{fonctiontest} 
\Psi = \begin{cases}
        g e ^{i\left( [\Om] - \om \right)\theta} v \mbox{ in } \A \\
	0 \mbox{ in } \B \setminus \A
       \end{cases}
\end{equation}
with 
\begin{equation}\label{fonctiontestreduite}
 v = c \xi e^{i\phi}.
\end{equation}
The real function $\xi$ is a cut-off function ensuring that 
\begin{itemize}
\item $v$ vanishes in small discs (of radius $t$) around the vortex locations, where the phase $\phi$ will have a singularity
\item $v$ vanishes continuously at the inner boundary of $\A$, so that $\Psi$ is indeed in the energy space. 
\end{itemize}
The constant $c$ is chosen so that $\int_{\B} |\Psi| ^2 = 1$. As for the phase $\phi$, we define it in Lemma \ref{lem:phase} in such a way that $v$ has degree one around each vortex in the collection.\\

Let us introduce functions $\xi_{i,j}$ , $i=1 \ldots N$, $j=1 \ldots M$ satisfying
\begin{equation}\label{cutoffi}
 \xi_{i,j} :=  \begin{cases}
           0\mbox{ in } B(p_{i,j},t)\\
	 1 \mbox{ in }  \A \setminus B(p_{i,j},2t)
         \end{cases} 
\end{equation}
and 
\[
 0 \leq \xi_{i,j} \leq 1.
\]
Obviously one has
\begin{equation}
\nabla  \xi_{i,j}  =  0 \mbox{ in } B(p_{i,j},t) \cup B(p_{i,j},2 t)^c \label{gradcutoffi} 
\end{equation}
and we can impose
\begin{equation} 
 \left| \nabla \xi_{i,j} \right| \leq \frac{C}{t} \mbox{ in } B(p_{i,j},2t) \setminus B(p_{i,j},t) \label{gradcutoffi2} .
\end{equation}
Next we define, for some radius $\tilde{R}$ larger than $R_<$,
\begin{equation}\label{cutoffBL}
\xiBL (\vec{r}) = \begin{cases}
                      1 \mbox{ if } r \geq \tilde{R} \\
		\frac{r-R_<}{\tilde{R}-R_<} \mbox{ if } r \leq \tilde{R}.		
                     \end{cases}
\end{equation}
The subscript `` BL'' stands for boundary layer and is justified because our choice of $\tilde{R}$ will satisfy
\begin{equation}\label{Rtilde}
\tilde{R}  = R_<  +  \ep ^{n} < R_h - \frac{1}{2}\ep^{8/7}
\end{equation}
for some large enough power $n$ (we will use later the fact that one can choose $n$ as large as one pleases but for the present $n=2$ is sufficient) and $\ep$ small enough. The inequality above is a consequence of (\ref{the inner radius}) and allows one to use the exponential smallness of $g^2$ (proved in \cite{CRY} and recalled in Proposition \ref{cry pro:g exponential smallness}) in the region $\rt \leq r\leq \rtt$.\\
The final cut-off is given as 
\begin{equation}\label{cutoff}
\xi = \xiBL \prod_{i,j} \xi_{i,j} .
\end{equation}
The choice of $\xiBL$ vanishing continuously at $r=R_<$ will ensure that $\Psi \in H^1 (\B)$. We now choose 
\begin{equation}\label{choixt}
 t = \ep ^{3/2} |\log \ep| ^{1/2}.
\end{equation}
This is the optimal choice and allows to estimate the $L^2$ norm of $\Psi$ : it is easily seen that for any $i$ and any $j$
\[
 \int_{B(p_{i,j},t)} g^2 = \OO (\ep ^2),
\]
because $g^2 \leq C (\ep |\log \ep|) ^{-1}$. Moreover, using the exponential smallness result \eqref{cry eq:g exp small} for $g^2$ one can see that\footnote{saying that some quantity is $\OO(\ep ^{\infty})$ means that it goes to zero faster than any power of $\ep$ when $\ep\to 0$}
\[
 \int _{\rt \leq r \leq \tilde{R}} g^2  = \OO (\ep ^{\infty}).
\]
With these two estimates in hand, using the normalization of $g^2$ we conclude that
\begin{equation}\label{defautmasse}
\int_{\A} g^2 \xi ^2 \geq 1 - C M N \ep ^2.
\end{equation}
On the other hand 
\[
\int_{\A} g^2 \xi ^2 \leq \int_{\A} g^2  = 1
\]
and thus one can normalize $\Psi$ by taking a constant $c$ satisfying
\begin{equation}\label{choixc}
 c ^2 = 1+\OO (M N \ep ^2).
\end{equation}

\bigskip

We now turn to the definition of the phase of $v$. The following modified density $\rho$, defined for $r\in \A$, will be used in the definition:
\begin{equation}\label{defirhot}
 \rho (r) = \displaystyle \begin{cases}
                       \tfm (r)\mbox{ if } r\geq \rb \\
			g ^2 (r) \mbox{ if } r < \rb
                      \end{cases}
\end{equation}
where
\[
 \rt < \rb = R_h + \ep ^{5/6}
\]
with the actual choice specified in Lemma \ref{lem:rhob} below. This function is constructed to satisfy two properties that we shall need in the sequel. First its gradient has to be bounded by $C(\ep |\log \ep |)^{-2}$ for $r\geq \rb$, which is insured by the explicit form of $\tfm$. This is the main reason why we do not use simply $g^2$: the bound that is available on the gradient of this function \cite[Proposition 2.7]{CRY} is not sufficient for our purpose. Second $\rhob$ has to stay close to $g^2$ in $L^{\infty}$ norm, which is insured by the following

\begin{lem}[\textbf{Properties of $\rhob$}]\label{lem:rhob}\mbox{}\\
Let us define 
\begin{equation}\label{choixRb-}
\rb := R_h + \ep ^{5/6}.
\end{equation}
Then for any $r\in \A$
\begin{equation}\label{estimrhob}
\left| g^2 (r) - \rhob (r)\right|  \leq C \ep ^{3/4} |\log \ep |^{2} \rhob (r).
\end{equation}
Moreover $\rho$ is bounded below by a positive constant in $\A$ (the constant depends on $\ep$).
\end{lem}

\begin{proof}
The estimate is of course trivial when $r\leq \rb$. When $r \geq \rb$ it is a consequence of Proposition 2.6 in \cite{CRY}, recalled in  Proposition \ref{cry pro:point GP dens}, and of the explicit formula for $\tfm$.
The fact that $\rhob$ is bounded below is a consequence of the corresponding result for $g^2$. Indeed, $g$ satisfies the equation
\beq
			\label{hgpm var}
			- \Delta g + \frac{([\Omega] - \omega)^2}{r^2} g  - 2 \Omega ([\Omega] - \omega) g + 2 \eps^{-2} g ^3 = \hchem g ,
		\eeq
with Neumann boundary conditions on $\dd \A$, where $\hchem$ is the Lagrange multiplier associated with the mass constraint. It is well known that such a function, ground-state of a one-dimensional Schr\"{o}dinger operator, cannot vanish except at the origin (see e.g. \cite[Theorem 11.8]{LL}). A possible proof is through the Harnack inequality \cite{Mo}.
\end{proof}

Note that $\rhob$ has a jump discontinuity. We could have constructed a regular function instead, but since the discontinuity has no consequence in the sequel, we stick to the simple definition (\ref{defirhot}). 

We now define $\barh_f$ as the unique solution of  
\begin{equation}\label{defihbarf}
	\begin{cases}
		-\nabla \left( \frac{1}{\rhob} \nabla \barh_{f}\right) = f \mbox{ in } \A \\
		\barh_{f} = 0 \mbox{ on } \partial \A.
	\end{cases} 
\end{equation}
Note that the elliptic operators appearing in (\ref{defihnu}) and (\ref{defihbarf}) are similar up to the replacement of the weight $g^{-2}$ by $\rhob$. Thanks to Lemma \ref{lem:rhob} these two operators are close in some sense.  We denote
\begin{equation}\label{correctionphase1}
 \kappa = \int_{\dd \B } \frac{1}{\rhob} \frac{\dd \barh_{f}}{\dd n} - 2\pi \left[ \frac{1}{2\pi} \int_{\dd \B } \frac{1}{\rhob} \frac{\dd \barh_{f}}{\dd n} \right]
\end{equation}
where $\frac{\dd}{\dd n}$ denotes the deritaive in the outward normal direction. Also we introduce $\Gb$ as the solution of 
\begin{equation}\label{defiGamma}
	\begin{cases}
		-\nabla \left( \frac{1}{\rhob} \nabla \Gb \right) = 0 \mbox{ in } \A \\
		\Gb  = 1 \mbox{ on } \partial \B \\
		\Gb = 0 \mbox{ on } \dd B_{\rt} .
	\end{cases} 
\end{equation}
Note that $\Gb$ is radial and has the explicit expression
\begin{equation}\label{formuleGamma}
\Gb (\vec{r}) = \frac{\int ^ r _{\rt} \rhob (s) s^{-1} ds}{\int_{\rt} ^1 \rhob (s) s^{-1} ds}.
\end{equation}
The denominator in the above equation is a $\OO (1)$ thanks to (\ref{estimrhob}) and the normalization of $g^2$.\\
We prove the following lemma, which defines the phase of our trial function. 

\begin{lem}[\textbf{Phase of the trial function}]\label{lem:phase}\mbox{}\\
The formula
\begin{equation}\label{defiphase} 
 \nabla \phi = \frac{1}{\rhob} \nabla ^{\perp} \bar{\barh}_{f} 
\end{equation}
where 
\begin{equation}\label{defihbarbar}
\bar{\barh}_f =  \barh_{f} - \frac{\kappa}{\int_{\A} \frac{1}{\rhob} \left| \nabla \Gb \right|^2} \Gb
\end{equation}
defines a phase, i.e. $e^{i\phi}$ is well-defined, in $\A \setminus \cup_{i,j} B(p_{i,j},t)$. Moreover 
\begin{equation}\label{contributionphase}
\int_{\A} \rhob \xi ^2 \left| \nabla \phi \right|^2 \leq \int_{\A} \frac{1}{\rhob} \left| \nabla \bar{\barh}_f \right|^2 \leq \int_{\A} \frac{1}{\rhob} \left| \nabla \barh_f \right|^2 + \OO (1).
\end{equation}

\end{lem}

\begin{proof}
By definition of $\barh_f$ and $\Gb$, the right-hand side of (\ref{defiphase}) is irrotational in $\A \setminus \cup_{i,j} B(p_{i,j},t)$, thus $\phi$ is well-defined locally in this set. To see that $e^{i\phi}$ is well defined, we must check that for any closed curve $\sigma$ included in $\A \setminus \cup_{i,j} B(p_{i,j},t)$
\begin{equation}\label{quantification}
 \int_{\sigma} \nabla \phi \cdot \tau \in 2\pi \Z.
\end{equation}
Clearly, it is sufficient to consider two cases :
\begin{enumerate}
\item $\sigma$ winds around at most one ball $B(p_{i,j},t)$
\item $\sigma$ is a contour enclosing all the balls
\end{enumerate}
and one deduces the general case from these two.\\
In case 1, (\ref{quantification}) is a simple consequence of the quantification of the mass of $f$. One shows easily that (\ref{quantification}) is satisfied using (\ref{defihbarf}), (\ref{defiGamma}) and integrations by parts. In case 2, equations (\ref{defihbarf}) and (\ref{defiGamma}) give
\[
\int_{\sigma} \nabla \phi \cdot \tau =  \int_{\dd \B } \left( \frac{1}{\rhob} \frac{\dd \barh_{f}}{\dd n} - \frac{\kappa}{\int_{\A} \frac{1}{\rhob} \left| \nabla \Gb \right|^2} \frac{1}{\rhob}  \frac{ \dd \Gb}{\dd n} \right). 
\]
But, using again (\ref{defiGamma})
\[
 \int_{\dd \B } \frac{1}{\rhob} \frac{\dd \Gb }{\dd n} = \int_{\dd \B } \frac{1}{\rhob} \Gb \frac{\dd \Gb }{\dd n}  = \int_{\A} \frac{1}{\rhob} \left| \nabla \Gb \right|^2
\]
and thus it is clear from (\ref{defiphase}) that 
\[
 \int_{\dd \B} \left( \frac{1}{\rhob} \frac{\dd \barh_{f}}{\dd n} - \frac{\kappa}{\int_{\A} \frac{1}{\rhob} \left| \nabla \Gb \right|^2} \frac{1}{\rhob}  \frac{ \dd \Gb}{\dd n} \right) \in 2 \pi \Z .
\]
To prove (\ref{contributionphase}) one remarks that
\[
 \int_{\A} \rhob \xi ^2 \left| \nabla \phi \right|^2 \leq \int_{\A} \frac{1}{\rhob} \left| \nabla \bar{\barh}_f \right|^2 = \int_{\A} \frac{1}{\rhob} \left| \nabla  \barh_{f} - \frac{\kappa}{\int_{\A} \frac{1}{\rhob} \left| \nabla \Gb \right|^2} \nabla  \Gb \right|^2. 
\]
Then, using the equation for $\Gb$ and the boundary condition on $\barh_f$
\[
 \int_{\A} \frac{1}{\rhob} \nabla  \barh_{f} \cdot \nabla \Gb = 0.
\]
Finally one uses that $0 \leq \kappa \leq 2\pi $ and 
\begin{equation}\label{capGamma}
 \int_{\A} \frac{1}{\rhob} \left| \nabla \Gb \right|^2 \geq C
\end{equation}
which follows easily from (\ref{formuleGamma}).

\end{proof}

We can now start the computation of the energy of our trial function :

\begin{proof}[Proof of Proposition \ref{pro:upperbound}]

A classical technique of energy decoupling (see e.g. the proof of Proposition 3.1 in \cite{CRY}) uses the variational equation satisfied by $g$ to show that 
\begin{equation}\label{decoupletest}
\gpf [\Psi] = \hgpe_{\A,\om} + \E [v]
\end{equation}
where
\beq\label{eomega0} 
\mathcal E[v] :=\int_{\mathcal \A} g^2|\nabla v|^2- 2 g^2\vec B\cdot (iv,\nabla v)+\frac {g^4}{\eps^2}(1-|v|^2)^2 
\eeq
and
\beq
	\label{vec B}
	\vec B(r) : =B(r) \vec e_\theta = \lf(  \Omega r - \lf( [\Omega] - \omega \ri) r^{-1} \ri)  \vec e_\theta.
\eeq 
We have also used the notation 
\[
(iv,\nabla v) : =\half i(  v\nabla \bar v- \bar v\nabla  v).
\]
We now have to estimate the energy $\E [v]$. We first note that it can be decomposed into the contribution of the density profile $\xi$ and the contribution of the phase, i.e. the energy generated by the vortices :
\begin{equation}\label{splitEreduite}
\E [v] = c^2 \int_{\A} \left( g^2 |\nabla \xi| ^2 +  \frac{g^4}{\ep ^2} \left( 1 - c ^2 \xi ^2 \right) ^2 \right) +  \int_{\A} \left( c^2 g^2 \xi ^2 |\nabla \phi |^2 - 2 g^2 \vec{B} \cdot (iv,\nabla v) \right).
\end{equation}
The contribution of the profile (first term in the above equation) is readily estimated by using the definition of $\xi$. We separate the boundary layer $\{ r \leq \rtt \}$ where $\xi = \xiBL$ from the bulk where $\xi = \prod_{i,j} \xi_{i,j} $ :
\begin{multline*}
 \int_{\A} \left( g^2 |\nabla \xi| ^2 +  \frac{g^4}{\ep ^2} \left( 1 - c ^2 \xi ^2 \right) ^2 \right) = \int_{\A \cap \{ r \leq \rtt \} } \left( g^2 |\nabla \xi| ^2  +  \frac{g^4}{\ep ^2} \left( 1 - c ^2 \xi ^2 \right) ^2 \right) \\ +\int_{\A \cap \{ r \geq \rtt \} }  \left( g^2 |\nabla \xi| ^2 +  \frac{g^4}{\ep ^2} \left( 1 - c ^2 \xi ^2 \right) ^2 \right).
\end{multline*}
In the boundary layer we use \eqref{cry eq:g exp small} and the definition of $\xiBL$ (\ref{cutoffBL}) to obtain
\[
 \int_{\A \cap \{ r \leq \rtt \} } \left( g^2 |\nabla \xi| ^2  +  \frac{g^4}{\ep ^2} \left( 1 - c ^2 \xi ^2 \right) ^2 \right) = \OO ( \ep^{\infty} ).
\]
For the bulk term we observe that, since $g ^2 \leq C (\ep |\log \ep|) ^{-1}$, (\ref{gradcutoffi}) and (\ref{gradcutoffi2}) imply
\[
 \int_{\A \cap \{ r \geq \rtt \} }  g^2 |\nabla \xi| ^2 \leq C \frac{MN}{\ep |\log \ep|}.
\]
We then write  
\begin{multline}\label{Eprofile}
 \int_{\A \cap \{ r \geq \rtt \} } \frac{g^4}{\ep ^2} \left( 1 - c ^2 \xi ^2 \right) ^2  = \int_{\A \cap \{ r \geq \rtt \} } \frac{g^4}{\ep ^2} \left( 1 - \xi ^2 \right) ^2  
\\ + \left( c^2 -1 \right) \left( \frac{c^2 +1}{\ep ^2 } \int_{\A \cap \{ r \geq \rtt \} } g^4 \xi ^4 - \frac{1}{\ep ^2 }\int_{\A \cap \{ r \geq \rtt \} } g^4  \xi ^2 \right)
\end{multline}
We remark that $\xi = 1$ outside of $\cup_{i,j} B(p_{i,j},2t)$, thus 
\[
 \int_{\A \cap \{ r \geq \rtt \} } \frac{g^4}{\ep ^2} \left( 1 - \xi ^2 \right) ^2 \leq C \frac{MN}{\ep |\log \ep|}
\]
using (\ref{choixt}) and $g^{2} \leq C (\ep |\log \ep| )^{-1}$. For the second term in the right-hand side of (\ref{Eprofile}) it is sufficient to recall that $g^{2} \leq C (\ep |\log \ep| )^{-1}$, $\xi \leq 1$, $|\A| \propto \ep |\log \ep|$ and use (\ref{choixc}) to obtain
\[
 \left( c^2 -1 \right) \left( \frac{c^2 +1}{\ep ^2 } \int_{\A \cap \{ r \geq \rtt \} } g^4 \xi ^4 - \frac{1}{\ep ^2 }\int_{\A \cap \{ r \geq \rtt \} } g^4  \xi ^2 \right) \leq C \frac{MN}{\ep |\log \ep|}
\]
We can thus conclude that (using again (\ref{choixc}))
\begin{equation}\label{Evortex}
\E [v] \leq \left(1+\OO(MN \ep ^2)\right) \int_{\A} \left( g^2 \xi ^2 |\nabla \phi |^2 - 2 g^2 \vec{B} \cdot (iv,\nabla v) \right) + C \frac{MN}{\ep |\log \ep|}.
\end{equation}
We thus need to compute
\begin{equation}\label{Evortexkin}
 \int_{\A} g^2 \xi ^2 |\nabla \phi |^2
\end{equation}
which is the term taking into account the kinetic energy of the vortices, and 
\begin{equation}\label{Evortexrot}
 - 2 \int_{\A} g^2 \vec{B} \cdot (iv,\nabla v) 
\end{equation}
which represents the interaction of the vortices with the rotation field. These two terms will be estimated in the  next two subsections respectively. To conclude the proof it is sufficient to gather (\ref{decoupletest}), (\ref{Evortex}), (\ref{kin:rhob}) and the results of Proposition \ref{pro:kinetic} and \ref{pro:rotation} below.

\end{proof}

\subsection{Kinetic energy of the vortices}

We first relate (\ref{Evortexkin}) to the energy of $\barh_f $: using (\ref{estimrhob}) and (\ref{contributionphase}) we have
\begin{multline}\label{kin:rhob}
\int_{\A} g^2 \xi ^2 |\nabla \phi |^2 = \left( 1 + \OO (\ep ^{3/4} |\log \ep |^{2})\right)\int_{\A} \rhob \xi ^2 |\nabla \phi |^2 
\\ \leq \left( 1 + \OO (\ep ^{3/4} |\log \ep |^{2}) \right) \int_{\A} \frac{1}{\rhob} \left| \nabla \barh_f \right|^2 + \OO (1).
\end{multline}
The required estimate is then given in the following proposition

\begin{pro}[\textbf{Kinetic energy generated by the vortices}]\label{pro:kinetic} \mbox{} \\
Let $\NN$ be as in (\ref{Nbvortexsup}) and $I_*$ defined by (\ref{defiIstar}). For $\ep$ small enough there holds
\begin{equation}\label{Evortexkinfinal}
\int_{\A} \frac{1}{\rhob} \left| \nabla \barh_f  \right|^2 \leq 4 \pi ^2 (\NN) ^2  (1+o(1)) I_*  + \pi \NN g^2 (R_*) \left| \log \ep \right| + \OO \left( \frac{\NN \log |\log \ep| }{\ep |\log \ep|}\right) .
\end{equation}
\end{pro}

Note that the third term in the right-hand side of the above equation is always much smaller than the second one because $g^2 (R_*) \propto (\ep |\log \ep|) ^{-1}$. It is also a remainder compared to the first one because of the specific choice (\ref{Nbvortexsup}) and the assumption \eqref{Omega1ter}.\\
The proof requires several constructions. The term we want to compute can be estimated using a Green function defined as follows
\begin{equation}\label{greenglobal}
\begin{cases}
-\nabla_x \left( \frac{1}{\rhob (x)} \nabla_x G(x,y) \right) = \delta_y (x) \mbox{ for } x \in \A\\
G(x,y) = 0  \mbox{ for } x\in \dd \A.
\end{cases}
\end{equation}
The existence and symmetry ($G(x,y) = G(y,x)$) of such a function is classical, because at fixed $\ep$, $\rhob$ is bounded above and below in $\A$ (see Lemma \ref{lem:rhob}). Then, using the Green representation of $\barh_f$
\begin{equation}\label{kin:energiegreen}
\int_{\A} \frac{1}{\rhob} \left| \nabla \barh_f \right|^2 = \int_{x\in \A} \int_{y\in \A} G(x,y) f(x) dx f(y) dy. 
\end{equation}
Before going to the technical implementation, let us explain the ideas behind the computation. We have
\begin{equation}\label{bornesup:formel1}
\int_{x\in \A} \int_{y\in \A} G(x,y) f(x) dx f(y) dy = \sum_{i,j,k,l} \iint_{B (p_{i,j},t)\times B (p_{k,l},t)} G(x,y) \frac{4 dx dy}{t ^4}.
\end{equation}
As is well-known, the Green function has a logarithmic singularity at $x=y$. More precisely, for $x$ close to $y$ we expect the behavior (see Lemma \ref{lem:propGi} where this is proved for slightly different Green functions) 
\[
G(x,y) \sim \frac{\rhob (y)}{2 \pi} \log \left( \frac{\ep |\log \ep|}{ |x-y|} \right) \sim \frac{g^2 (y)}{2 \pi} \log \left( \frac{\ep |\log \ep|}{ |x-y|} \right).
\] 
The second estimate is a consequence of Lemma \ref{lem:rhob}. We replace $G$ by the above expression in the diagonal terms (i.e. $i=j$ and $k=l$) of the sum (\ref{bornesup:formel1}). A direct computation yields the second term in (\ref{Evortexkinfinal}). For the off-diagonal terms we simply note that $G$ is regular when $x$ is not too close to $y$, in particular when $x\in B (p_{i,j},t)$ and $y\in B (p_{k,l},t)$ with $i\neq k$ or $j\neq l$. Then the off-diagonal sum can be seen as a Riemann sum and yields to leading order 
\begin{multline*}
\sum_{(i,j)\neq (k,l)} \iint_{B (p_{i,j},t)\times B (p_{k,l},t)} G(x,y) \frac{4 dx dy}{t ^4} \sim \sum_{(i,j)\neq (k,l)} G(p_{i,j},p_{k,l}) \\ 
\sim (4 \pi \NN ) ^2 \iint_{\A \times \A} G(x,y) \DirC(dx) \DirC (dy) =  (4 \pi \NN ) ^2 \int_{\A} \frac{1}{\rho} |\nabla \barh_{\DirC}| ^2.
\end{multline*}
There remains to note that because $\rho \sim g^2$ (Cf Lemma \ref{lem:rhob}) 
\[
\int_{\A} \frac{1}{\rho} |\nabla \barh_{\DirC}| ^2 \sim \int_{\A} \frac{1}{g^2} |\nabla h_{\DirC}| ^2 \sim I_*.
\]
For this heuristic presentation we have deviated from the actual proof procedure that we present below.\\
We note that several authors (see e.g. \cite{AAB,ASS,AB1,AB2,ABM,SS}) have already successfully used the representation (\ref{kin:energiegreen}) for the computation of similar quantities. However, in our case, the particular geometry of $\A$ (fixed radius but shrinking width) makes it difficult to obtain the properties of $G$ required in the computation. We thus prefer to take advantage of the symmetry of the vorticity measure $f$ to obtain another expression of (\ref{kin:energiegreen}). \\
Let us introduce for any $i$ the function $\barh_f ^i $ defined on the cell $\A_i$ as follows
\begin{equation}\label{defihbarfi}
 \begin{cases}
		-\nabla \left( \frac{1}{\rhob} \nabla \barh_{f} ^i \right) = f_i \mbox{ in } \A_i \\
		\barh_{f} ^i = 0 \mbox{ on } \partial \A_i \cap \partial \A \\
		\frac{1}{\rhob} \frac{\dd \barh_f ^i }{\dd n} = 0 \mbox{ on } \dd \A_i \setminus \partial \A.
 \end{cases}
\end{equation}
See (\ref{vorticitetestcell}) for the definition of $ f_i$. Note that we impose Dirichlet conditions only on the azimuthal part of the boundary of $\A_i$ (that is, the part that coincides with the boundary of $\A$), while the behavior on the radial parts of the boundary is left free, which leads to the Neumann condition in (\ref{defihbarfi}). Because of the symmetry of $f$, the following holds

\begin{lem}[\textbf{Alternate definition of $\barh_f$}]\label{lem:hfhfi}\mbox{}\\
We have 
\begin{equation}
\barh_f = \sum_i  \barh_f ^i \mbox{ on } \A
\end{equation}
i.e. for any $i\in \{ 1,\ldots N \}$
\[
 \barh_f = \barh_f ^i \mbox{ on } \A_i.
\]
As a consequence 
\begin{equation}\label{kin:divise}
\int_{\A} \frac{1}{\rhob} \left| \nabla \barh_f \right|^2 = \sum_i \int_{\A_i} \frac{1}{\rhob} \left| \nabla \barh_f ^i \right|^2.
\end{equation}
\end{lem}

\begin{proof}
Using (\ref{vortexsym2}) and the fact that $\rhob$ is radial, we obtain that $\barh_f ^i $ is symmetric with respect to a reflection of the angular variable, that is
\begin{equation}\label{symmetryhbarf1}
\barh_f ^i (r,\theta) = \barh_f ^i (r,\theta_i + \theta_{i+1} -\theta).
\end{equation}
On the other hand, using (\ref{vorticitetestcell}) we have also
\begin{equation}\label{symmetryhbarf2}
 \barh_f ^i (r,\theta) = \barh_f ^j (r,\theta + \theta_j - \theta_i).
\end{equation}
Using the last two properties, it is clear that $\sum_i  \barh_f ^i$ is well-defined, continuous on $\A$ and satisfies 
\[
 \int_{\A} \left| \nabla \left( \sum_i \barh_f ^i \right)\right| ^2 < +\infty.
\]
A simple computation shows that
\[
 -\nabla \left( \frac{1}{\rhob} \nabla \left( \sum_i \barh_f ^i \right) \right) = f 
\]
on $\A$, in the sense of distributions. The conclusion follows by uniqueness of the solution to the elliptic problem (\ref{defihbarf}).
\end{proof}

The point of using this expression of $\barh_f$ is that we are now reduced to the computation of the terms in the right-hand side of (\ref{kin:divise}). We introduce new Green functions, defined on each cell
\begin{equation}\label{greenlocal}
\begin{cases}
-\nabla_x \left( \frac{1}{\rhob (x)} \nabla_x G_i(x,y) \right) = \delta_y (x) \mbox{ for } x \in \A_i\\
G_i(x,y) = 0  \mbox{ for } x\in \dd \A_i \cap \dd \A \\
\frac{1}{\rhob} \frac{\dd G_i (x,y) }{\dd n} = 0 \mbox{ for } x\in \dd \A_i \setminus \partial \A.
\end{cases}
\end{equation}
Then, using the Green representation of $\barh_f ^i$
\begin{equation}\label{kin:enegiegreeni}
\int_{\A_i} \frac{1}{\rhob} \left| \nabla \barh_f ^i \right|^2 = \int_{x\in \A_i} \int_{y\in \A_i} G_i(x,y) f(x) dx f(y) dy. 
\end{equation}
Now we are reduced to a computation on each cell. What makes the computation easier compared to what it would have been using (\ref{kin:energiegreen}) is that each cell can be mapped by a dilation of the variables to a fixed domain. The difficulty of the shrinking width of $\A$ is thus avoided.\\ Note also that instead of computing in (\ref{kin:energiegreen}) the interaction through the Green potential $G$ of each pair of vortices in our collection, we can now simply compute the interaction through $G_i$ of each pair of vortices in $\A_i$, and sum the contributions of each cell. Thus each vortex effectively interacts only with $M \propto \Om_1 |\log \ep|$ vortices through a modified Green function instead of $NM \propto \Om_1 \ep ^{-1}$ vortices through the original Green function. \\
We now prove useful properties of the Green functions defined in (\ref{greenlocal}). It is mostly here that our choice of $\rhob$ will prove useful, in particular the fact that the gradient of this function is properly bounded above close to $\Cet$. Since $f$ has its support close to the circle of radius $R_*$ we are mainly concerned with properties of $G_i$ in that region for the computation of (\ref{kin:enegiegreeni}).  

\begin{lem}[\textbf{Properties of the Green functions}]\label{lem:propGi}\mbox{}\\
Let $G_i$ be the Green function defined in (\ref{greenlocal}). Let $C_G$ be some constant, small enough and independent of $\ep$. There holds
\begin{enumerate}
\item for any $y\in \A_i$, $G(.\: ,y) \in W^{1,p} (\A_i)$ for any $1 \leq p < 2$  
\item $G_i(x,y) \geq 0$ for any $(x,y) \in \A_i \times \A_i \setminus \left \{ x=y \right \}$
\item $G_i$ is symmetric, i.e. $G_i(x,y) = G_i (y,x)$
\item Let $y\in \A_i$ be such that $||y|-R_*| \leq C_G \ep |\log \ep|$. For any compact set $K\subset \subset \A_i$ whose diameter satisfies
\begin{equation}\label{asumK}
 \diam (K) \leq C_G \ep |\log \ep|,
\end{equation}
and any $y\in K$, there exists a constant $C_K$ such that, for any $x\in K$
\begin{equation}\label{estimGi}
 \left| G_i(x,y) + \frac{\rhob (y)}{2\pi} \log \left( \frac{|x-y|}{\ep |\log \ep|} \right)\right| \leq C_K \frac{1}{\ep |\log \ep|}.
\end{equation}
\item for any $y\in \A_i$, $G_i(x,y)$ is continuous, uniformly in $\ep$, on 
\[
 \A_i \cap \left \{ ||x|-R_*| \leq C_G \ep |\log \ep|\right \}\cap \left\{ |x-y| \geq C_G \ep |\log \ep| \right\}.
\] 

% More precisely, for any $y\in \A_i$, any $a\in \A_i$ and $\Gb$ such that $B(a,\Gb \ep |\log \ep|)\subset $
% \begin{equation}\label{oscGi}
% \mathrm{osc} \left( G_i (.\:,y), B(a, \Gb \ep |\log \ep|) \right) \leq \frac{o(\Gb)}{\ep |\log \ep|} 
% \end{equation}

\end{enumerate}

\end{lem}
 
Before giving the proof of this result we stress that we are going to use some results of \cite{St} (see also \cite{LSW}), where the main emphasis is on homogeneous Dirichlet boundary conditions. However, a careful inspection of the proofs shows that the results apply as well to our case where the boundary conditions are mixed (Neumann/Dirichlet), as long as we are only concerned with interior estimates, as we are.

\begin{proof}
The first three properties are classical, because the operator we consider is uniformly elliptic in $\A_i$ (note however that the ellipticity is not uniform with respect to $\ep$).\\
To prove items 4 and 5, we introduce a dilation of the variables
\begin{equation}\label{chgtvariable}
x= \ep |\log \ep| \xt, \quad y = \ep |\log \ep| \yt 
\end{equation}
mapping $\A_i$ to $(\ep |\log \ep|)^{-1} \A_i$, which is a domain of fixed size. We define
\begin{eqnarray}
\rhot (\xt) &:=& \ep |\log \ep| \rhob (x) \\
\Gt (\xt,\yt) &:=& \ep |\log \ep| G_i (x,y) \\
w (\xt,\yt) &:=& \Gt (\xt,\yt) + \frac{\rhot(\yt)}{2\pi} \log |\xt - \yt|.
\end{eqnarray}
We have from (\ref{greenlocal})
\begin{equation}\label{greenlocalrescale}
-\nabla_{\xt} \left( \frac{1}{\rhot (\xt)} \nabla_{\xt} \Gt (\xt,\yt) \right) = \delta_{\yt} (\xt) 
\end{equation}
and
\begin{equation}\label{equationgreenw}
-\nabla_{\xt} \left( \frac{1}{\rhot (\xt)} \nabla_{\xt} w (\xt,\yt) \right) = \frac{\rhot (\yt)\nabla_{\xt} \rhot (\xt)}{2\pi \rhot ^2 (\xt)} \nabla_{\xt} \left(\log |\xt - \yt |\right) := F(\xt,\yt)
\end{equation}
on $(\ep |\log \ep|)^{-1} \A_i$ in the sense of distributions. We now proceed as in \cite[Lemma 3.1]{ASS} :\\
Let $y\in \A_i$ be such that $||y|-R_*| \leq C_G \ep |\log \ep|$ for a small enough constant $C_G$ and $K$ be a compact set included in $\A_i$ satisfying (\ref{asumK}) and $y\in K$. We have, for any $x\in K$
\begin{equation}\label{KRstar}
||x|-R_*|\leq 2 C_G \ep |\log \ep|
\end{equation}
for a small enough constant $C_G$. Let $\Kt$ be the image of $K$ under the dilation of variables (\ref{chgtvariable}). Using Theorem 2 of \cite{Me}, there exists $p>2$, $p'<2$ and a constant $C_K$ such that (here we fix $\yt$ and consider the functions as depending only on $\xt$)
\begin{equation}\label{appliMe}
\left \Vert \nabla w \right \Vert _{L^p (\Kt)} \leq C_{\tilde{K}} \left( \left \Vert \nabla w \right \Vert_{L^{p'}(K')} + \left \Vert F \right \Vert_{W^{-1,p} (K')} \right)
\end{equation}
where $K'$ is some set such that $\Kt \subset K'$. Using a Sobolev embedding, we can take some $1<q<2$ such that
\begin{equation}\label{appliMe2}
\left \Vert F \right \Vert_{W^{-1,p} (K')} \leq C \left \Vert F \right \Vert_{L^{q} (K')}.  
\end{equation}
But, using the definition of $\rhob$ and the explicit expression of $\tfm$ we see that as long as $x$ satisfies (\ref{KRstar}) one has $\rhob (x) \geq C (\ep |\log \ep|)^{-1}$ and thus 
\[
\frac{\rhot (\yt)}{\rhot ^2 (\xt)} \leq C 
\]
because $\rhob (y) \leq C (\ep |\log \ep|)^{-1}$ for any $y \in \A$.
Also, if $C_G$ is small enough , $\rhob (x) = \tfm (x)$ for any $x \in K$, and thus
\[
 \left| \nabla \rhob (x)\right| \leq  \frac{C}{\ep ^2 |\log \ep| ^2}
\]
on $K$, which implies that on the domain we are considering
\[
\frac{\rhot (\yt)\nabla_{\xt} \rhot (\xt)}{2\pi \rhot ^2 (\xt)} \leq C. 
\]
Hence, for any $q<2$
\begin{equation}\label{borneF Green}
\left \Vert F \right \Vert_{L^{q} (K')} \leq C. 
\end{equation}
On the other hand, using Theorems 9.1 of \cite{St} (note that the value of the constant in the right-hand side of (9.6) in \cite{St} is given by Theorem 4.2 of the same paper) to estimate the $W^{1,p'}$ norm of $\Gt$ we have, for any $p'<2$
\begin{equation}\label{bornew}
\left \Vert \nabla w \right \Vert_{L^{p'}(K')} \leq \left \Vert w \right \Vert_{W^{1,p'}(K')} \leq C \left|(\ep |\log \ep|)^{-1} \A_i \right| ^{1/2 - 1/p'} \int_{(\ep |\log \ep|)^{-1} \A_i} \delta_{\yt} \leq C
\end{equation}
because $(\ep |\log \ep|)^{-1} \A_i$ has a fixed size. Plugging (\ref{appliMe2}), (\ref{borneF Green}) and (\ref{bornew}) in (\ref{appliMe}) and using a Sobolev embedding (recall that $p>2$) we obtain
\begin{equation}\label{appliMe3}
\left \Vert w \right \Vert_{L^{\infty} (K')} \leq C \left \Vert w \right \Vert_{W ^{1,p}(K')} \leq C. 
\end{equation}
There only remains to change variables to conclude that (\ref{estimGi}) holds.\\
To prove Item 5 it is sufficient to fix $\yt \in (\ep |\log \ep|)^{-1} \A_i$ and show that $\Gt(\xt,\yt)$ is continuous on 
\[
     (\ep |\log \ep|)^{-1} \A_i \cap \left \{ |\xt|-(\ep |\log \ep|)^{-1}R_*| \leq C_G \right \}\cap \left\{ |\xt-\yt| \geq C_G\right\}.
\]
This follows from Theorem 7.1 of \cite{St}. Note that the constant in the right-hand side of (7.5) in that paper does not depend on the domain, so that the continuity is indeed uniform with respect to $\ep$.
\end{proof}

We are now able to present the proof of Proposition \ref{pro:kinetic}. With the above constructions in hand there is now no real difficulty to adapt a method that has already been used several times in the literature \cite{AAB,ASS,AB1,AB2,ABM,SS}. We will thus be a little sketchy in places.

\begin{proof}[Proof of Proposition \ref{pro:kinetic}]

\emph{Step 1.} 
Clearly, from (\ref{symmetryhbarf2}) we have for any $i,j\in \{ 1, \ldots , N\}$
\[
\int_{\A_i} \frac{1}{\rhob} \left| \nabla \barh_f ^i \right|^2 = \int_{\A_j} \frac{1}{\rhob} \left| \nabla \barh_f ^j \right|^2 .
\]
It is thus sufficient to compute one of the terms in the sum (\ref{kin:divise}) and the result will follow.

% \begin{equation}\label{kin:rassemble}
% \int_{\A} \frac{1}{\rhob} \left| \nabla \barh_f  \right|^2 = N \int_{\A_1} \frac{1}{\rhob} \left| \nabla \barh_f ^1 \right|^2.
% \end{equation}
We use the Green representation of $\barh _f ^1$ for the computation :
\begin{eqnarray}\label{kin:calcul1}
\int_{\A_1} \frac{1}{\rhob} \left| \nabla \barh_f ^1 \right|^2 &=&  \iint_{(x,y)\in \A_1 \times \A_1} G_1 (x,y) f(x) dx f(y) dy  \nonumber \\
&=& \sum_{i} \iint_{(x,y)\in B(p_{1,i},t) \times B(p_{1,i},t)   } G_1 (x,y)  \frac{4}{t ^4} dx dy \nonumber \\
&+& \sum_{j\neq i, \: |p_{1,i}-p_{1,j}|< \beta \ep |\log \ep|} \iint_{(x,y)\in B(p_{1,i},t)\times B(p_{1,j},t) } G_1 (x,y)  \frac{4}{t ^4} dx dy  \nonumber \\
&+& \iint_{(x,y)\in \A_1 \times \A_1 \cap \left \{ |x-y| \geq \beta \ep |\log \ep|\right\} } G_1 (x,y) f(x) dx f(y) dy 
\end{eqnarray}
where $\beta$ is a small parameter (see below). On $B(p_{1,i},t)\times B(p_{1,i},t)$, according to (\ref{estimGi}) there is a constant (independent of $i$) so that (recall that $t=\ep ^{3/2} |\log \ep |^{1/2}$ and $|p_{1,i} | = R_*$)
\[
 \left| G_i(x,y) + \frac{\rhob (y)}{2\pi} \log \left( \frac{|x-y|}{\ep |\log \ep|} \right)\right| \leq  \frac{C_K}{\ep |\log \ep|}.
\]
Also, on $B(p_{1,i},t)$, $\rhob \equiv \tfm$ which gradient is bounded by $C(\ep |\log \ep|)^{-2}$, thus on $B(p_{1,i},t)\times B(p_{1,i},t)$
\[
 \left| G_i(x,y) + \frac{\rhob (R_*)}{2\pi} \log \left( \frac{|x-y|}{\ep |\log \ep|} \right)\right| \leq  \frac{C_K }{\ep |\log \ep|}
\]
which allows to compute
\begin{equation}\label{kin:termeI1}
\sum_{i} \iint_{(x,y)\in B(p_{1,i},t) \times B(p_{1,i},t)   } G_1 (x,y)  \frac{4}{t ^4} dx dy =  2\pi M \rhob (R_*) \log \left(\frac{\ep |\log \ep|}{t}  \right) + \OO \left( \frac{M}{\ep |\log \ep|}\right).
\end{equation}
Using (\ref{choixt}) and (\ref{estimrhob}) we thus obtain
\begin{equation}\label{kin:termeI2}
\sum_{i} \iint_{(x,y)\in B(p_{1,i},t) \times B(p_{1,i},t)   } G_1 (x,y)  \frac{4}{t ^4} dx dy = \pi M g^2 (R_*) \left| \log \ep \right|   + \OO \left( \frac{M \log |\log \ep| }{\ep |\log \ep|}\right).
\end{equation}
We recall that
\[
|p_{1,i}- p_{1,i+1}| = 2 R_* \sin( \pi / MN) \sim \frac{2\pi R_*}{MN} \gg t. 
\]
Combining this fact with the estimate (\ref{estimGi}), the upper bound $\rhob \leq C (\ep |\log \ep|)^{-1}$  and a Riemann sum yields
\begin{multline}\label{kin:termeII1}
 \sum_{j\neq i, \: |p_{1,i}-p_{1,j}|< \beta \ep |\log \ep|} \iint_{(x,y)\in B(p_{1,i},t)\times B(p_{1,j},t) } G_1 (x,y)  \frac{4}{t ^4} dx dy 
\\ \leq \frac{C}{\ep |\log \ep|} \sum_i \sum_{j\neq i, \: |p_{1,i}-p_{1,j}|< \beta \ep |\log \ep|} \left(\left| \log  \frac{\left| p_{1,i} - p_{1,j} \right|}{\ep |\log \ep|} \right| +1 \right)
\\ \leq C \frac{MN}{\ep |\log \ep|} \sum_{i=1} ^M \int_{x \in \Cet ^i } \left( \left |\log \frac{|x-p_i|}{\ep |\log \ep|} \right| +1 \right)dx 
 \end{multline}
where we denote
\[
\Cet ^i : = \left \{ x \in \Cet |\: 2\pi R_* (MN) ^{-1}  < |x-p_i| < \beta \ep |\log \ep|\right \}. 
\]
We conclude from (\ref{kin:termeII1}), using $MN \propto \Om_1 \ep ^{-1}$ that
\begin{equation}\label{kin:termeII2}
\sum_{j\neq i, \: |p_{1,i}-p_{1,j}|< \beta \ep |\log \ep|} \iint_{(x,y)\in B(p_{1,i},t)\times B(p_{1,j},t) } G_1 (x,y)  \frac{4}{t ^4} dx dy \leq C \beta |\log \beta |M ^2 N.
\end{equation}
The term on the fourth line of (\ref{kin:calcul1}) is also estimated using a Riemann sum. Using the fact that $G_1$ is continuous uniformly in $\ep$ on $\A_i \times \A_i \setminus \left \{ x=y \right \}$ we obtain
\begin{multline}\label{kin:termeIII} 
\iint_{(x,y)\in \A_1 \times \A_1 \cap \left \{ |x-y| \geq \beta \ep |\log \ep|\right\} } G_1 (x,y) f(x) dx f(y) dy \\
= 4 \pi ^2 \left(MN \right) ^2 \iint_{(x,y)\in \A_1 \times \A_1 \cap \left \{ |x-y| \geq \beta \ep |\log \ep|\right\} } G_1 (x,y) \DirC (x) \DirC  (y) (1+o(1))
\\ \leq 4 \pi ^2 \left(MN \right) ^2 \iint_{(x,y)\in \A_1 \times \A_1 } G_1 (x,y) \DirC (x) \DirC (y) (1+o(1)).
\end{multline}
Let us denote $\barh_* ^i$ the function defined in $\A_i$ satisfying
\begin{equation}\label{defihbarstari}
 \begin{cases}
		-\nabla \left( \frac{1}{\rhob} \nabla \barh_{*} ^i \right) = \DirC \mbox{ in } \A_i \\
		\barh_{*} ^i = 0 \mbox{ on } \partial \A_i \cap \partial \A \\
		\frac{1}{\rhob} \frac{\dd \barh_{*} ^i }{\dd n} = 0 \mbox{ on } \dd \A_i \setminus \partial \A
 \end{cases}
\end{equation}
and $\barh_*$ satisfying in $\A$
\begin{equation}\label{defihbarstar}
	\begin{cases}
		-\nabla \left( \frac{1}{\rhob} \nabla \barh_{*}\right) = \DirC \mbox{ in } \A \\
		\barh_{*} = 0 \mbox{ on } \partial \A.
	\end{cases} 
\end{equation}
Recall that $\DirC$ is the normalized arclength measure on the circle $\Cet$ of radius $R_*$. Up to now we have proved, letting $\beta \rightarrow 0$ in (\ref{kin:termeII2})
\[
\int_{\A} \frac{1}{\rhob} \left| \nabla \barh_f  \right|^2 \leq 4 \pi ^2(MN) ^2 (1+o(1)) \sum_i \int_{\A_i} \frac{1}{\rhob} \left| \nabla \barh_* ^i  \right|^2  + \pi MN g^2 (R_*) \left| \log \ep \right| +\OO \left( \frac{MN \log |\log \ep| }{\ep |\log \ep|}\right)
\]
and arguing as in the proof of Lemma \ref{lem:hfhfi} we deduce
\begin{equation}\label{Evortexkinfinalpresque}
\int_{\A} \frac{1}{\rhob} \left| \nabla \barh_f  \right|^2 \leq  4 \pi ^2 (MN) ^2  (1+o(1)) \int_{\A} \frac{1}{\rhob} \left| \nabla \barh_*  \right|^2  + \pi MN g^2 (R_*) \left| \log \ep \right|+\OO \left( \frac{MN \log |\log \ep| }{\ep |\log \ep|}\right) .
\end{equation}
Compared to (\ref{Evortexkinfinal}) there only remains to estimate the difference
\[
 \int_{\A} \frac{1}{\rhob} \left| \nabla \barh_*  \right|^2 - \int_{\A} \frac{1}{g^2} \left| \nabla h_*  \right|^2 =  \int_{\A} \frac{1}{\rhob} \left| \nabla \barh_*  \right|^2 - I_*
\]
where we have denoted $h_* = h_{\DirC}$ for short (see (\ref{defihnu}) for the definition of $h_{\DirC}$).\\

\emph{Step 2.} To estimate the above difference we note that 
\[
-\nabla \left( \frac{1}{\rhob} \left(\nabla \barh_{*} - \nabla h_{*} \right)\right) = \nabla\left( \left(  \frac{1}{\rhob}-\frac{1}{g^2} \right) \nabla h_*\right)  
\]
in $\A$. Multiplying by $\barh_* - h_*$, integrating over $\A$ and using Cauchy-Schwarz, we obtain
\[
 \int_{\A} \frac{1}{\rhob} \left|\nabla \barh_{*} - \nabla h_{*} \right| ^2 \leq \left( \int_{\A} \frac{1}{\rhob} \left|\nabla \barh_{*} - \nabla h_{*} \right| ^2 \right) ^{1/2} \left( \int_{\A} \rhob \left(  \frac{1}{\rhob}-\frac{1}{g^2} \right) ^2  |\nabla h_*| ^2 \right) ^{1/2}.
\]
Using the definition of $\rhob$ and Lemma \ref{lem:rhob} it is straightforward to deduce that
\[
 \left( \int_{\A} \frac{1}{\rhob} \left|\nabla \barh_{*} - \nabla h_{*} \right| ^2 \right) ^{1/2} \leq \sup_{\A} \left(\rhob g ^2 \left(  \frac{1}{\rhob}-\frac{1}{g^2} \right) ^2 \right)^{1/2} \left( \int_{\A} \frac{1}{g^2} |\nabla h_*| ^2 \right) ^{1/2} \ll \left( \int_{\A} \frac{1}{g^2} |\nabla h_*| ^2 \right) ^{1/2}.
\]
On the other hand, by Lemma \ref{lem:rhob} again
\[
 \int_{\A} \frac{1}{g^2} |\nabla h_*| ^2 = \int_{\A} \frac{1}{\rhob}|\nabla h_*| ^2 (1+o(1)).
\]
We conclude that
\begin{equation}\label{IstarIstarbar}
 \int_{\A} \frac{1}{\rhob} \left| \nabla \barh_*  \right|^2 = I_* (1+o(1)) 
\end{equation}
and thus that (\ref{Evortexkinfinal}) follows from (\ref{Evortexkinfinalpresque}).

\end{proof}

\subsection{Interaction with the rotation field}

In this subsection we estimate the interaction of the vortices with the rotation potential (\ref{Evortexrot}):

\begin{pro}[\textbf{Interaction of the vortices with the rotation potential}]\label{pro:rotation}\mbox{}\\
Let $v$ be defined in (\ref{fonctiontestreduite}) and $MN$ be as in (\ref{Nbvortexsup}). There holds 
\begin{equation}\label{Evortexrotfinal}
- \int_{\A} 2 g^2 \vec{B} \cdot (iv,\nabla v) = 2\pi \NN F(R_*) + \OO \left( \frac{\NN}{\ep ^{1/2}|\log \ep|^{1/2}}\right)+\OO \left( \frac{\NN^{1/2}}{\ep ^{11/8}}\right)
\end{equation}
where $F$ is the potential function defined in (\ref{F}).

\end{pro}
 
Note that $F(R_*)$ is negative and proportional to $\ep ^{-1}$ so that the second and third term above are really remainders.\\ 
As in the preceding subsection, the proof below will require many technical estimates. We thus find useful to give the core idea before going into the rigorous proof : Integrating by parts
\[
- \int_{\A} 2 g^2 \vec{B} \cdot (iv,\nabla v) = \int_{\A} F \curl(iv,\nabla v) - \int_{\dd \B} F (iv,\dd _{\tau} v).
\] 
Now, $F$ (see (\ref{F})) is defined in such a way that $F(1) = \OO (1)$. This has been proved in \cite{CRY} to be a consequence of (\ref{compatibility}). Since on the other hand $F(R_*) \propto \ep ^{-1}$, it is natural to neglect the boundary term above in a first approach. Then  
\[
- \int_{\A} 2 g^2 \vec{B} \cdot (iv,\nabla v) \sim \int_{\A} F \curl(iv,\nabla v) = \sum_{i,j} \int_{B(p_{i,j}, 2 t)} F \curl(iv,\nabla v)
\] 
because $|v|\sim 1$ outside $\cup_{i,j} B(p_{i,j},2 t)$, which implies $\curl (iv,\nabla v) = 0$ there. Note that we neglect the thin boundary layer where $|v|$ is small for the sake of clarity. Then we use that $t$ is very small to obtain 
\[
\int_{B(p_{i,j},2 t)} F \curl(iv,\nabla v) \sim F(p_{i,j}) \int_{B(p_{i,j},2 t)} \curl(iv,\nabla v).
\]
Finally, using the fact that $F$ is radial and by definition of the phase of $v$
\[
 F(p_{i,j}) \int_{B(p_{i,j},2 t)} \curl(iv,\nabla v) \sim 2\pi F (R_*) \deg \{v, \dd B(p_{i,j},2t) \} = 2\pi F (R_*).
\]
Collecting the above heuristics justifies (\ref{Evortexrotfinal}), recalling that there are $\NN$ balls in the collection.\\

That being said, let us go into details. We need to introduce a new potential function. Indeed, since the phase of our trial function is defined using the modified density $\rhob$ it is useful for the proof to define also a potential function accordingly. We denote
\begin{equation}\label{Fbar}
\Fb(r):= 2 \int_{\rt} ^r \: \rhob   (s) \left(\Om  s- \left([\Om] - \om \right)\frac{1}{s}\right) ds = 2 \int_{\rt} ^r \: \rhob   (s) \vec{B}(s)\cdot \vec{e}_{\theta} \: ds.
\end{equation}

The function $\Fb$ satisfies
\begin{equation}\label{propFb}
\nabla ^{\perp} \Fb  = 2 \rhob \vec{B} , \quad \Fb (\rt) = 0
\end{equation}
and
\[
 -\nabla \left( \frac{1}{\rhob} \nabla \Fb \right) = -2 \nabla \left( B \vec{e}_r \right) \mbox{ in } \A .
\]
We can thus decompose $\Fb$ as follows :
\begin{equation}\label{decomposeFb}
\Fb = \Fbin + \Fbout 
\end{equation}
where $\Fbin$ and $\Fbout$ are the solutions to the following elliptic problems :
\begin{equation}\label{defiFbin}
\begin{cases}
 -\nabla \left( \frac{1}{\rhob} \nabla \Fbin \right) = -2 \nabla \left( B \vec{e}_r \right) \mbox{ in } \A \\
\Fbin = 0 \mbox{ on } \dd \A
\end{cases}
\end{equation}
and
\begin{equation}\label{defiFbout}
\begin{cases}
 -\nabla \left( \frac{1}{\rhob} \nabla \Fbout \right) = 0 \mbox{ in } \A \\
\Fbout = \Fb (1) \mbox{ on } \dd \B \\
\Fbout = 0 \mbox{ on } \dd B_{\rt}. 
\end{cases}
\end{equation}
Both $\Fbin$ and $\Fbout$ are radial. The maximum principle implies
\begin{equation}\label{estimFbout}
 \left| \Fbout \right| \leq |\Fb (1)|
\end{equation}
and more precisely we have
\begin{equation}\label{Fbout}
\Fbout =  \Fb (1) \: \Gb
\end{equation}
where $\Gb$ is defined in (\ref{defiGamma}). We sum up some useful properties of $\Fb$ in the following Lemma:

\begin{lem}[\textbf{Properties of $\Fb$ and $\Fb_{in}$}] \mbox{} \\
Let $\Fb$, $\Fbin$ and $\Fbout$ be defined as above and $F$ in (\ref{F}). The following properties hold
\begin{eqnarray}\label{F-Fb}
\left| F(r) - \Fb (r) \right| &\leq& C \frac{|\log \ep |^2}{\ep^{1/4}} \\
\label{Fbfrontier}
|\Fb (1)| &\leq& C \frac{|\log \ep |^2}{\ep^{1/4}} \\
\label{Fbbulk}
\left| F(r) - \Fbin (r) \right| &\leq& C \frac{|\log \ep |^2}{\ep^{1/4}}\\
\label{gradientFbin}
\left| \nabla \Fbin \right| &\leq & \frac{C }{\ep ^{2}|\log \ep |}. 
\end{eqnarray}
\end{lem}

\begin{proof}
It is easy using the definitions of $F$ and $\Fb$ to deduce (\ref{F-Fb}) from (\ref{estimrhob}). As $F(1) = \OO (1)$ (see \cite[Lemma 4.1]{CRY}), (\ref{Fbfrontier}) follows. A combination of (\ref{F-Fb}), (\ref{estimFbout}) and (\ref{Fbfrontier}) yields (\ref{Fbbulk}). Using the definitions of $\Fb$ and $\Fbout$, Lemma \ref{lem:rhob}
and (\ref{Fbfrontier}) , we obtain
\[
 \left| \nabla \Fb \right| \leq  \frac{C }{\ep ^{2}|\log \ep |}
\]
and
\[
 \left| \nabla \Fbout \right| \leq  \frac{C |\log \ep |}{\ep ^{5/4}}.
\]
The estimate (\ref{gradientFbin}) then follows from (\ref{decomposeFb}).
\end{proof}

We now present the

\begin{proof}[Proof of Proposition \ref{pro:rotation}]
\emph{Step 1.} We first replace $g^2$ by $\rhob$ in (\ref{Evortexrot}): Using $|\vec{B}| \leq C \ep ^{-1}$ 
\begin{eqnarray}
\left| \int_{\A} \left( g^2 - \rhob \right)\vec{B} \cdot (iv,\nabla v) \right| &\leq& \frac{C}{\ep} \int_{\A} \left| g^2 - \rhob \right| \xi ^2 |\nabla \phi| \nonumber \\
&\leq & \frac{C}{\ep} \left( \int_{\A} \rhob \xi ^2 |\nabla \phi| ^2 \right) ^{1/2} \left( \int_{\A}  \frac{\left( g^2 - \rhob \right) ^2}{\rhob}\right)^{1/2}. 
\label{rot:g2rhob}
\end{eqnarray}
From (\ref{kin:rhob}), (\ref{Evortexkinfinal}) and (\ref{Nbvortexsup}) we have 
\[
 \int_{\A} \rhob \xi ^2 |\nabla \phi| ^2 \leq C (MN) ^2 + C \frac{MN}{\ep} \leq C \frac{MN}{\ep} .
\]
We also note that from the definition of $\rhob$
\[
 \int_{\A}  \frac{\left( g^2 - \rhob \right) ^2}{\rhob} = \int_{\A \cap \{ r\geq \rb \} }  \frac{\left( g^2 - \rhob \right) ^2}{\rhob} \leq C \ep ^{1/4} \int_{\A \cap \{ r\geq \rb \} }  \left( g^2 - \rhob \right) ^2 \leq C \ep ^{1/4}
\]
where we have used Propositions \ref{cry pro:g prelim} and \ref{cry pro:point GP dens} for the estimate of $\int_{\A \cap \{ r\geq \rb _ - \} }  \left( g^2 - \rhob \right) ^2$ and the lower bound to $\rhob$ respectively. We conclude
\begin{equation}\label{rot:g2rhob2}
\int_{\A} g^2 \vec{B} \cdot (iv,\nabla v) = \int_{\A} \rhob \vec{B} \cdot (iv,\nabla v) + \OO \left( \frac{(MN)^{1/2}}{\ep ^{11/8}}\right). 
\end{equation}
We now compute, using (\ref{propFb}), (\ref{decomposeFb}) and Stokes' formula,
\begin{equation}\label{rot:decompose}
- \int_{\A} 2 \rhob \vec{B} \cdot (iv,\nabla v) = - \int_{\A} \nabla^{\perp} (\Fb_{in}+\Fb_{out}) (iv, \nabla v)=\int_{\A} \Fb_{in} \curl (iv,\nabla v) - \int_{\A} \nabla^{\perp} (\Fb_{out}) (iv, \nabla v).
\end{equation}
The second term is estimated as follows :
\begin{equation}\label{rot:termout}
\int_{\A} \nabla^{\perp} (\Fb_{out}) (iv, \nabla v) = \int _{\A} \nabla^{\perp} (\Fb_{out}) \xi ^2 \nabla \phi = \int_{\A} \Fb(1) \frac{\xi ^2}{\rhob} \nabla \Gb  \cdot  \nabla \bar{\barh}_f .
\end{equation}
We then note that
\begin{equation}\label{rot:termout1}
\left| \Fb (1) \int _{\A}  \frac{\left(\xi ^2 -1 \right)}{\rhob} \nabla \Gb \cdot \nabla\bar{\barh}_f \right| \leq |\Fb (1)| \sup_{\A}  \frac{| \nabla \Gb  |}{\rhob ^{1/2}}  \left( \int_{\A} \left( \xi^2 - 1 \right) ^2\right) ^{1/2} \left( \int_{\A } \frac{1}{\rhob}  |\nabla \bar{\barh}_f |^2 \right) ^{1/2} .
\end{equation}
But by definition (see (\ref{formuleGamma}))
\[
 \sup_{\A}  \frac{| \nabla \Gb  |}{\rhob ^{1/2}} \leq C \sup_{\A}  \rhob ^{1/2} \leq C (\ep |\log \ep|)^{-1/2}.
\]
Also, $\xi \leq 1 $ with equality everywhere except in a boundary layer $\rt \leq r \leq \rtt$ and in $\cup_{i,j} B(p_{i,j},t)$. The former domain has an area of order $\ep^{n}$ with $n$ large (we can choose it as large as needed) and the latter an area of order $MN \ep ^{3} |\log \ep|$. Thus
\[
 \int_{\A} \left( \xi^2 - 1 \right) ^2 \leq C MN \ep ^{3} |\log \ep|.
\]
Using also (\ref{Fbfrontier}), Lemma \ref{lem:phase} and Proposition \ref{pro:kinetic} we obtain from (\ref{rot:termout1}) the estimate
\begin{equation}\label{rot:termout2}
 \left| \Fb (1) \int _{\A}  \frac{\left(\xi ^2 -1 \right)}{\rhob} \nabla \Gb \cdot \nabla\bar{\barh}_f \right| \leq C (MN)^{3/2} \ep ^{3/4} |\log \ep|^2.
\end{equation}
On the other hand 
\[
\int_{\A} \Fb(1) \frac{1}{\rhob} \nabla \Gb  \cdot \nabla \bar{\barh}_f = \Fb(1) \left( \int_{\A}  \frac{1}{\rhob} \nabla \Gb  \cdot \nabla \barh_f - \frac{\kappa}{\int_{\A} \frac{1}{\rhob} \left| \nabla \Gb \right|^2} \int_{\A} \frac{1}{\rhob} \left| \nabla \Gb \right|^2 \right) = \OO \left( \frac{|\log \ep |^2}{\ep^{1/4}} \right)
\]
because $\kappa = \OO(1)$, $\Fb (1)=\OO \left( \frac{|\log \ep |^2}{\ep^{1/4}} \right)$ and
\[
 \int_{\A}  \frac{1}{\rhob} \nabla \Gb  \cdot \nabla \barh_f = 0
\]
which follows from the definition of $\Gb$ and $\barh_f = 0$ on $\dd \A$.\\
At this stage we have, gathering equations (\ref{rot:g2rhob2}) to (\ref{rot:termout2}) and using (\ref{Nbvortexsup}) to estimate the remainders
\begin{equation}\label{rot:presque}
- \int_{\A} 2 g^2 \vec{B} \cdot (iv,\nabla v) = \int_{\A} \Fb_{in} \curl (iv,\nabla v) + \OO \left(\frac{(MN)^{1/2}}{\ep ^{11/8}}\right).
\end{equation}

\emph{Step 2.} To compute the remaining term we separate the contribution of the boundary layer from that of the bulk:
\begin{equation}\label{rot:presque2}
 \int_{\A} \Fb_{in} \curl (iv,\nabla v) = \int_{\A \cap \{ r\leq \rtt \}} \Fb_{in} \curl (iv,\nabla v) +\int_{\A \cap \{ r\geq \rtt \}} \Fb_{in} \curl (iv,\nabla v).
\end{equation}
To estimate the boundary layer contribution we remark that for any $r\leq \rtt$
\begin{equation}\label{Fbsmallboundary}
|\Fb (r)| \leq \frac{C}{\ep} (r-\rt) g^2 (r)  
\end{equation}
because $|\vec{B}| \leq C \ep ^{-1}$, $\rhob (r) = g^2 (r)$ for any $r\leq \rtt$ and $g^2$ is an increasing function.
Also 
\[
|\Fb_{out} (r)|  \leq C \frac{|\log \ep |^2}{\ep^{1/4}} (r-\rt) g^2 (r)
\]
for any $r\leq \rtt$, using the definition of $\Gb$ and the same arguments as above. We conclude that
\begin{equation}\label{Fbinsmallboundary}
|\Fb_{in} (r)|  \leq  \frac{C}{\ep} (r-\rt) g^2 (r) \leq C \ep ^{n-1} g ^2 (r)
\end{equation}
if $r\leq \rtt$ (recall that $\rtt = R_< + \ep ^n$ and that we are free to choose $n$ as large as we want). We now compute: 
\begin{multline} \label{pouet-1}
 \left| \int_{\A \cap \{ r\leq \rtt \}} \Fb_{in} \curl (iv,\nabla v) \right| \leq C \int_{\A \cap \{ r\leq \rtt \}} | \Fb_{in}| |\nabla v| ^2 \\ 
\leq C \int_{\A \cap \{ r\leq \rtt \}} | \Fb_{in}| |\nabla \xi|^2 + C \int_{\A \cap \{ r\leq \rtt \}} | \Fb_{in}| |\xi| ^2 |\nabla \phi| ^2 
\end{multline}
For the first term we use (\ref{Fbinsmallboundary}) combined with the exponential smallness of $g^2$ in the boundary layer (Equation \eqref{cry eq:g exp small}). For the second one we use (\ref{Fbinsmallboundary}), \eqref{kin:rhob} and \eqref{Evortexkinfinal} and conclude 
\begin{equation}\label{pouet1}
\left| \int_{\A \cap \{ r\leq \rtt \}} \Fb_{in} \curl (iv,\nabla v) \right| \leq C\ep^{n'}
\end{equation}
where it suffices to take $n$ arbitrarily large in \eqref{Rtilde} to obtain that $n'$ is arbitrarily large.\\
For the bulk term we use that in $\A \cap \{ r\geq \rtt \} \setminus \cup_{i,j} B(p_{i,j},2t)$, $|v|\equiv c$ which is a constant, thus $\curl(iv,\nabla v)=0$ there. This yields
\begin{equation}\label{pouet2}
\int_{\A \cap \{ r\geq \rtt \}} \Fb_{in} \curl (iv,\nabla v) = \sum_{i,j} \int _{B(p_{i,j},2t)} \Fb_{in} \curl (iv,\nabla v). 
\end{equation}
But, because of (\ref{choixt}) and (\ref{gradientFbin})  
\begin{eqnarray}
 \left|\sum_{i,j} \int _{B(p_{i,j},2t)} \left(\Fb_{in} - \Fb_{in} (p_{i,j}) \right) \curl (iv,\nabla v) \right| &\leq&  \frac{C}{\ep^{1/2} |\log \ep| ^{1/2}} \sum_{i,j} \int _{B(p_{i,j},2t)} |\nabla v| ^2 \nonumber \\
&\leq& C \ep ^{1/2} |\log \ep| ^{1/2} \int_{\A} g^2 |\nabla v| ^2 \nonumber \\
&\leq& C \frac{MN |\log \ep|^{1/2}}{\ep^{1/2}}.\label{pouet3}
\end{eqnarray}
Indeed, on $B(p_{i,j},2t)$ $g^2 \geq C (\ep |\log \ep|)^{-1}$ and $\int_{\A} g^2 |\nabla v| ^2$ has been shown in the preceding subsections to be a $\OO \left( \frac{MN}{\ep}\right)$. Finally
\begin{multline}\label{pouet4}
\Fb_{in} (p_{i,j}) \int_{B(p_{i,j},2t)} \curl (iv,\nabla v)  = 2\pi \Fb_{in} (p_{i,j}) \deg \{v, \dd B(p_{i,j},2t) \} = 2\pi \Fb_{in} (p_{i,j})  \\ = 2\pi F(R_*) + \OO\left( \frac{|\log \ep |^2}{\ep ^{1/4}}\right)
\end{multline}
by definition of the phase of $v$ and (\ref{Fbbulk}). The conclusion follows by gathering (\ref{rot:presque}), (\ref{rot:presque2}) and equations (\ref{pouet1}) to (\ref{pouet4}).
\end{proof}

\section{Energy Lower Bound}\label{sec:lower bound}

In this Section we provide the lower bound announced in Theorem \ref{theo:energy}.
The key to a lower bound matching the upper bound of Proposition \ref{pro:upperbound} is the identification, in the energy of $\gpm$, of terms representing the kinetic energy of the vortices and their interaction with the rotation potential. This terms should then be bounded from below to show that our construction in the preceding Section is optimal. Ultimately we obtain the following

\begin{pro}[\textbf{Lower bound to the energy}]\label{pro:lowerbound}\mbox{}\\
Recall the decomposition
\begin{equation}\label{rappeldecomp}
\gpm = u g e^{i\left( [\Om] - \om \right)\theta}
\end{equation}
valid in $\A$ and the definition of the reduced energy (\ref{eomega0}).
Let $I_*$ be defined in (\ref{defiIstar}) and $H$ be the cost function (\ref{fonctioncout}). For $\ep$ small enough there holds :
\begin{equation}\label{borneinf}
\gpe \geq \hgpe_{\A,\om} + \E [u] -\OO(\ep^{\infty}) \geq \hgpe_{\A,\om} - \frac{H(R_*) ^2}{4 I_*} (1+o(1)). 
\end{equation}

\end{pro}

The proof of this result will occupy the rest of the section. The main new ingredient with respect to \cite{CRY} is a lower bound to the kinetic energy located outside the vortex balls, that is, relatively far from the vortex cores. In a first subsection we recall constructions from \cite{CRY} and deduce some basic bounds that will be our starting point for the evaluation of the kinetic energy outside vortex balls. We refer to that paper for detailed explanations and comments. In a second subsection we show that the kinetic energy outside vortex balls can be bounded below using a variational problem related to electrostatics. We also present the analysis of the problem (\ref{defiIstar}) as well as some important properties of the potential $h_*$ associated to the minimizing measure $\DirC$. We conclude the proof in a final subsection.

\subsection{Preliminary Constructions}

The first inequality in (\ref{borneinf}) has actually been proved in \cite[Proposition 3.1]{CRY} (note that this result stays available with our assumptions on $\Omega$). The method is to combine the variational equation satisfied by $g$ to obtain an energy decoupling, and the exponential smallness of $\gpm$ in the complement of $\A$. Our goal is now to bound $\E [u]$ from below to obtain the second inequality.

\bigskip

To begin with, we define the energy
\begin{equation}\label{energieF}
 \F[v] : = \int_{\A} g ^{2} \left| \nabla v \right|^2+ \frac{g ^4 }{\ep ^2} \left(1-|v|^2 \right)^2 
\end{equation}
which plays a crucial role here. Indeed, a control on this energy allows to construct the vortex balls that are our main tool in this Section. The  bounds that we use as starting point are the following (see Proposition 3.1 and Lemma 4.2 of \cite{CRY}, again, this result stays valid under assumption \eqref{Omega1}, \eqref{Omega1bis} and \eqref{Omega1ter})

\begin{lem}[\textbf{First energy bounds}]\label{lem:initialbound}\mbox{}	\\
Let $u$ be defined in (\ref{u}). We have 
\begin{eqnarray}\label{Fgfirstbound}
\F [u] &\leq& C \ep ^{-2}
\\ \OO (\ep ^{\infty}) \geq \E [u] & \geq & - C\ep ^{-2} \label{Egfirstbound}.
\end{eqnarray}
\end{lem}

Our analysis uses the same covering of $\A$ by (almost rectangular) cells as in the upper bound Section, but we distinguish different types of cells : 
\begin{defi}[\textbf{Good and Bad Cells}]\mbox{}\label{defi:GC}\\
We cover $\A$ with (almost rectangular) cells of side length $C \ep |\log \ep|$, using a corresponding division of the angular variable as in Section 2.1. We note $N\propto \frac{1}{\ep |\log \ep|}$ the total number of cells and label the cells $\A_n, n\in \left\lbrace 1,...,N \right\rbrace$. Let $0\leq \alpha< \frac{1}{2}$ be a parameter to be chosen below.
\begin{itemize}
\item We say that $\A_n$ is an $\alpha$-good cell if 
\begin{equation}\label{defgc}
\int_{\A _n} g ^{2} \left| \nabla u \right|^2+ \frac{g ^4 }{\ep ^2} \left(1-|u|^2 \right)^2 \leq \frac{|\log \ep|}{\ep} \ep^{-\alpha}
\end{equation}
We will denote $N_{\alpha} ^G$ the number of $\alpha$-good cells and $GS_{\alpha}$ the (good) set they cover.
\item We say that $\A_n$ is an $\alpha$-bad cell if
\begin{equation}\label{defbc}
\int_{\A _n} g ^{2} \left| \nabla u \right|^2+ \frac{g ^4 }{\ep ^2} \left(1-|u|^2 \right)^2 > \frac{|\log \ep|}{\ep} \ep^{-\alpha}
\end{equation}
We will denote $N_{\alpha} ^B$ the number of $\alpha$-bad cells and $BS_{\alpha}$ the (bad) set they cover.
\end{itemize}
\end{defi}

Note that the annulus $\A$ has a width $\ell \propto \ep|\log \ep|$ (which implies that $N\propto (\ep |\log \ep| )^{-1}$) so that we are dividing it into cells where there is much more energy than what would be expected from the localization of the bound (\ref{Fgfirstbound}) (namely $\frac{C \ell}{\ep ^2} \propto \frac{|\log \ep|}{\ep}$) and regions of reasonably small energy. A first consequence of this is, using (\ref{Fgfirstbound}) and neglecting the good cells 
	\begin{equation}\label{numberbad}
		N^B_{\al} \leq \frac{\ep}{|\log \ep|} \ep^{\alpha} \frac{C}{\ep ^2} \leq \frac{C}{\ep |\log \ep|}\ep^{\al} \ll N, 	 
	\end{equation}
i.e. there are (relatively) very few $\alpha$-bad cells.\\

The construction of vortex balls is feasible only in the regions of sufficient density. We thus introduce a reduced annulus
\begin{equation}\label{defi annt}
 \at := \left\lbrace \rv  \: : \: \rd \leq r \leq 1 \right\rbrace
\end{equation}
with
\begin{equation}\label{defi R large}
\rd := \rtf + \ep |\log \ep|^{-1}. 
\end{equation}
An important point is that from \eqref{cry eq:pointwise bounds} and the definition \eqref{minimiseurTF} of $\tfm$ we have the lower bound
\begin{equation}\label{g low bound}
 g^2(r) \geq \frac{C}{\ep |\log \ep|^{3}} \mbox{ on } \at.
\end{equation}

We can now recall the vortex balls construction (see \cite[Proposition 4.2]{CRY} for the proof) : 

\begin{pro}[\textbf{Vortex ball construction in the good set}]\label{pro:vortexballs}
 	\mbox{}\\
	Let $0\leq \al< \frac{1}{2}$. There is a certain $\ep_0$ so that, for $\ep \leq \ep_0$ there exists a finite collection $ \{ B_i \}_{i \in I} := \left\lbrace B (\avi, r_i)\right\rbrace_{i\in I}$ of disjoint balls with centers $ \avi $ and radii $ r_i $ such that
	\begin{enumerate}
		\item $\left\lbrace \rv \in GS_{\al} \cap \at \: : \: \left| |u| - 1  \right| > |\log \ep| ^{-1}  \right\rbrace \subset \cup_{i=1} ^I B_i$,
		\item for any $\al $-good cell $\A_n$, $\sum_{i, B_i \cap \A_n \neq \varnothing } r_i = \ep |\log \ep |^{-5} $.		   
	\end{enumerate}
	Setting $d_i:= \dg \{ u, \partial B_i \} $, if $ B_i \subset \at \cap GS_{\al} $, and $d_i=0$ otherwise, we have the lower bounds
		\begin{equation}\label{lowboundballs}
		 	\int_{B_i}  \diff \rv \: g ^2 \left|\nabla u\right|^2 \geq 2\pi \left(\frac{1}{2} -\al \right) |d_i|   g^2 (a_i) \left| \log \ep \right| \left(1-C \frac{\log \left| \log \ep \right|}{\left|\log \ep\right|}\right).
		\end{equation}
\end{pro}

The second main tool that we need to introduce is the so-called Jacobian Estimate. For convenience we recall the result of \cite[Proposition 4.3]{CRY}:

\begin{pro}[\textbf{Jacobian estimate}]\label{pro:jacest}
 	\mbox{}\\
	Let $0\leq \al<\frac{1}{2}$ and $\phi$ be any piecewise-$C^1$ test function with compact support 
	\[
	 {\rm supp}(\phi) \subset \at \cap GS_{\al}.  
	\]
	Let $\left\lbrace B_i \right\rbrace_{i\in I} : = \lf\{ B(\avi,r_i) \ri\}_{i \in I} $ be a collection of disjoint balls as in Proposition \ref{pro:vortexballs}. Setting $d_i:= \dg \{ u, \partial B_i \} $, if $ \B_i \subset \at \cap GS_{\al} $, and $d_i=0$ otherwise, one has
	\begin{equation}\label{JE}
		\bigg|\sum_{i\in I}  2 \pi d_i \phi (\avi)- \int_{GS_{\al}\cap \at} \diff \rv \: \phi \:  \curl (iu,\nabla u) \bigg| \leq  C \left\Vert \nabla \phi \right\Vert_{L^{\infty}(GS_{\al})}\ep ^{2} |\log \ep|^{-2} \Fg[u] .  
	\end{equation}	
%Also 
%\begin{equation}\label{JEP}
%\bigg|\sum_{i\in I}  d_i g^2 (\avi) \phi (\avi)- \int_{GS_{\al}\cap \at} \diff \rv \: g ^2 \phi \:  \curl (iu,\nabla u) \bigg| \leq  C \left\Vert \nabla \phi \right\Vert_{L^{\infty}(GS_{\al})}\ep  |\log \ep|^{-3} \Fg[u] .  
%\end{equation}	
	
\end{pro}

%\begin{proof}
%The unweighted estimate (\ref{JE}) has already been proved in \cite{CRY}. The proof of (\ref{JEP}) follows along the same lines. The only additional ingredients are \cite[Proposition 2.6]{CRY} that we use to estimate the oscillations of $g^2$ in the vortex balls and \cite[Proposition 2.7]{CRY} that provides an upper bound to the gradient of $g^2$.
%\end{proof}

We refine our classification of cells :
	\begin{defi}[\textbf{Pleasant and unpleasant cells}]\label{defi:PC}\mbox{}	\\
		Recall the covering of the annulus $\ann$ by cells $\A_n,\: n\in \left\lbrace 1,..,N \right\rbrace$. We say that $\A_n$ is 
		\begin{itemize}
		 \item an $\al$-pleasant cell if $\A_n$ and its two neighbors are good cells. We denote $PS_{\al}$ the union of all $\al$-pleasant cells and $N_{\alpha} ^{\mathrm{P}}$ their number,
		\item an $\al$-unpleasant cell if either $\A_n$ is a bad cell, or $\A_n$ is a good cell but its two neighbors are bad cells. We denote $UPS_{\al}$ the union of all $\al$-unpleasant cells and $N_{\al}^{\mathrm{UP}}$ their number,
		\item an $\al$-average cell if $\A_n$ is a good cell but exactly one of its neighbors is not. We denote $AS_{\al}$ the union of all $\al$-average cells and $N_{\al}^{\mathrm{A}}$ their number.	
		\end{itemize}
	\end{defi}
	
Remark that one obviously has, recalling \eqref{numberbad},
\begin{equation}\label{numberunpleasant}
 N_{\al} ^{\mathrm{UP}} \leq 2 N_{\al} ^{\mathrm{B}} \ll N
\end{equation}
and
\begin{equation}\label{numberaverage}
N_{\al} ^{\mathrm{A}} \leq 2 N_{\al} ^{\mathrm{B}} \ll N.
\end{equation}
The average cells will play the role of transition layers between the pleasant set, where we will use the tools described above, and the unpleasant set, where we have little information and therefore have to rely on more basic estimates.\\
We now introduce an azimuthal partition of unity that will allow us to avoid some boundary terms when integrating by parts in the sequel.  
Let us label $UPS_{\al}^l, l\in\left\lbrace1,\ldots,L\right\rbrace$, and $PS_{\al}^m,m\in\left\lbrace1,\ldots,M\right\rbrace$, the connected components of the $\al$-unpleasant set and  $\al$-pleasant set respectively. We construct azimuthal positive functions, bounded independently of $\ep$, denoted by $\chi_l ^{\mathrm{U}}$ and $\chi_m ^{\mathrm{P}}$ (the labels U and P stand for ``pleasant set" and ``unpleasant set") so that
\begin{eqnarray}\label{partition}
	\chi_l ^{\mathrm{U}} &:=& 1 \mbox{ on } UPS_{\al}^l, \nonumber \\
	\chi_l ^{\mathrm{U}} &:=& 0 \mbox{ on } PS_{\al}^m, \mbox{ } \forall m\in\left\lbrace1,\ldots,M\right\rbrace, \mbox{ and on } UPS_{\al}^{l'}, \mbox{ } \forall  l' \neq l, \nonumber\\
	\chi_m ^{\mathrm{P}} &:=& 1 \mbox{ on } PS_{\al}^m, \nonumber\\
	\chi_m ^{\mathrm{P}} &:=& 0 \mbox{ on } UPS_{\al}^l, \mbox{ } \forall  l \in\left\lbrace1,\ldots,L\right\rbrace, \mbox{ and on } PS_{\al}^{m'}, \mbox{ } \forall m' \neq m, \nonumber \\
	\sum_m \chi_m ^{\mathrm{P}} + \sum_l \chi_l ^{\mathrm{U}} &=& 1 \mbox{ on } \ann. 
\end{eqnarray}
It is important to note that each function so defined varies from $0$ to $1$ in an average cell. A crucial consequence of this is that we can take functions satisfying
\beq\label{gradientchi}
 	\lf|\nabla \chi_l ^{\mathrm{U}} \ri| \leq \frac{C}{\ep |\log \ep |},	\hspace{1,5cm}	\lf|\nabla \chi_m ^{\mathrm{P}}\ri| \leq \frac{C}{\ep |\log \ep |},
\eeq
because the side length of a cell is $\propto \ep |\log \ep|$.
\newline
We will use the short-hand notation
\begin{eqnarray}\label{chi in}
 \chi_{\mathrm{in}} &: =& \sum_{m=1} ^M \chi_m ^{\mathrm{P}}, \\
\chi_{\mathrm{out}} &: =& \sum_{l=1} ^L \chi_l ^{\mathrm{U}} \label{chi out}.
\end{eqnarray}
The subscripts `in' and `out' refer to `in the pleasant set' and `out of the pleasant set' respectively.

\bigskip

We want to use the Jacobian estimate of Proposition \ref{pro:jacest} with $\phi = \chi_{\mathrm{in}} F$, which does not vanish on $\dd \B$ and has its support included in $\A$ which is larger than $\at$. We will need one more construction to make this possible :
We introduce two radii $R_{\mathrm{cut}}^+$ and $R_{\mathrm{cut}}^-$ as
\begin{eqnarray}\label{Rcutplus}
 R_{\mathrm{cut}}^+ &: =& 1 -\ep |\log \ep|^{-1}, \\
R_{\mathrm{cut}}^- &: =& \rd + \ep |\log \ep|^{-1}. \label{Rcutminus}
\end{eqnarray}
Let $\xi_{\mathrm{in}}(r)$ and $\xi_{\mathrm{out}}(r)$ be two positive radial functions satisfying
\begin{eqnarray}\label{radial partition}
\xi_{\mathrm{in}} (r) &: =& 1 \mbox{ for } R_{\mathrm{cut}}^- \leq r \leq R_{\mathrm{cut}}^+, \nonumber \\
\xi_{\mathrm{in}} (r) &: =& 0 \mbox{ for }  \rt \leq r \leq \rd  \mbox{ and for } r=1, \nonumber \\
\xi_{\mathrm{out}} (r) &: =& 1 \mbox{ for } \rt \leq r \leq \rd, \nonumber \\
\xi_{\mathrm{out}} (r) &: =& 0 \mbox{ for } R_{\mathrm{cut}}^- \leq r \leq R_{\mathrm{cut}}^+, \nonumber \\
\xi_{\mathrm{in}} + \xi_{\mathrm{out}} &=& 1 \mbox{ on } \ann.
\end{eqnarray}
Moreover, because of \eqref{Rcutplus} and \eqref{Rcutminus}, we can impose
\beq\label{gradientxi}
 	\lf|\nabla \xi_{\mathrm{in}} \ri| \leq \frac{C|\log \ep|}{\ep},	\hspace{1,5cm}	\lf|\nabla \xi_{\mathrm{out}} \ri| \leq \frac{C|\log \ep|}{\ep}.
\eeq
The subscripts `in' and `out' refer to `inside $\at$' and `outside of $\at$' respectively. 

\bigskip

In the sequel $\lf\{ B_i \ri\}_{i \in I} : = \left\lbrace B(\avi, r_i)\right\rbrace_{i\in I}$ is a collection of disjoint balls as in Proposition \ref{pro:vortexballs}. For the sake of simplicity we label $B_j$, $j\in J \subset I $, the balls such that $B_j \subset \at \cap GS_{\al}$. 

\bigskip

An important step towards a lower bound to $\E [u]$ is to remark that, integrating by parts,
\begin{equation}\label{ipp}
-\int_{\A} 2 g^2 \vec{B}\cdot (iu, \nabla u) = \int_{\A} F \curl (iu,\nabla u) - F(1) \int_{\dd \B} (iu,\dd _{\tau} u) 
\end{equation}
and to note that the boundary term above can be neglected. In \cite{CRY} (see equation (4.100) therein) we have proved
\[
	 \left| \int_{\partial \B} \: F(1) (iu , \partial_{\tau}u) \right| \leq C \left( |\log \ep|^{1/2}|\E [u]|^{1/2}+\frac{|\log \ep|^{1/4}}{\ep ^{1/4}} \F[u] ^{1/2}\right)
\]
and thus we deduce from Lemma \ref{lem:initialbound} that
\begin{equation}\label{ridboundary}
\left| \int_{\partial \B} \: F(1) (iu , \partial_{\tau}u) \right| \leq C \frac{|\log \ep|^{1/4}}{\ep ^{5/4}}.
\end{equation}
Note that this term is much smaller (in absolute value) than the lower bound we are aiming at.\\
Our first lower bound is the intermediate result (4.86) in \cite{CRY}, that we reorganize to obtain : 
\begin{multline}\label{lowbound1}
	\int_{\ann}  g ^{2} \left| \nabla u \right|^2 + F \: \curl (iu,\nabla u) \geq \left(1-2 \gamma \right)\int_{\A \setminus \cup_{j\in J} B_j } \xiin g^2 |\nabla u | ^2  + \sum_j \frac{1}{2} g ^2 (a_j) \chiout (a_j) |d_j| |\log \ep | (1-o(1)) 
	\\  + 2\pi \sum_{j\in J} \chiin (a_j) \left[ |d_j| \left(1-2 \gamma \right)\left(\frac{1}{2} -\al \right)  g^2 (a_j) \left| \log \ep \right| \left(1-C \frac{\log \left| \log \ep \right|}{\left|\log \ep\right|}\right) + d_j \xi_{\mathrm{in}} (a_j) F(a_j)   \right] 
	\\ + \gamma \int_{\ann} \diff \rv \: \xi_{\mathrm{out}} g^2 |\nabla u |^2 -  \int_{\ann } \diff \rv \: \xi_{\mathrm{out}}  |F| |\nabla u|^2  - C \int_{\dd \B} \diff \sigma \: |F(1)| \lf| (iu,\partial_{\tau} u) \ri| 
	\\ + \left(\gamma -\delta \right) \int_{\ann} \diff \rv \: g^2 | \nabla u | ^2 - \frac{C}{\delta \ep ^2} \int_{UPS_{\al}\cup AS_{\al}} \diff \rv \: g^2 |u|^2  - C |\log \ep|^{-1} \Fg[u].
	\end{multline}
We emphasize that we have kept the kinetic energy contained outside the vortex balls (first term in the right-hand side, on the first line) that was neglected in \cite{CRY}. The parameters in (\ref{lowbound1}) are chosen as follows : 
\begin{equation}\label{parameters}
	 	\gamma = 2\delta = \frac{\log |\log \ep|}{|\log \ep|}, \hspace{1,5cm} \al = \alt \frac{\log |\log \ep|}{|\log \ep|},
\end{equation}
where $\alt$ is a large enough constant (see below).\\
Using the estimate on the number of bad cells (\ref{numberbad}) and the upper bound on $g^2 |u| ^2 = | \gpm | ^2$ of Proposition \ref{cry pro:GP exp small} we have
\[
 \frac{1}{\delta \ep ^2} \int_{UPS_{\al}\cup AS_{\al}} g^2 |u|^2  \leq C \frac{|\log \ep|\ep ^{\al}}{\ep ^2 \log |\log \ep| }.
\]
Also, using \cite[Equation (4.22)]{CRY},
\[
\gamma \int_{\ann} g^2 |\nabla u |^2 -  \int_{\ann }  \xi_{\mathrm{out}}  |F| |\nabla u| ^2 \geq C\frac{\log |\log \ep|}{|\log \ep|} \int_{\A} g^2 |\nabla u| ^2. 
\]
Plugging these estimates and (\ref{ridboundary}) in (\ref{lowbound1}), using Lemma \ref{lem:initialbound} to estimate the last term, and taking $\alt \geq 2$ we obtain
\begin{multline}\label{lowbound2}
	\int_{\A} g ^{2} \left| \nabla u \right|^2 + F \curl (iu,\nabla u) \geq \left(1-C \frac{\log |\log \ep |}{|\log \ep |}\right)\int_{\A \setminus \cup_{j\in J}  B_j } \xiin g^2 |\nabla u | ^2 + C \sum_j g ^2 (a_j) \chiout (a_j) |d_j| |\log \ep | 
	\\  + \sum_{j\in J} 2 \pi  \chiin (a_j)  \left[ |d_j| \frac{1}{2}  g^2 (a_j) \left| \log \ep \right| \left(1-C \frac{\log \left| \log \ep \right|}{\left|\log \ep\right|}\right) + d_j \xiin (a_j) F(a_j)   \right] 
  - C \frac{|\log \ep|^{1-\alt}}{\ep ^2 \log |\log \ep| }.
\end{multline}

The critical speed $2 (3\pi \ep ^2 |\log \ep|)^{-1}$ is (roughly speaking) defined as the first speed at which the terms on the second line of (\ref{lowbound2}) all become positive. This corresponds to the speed at which the vortices cease to be energetically favorable. In \cite{CRY} we were above the critical speed, so the preceding lower bound was enough for our purpose, because the terms on the second line were positive. Here we are in the opposite situation where the vortices can become energetically favorable if they are suitably located in the annulus. Thus, a lower bound to $\E [u]$ requires an upper bound to the number of vortices (more precisely, to the sum of their degrees). We will provide this upper bound in the sequel.\\

As is standard in such problems, we need to distinguish between different types of vortex balls. First we need to distinguish the vortices lying close to the inner boundary of $\A$ from those in the bulk. The vortices close to the inner boundary of $\A$ have positive energy and can thus be neglected in the lower bound. However their energetic cost is not large enough to show that there are few vortices of this type. Also, vortices in the bulk can be energetically favorable only if their degrees are positive. On the other hand, in the range of $\Om$ that we consider even a vortex of positive degree can lower the energy only if it is close to the circle $\Cet$. \\
Thus we first divide $J$ into two subsets
\begin{eqnarray}
\Jin &=& \left\lbrace j\in J \: , \: |a_j| \geq \Rb = R_h + \ep |\log \ep| \Om_1 ^{1/2}  \right\rbrace \label{Jin}\\
\Jout &=& J\setminus \Jin \label{Jout}.
\end{eqnarray}
We next divide $\Jin$ into three subsets : 
\begin{eqnarray}
J_{-} &=& \left\lbrace j\in \Jin,\quad d_j < 0 \right\rbrace \label{Jmoins} \\
J_{+} &=& \left\lbrace j\in \Jin, \quad d_j\geq 0 \mbox{ and } \left| |a_j|- R_* \right| >  \ep |\log \ep| \Om_1 ^{1/4} \right\rbrace \label{Jplus} \\
J_* &=& \Jin \setminus J_{-}\setminus J_{+} \label{Jetoile}.
\end{eqnarray}
With these definition we can state our lower bound to the kinetic energy contained `far' from the vortex cores :

\begin{pro}[\textbf{Lower bound to the kinetic energy outside vortex balls}]\label{pro:kineticinf}\mbox{}\\
Recall the definition of $I_*$ (\ref{defiIstar}). Let $J_*$ be as above. There holds
\begin{equation}\label{infkineticnew}
\int_{\A \setminus \cup_{j\in J} B_j  } \xiin g^2 |\nabla u | ^2  \geq \left( 2\pi \sum_{j\in J_*} \chiin(a_j) d_j \right) ^2 I_* (1+\OO (\Om_1 ^{1/4})) -C \frac{\sum_{j\in J_*} \chiin(a_j) d_j}{\ep |\log \ep |^3}.
\end{equation}
\end{pro}
   
The energy evaluated in Proposition \ref{pro:kineticinf} is associated with the superfluid currents that the vortices induce in the condensate. As is well known the current generated by vortices of negative degrees could compensate that generated by vortices of positive degrees. In the sequel we prove that there are relatively few vortices of negative degrees (see (\ref{vorticite moins}) below). We will thus deduce that this phenomenon does not affect the leading order of the energy in the situation we consider. \\ 
Note also that the minimization problem appearing in our lower bound involves vorticity measures with support on $\Cet$. Such a restriction of the set of admissible measures will be proved to be favorable because most vortices have to be located close to the circle $\Cet$ where they are energetically favorable (see (\ref{vorticite plus}) below).\\ 
  
We begin the proof of Proposition \ref{pro:kineticinf} with the following lemma that gives lower bounds to the energetic cost of the vortices, depending on their degrees and locations in the annulus. The proof is postponed to Appendix A.

\begin{lem}[\textbf{Energetic cost of the different types of vortices}]\label{lem:couts}\mbox{}\\
For any $j\in \Jout$
\begin{equation}\label{coutJout}
|d_j| \frac{1}{2}  g^2 (a_j) \left| \log \ep \right| \left(1-C \frac{\log \left| \log \ep \right|}{\left|\log \ep\right|}\right) + d_j \xiin (a_j) F(a_j) \geq 0 . 
\end{equation}
For any $j\in \Jin$
\begin{equation}\label{coutJetoile}
\frac{1}{2} |d_j| g^2 (a_j) \left| \log \ep \right|  + d_j \xiin (a_j) F(a_j) \geq  |d_j| H(R_*) (1+ \OO(\Om_1)) \geq -C|d_j| \frac{\Om_1}{\ep}.
\end{equation}
Moreover, if $j\in J_{-}\cup J_{+}$
\begin{equation}\label{coutJmoinsplus}
\frac{1}{2}  |d_j| g^2 (a_j) \left| \log \ep \right| + d_j \xiin (a_j) F(a_j) \geq   |d_j| \frac{C\Om_1^{1/2}}{\ep}.
\end{equation}
\end{lem}

Recalling that
\[
 \frac{\log |\log \ep|}{|\log \ep|} \ll \Om_1 \ll 1 
\]
we can use Lemma \ref{lem:couts} to simplify further (\ref{lowbound2})
\begin{multline}\label{lowbound3}
	\int_{\A} g ^{2} \left| \nabla u \right|^2 + F \curl (iu,\nabla u) \geq \left(1-\frac{\log |\log \ep |}{|\log \ep |}\right)\int_{\A \setminus \cup_{j\in J} B_j } \xiin g^2 |\nabla u | ^2 +  C \sum_j g ^2 (a_j) \chiout (a_j) |d_j| |\log \ep |
	\\  + 2\pi H(R_*) \left( 1 + o(1)\right)\sum_{j\in J_{*}}  \chiin(a_j) d_j + \frac{C\Om_1^{1/2}}{\ep}\sum_{j\in J_{+}\cup J_-} \chiin(a_j) | d_j |
  - C \frac{|\log \ep|\ep ^{\al}}{\ep ^2 \log |\log \ep| }.
\end{multline}
Note that by definition of $\xiin$, $\xiin (a_j) = 1$ for any $j\in J_{*}$. Adding $\int_{\ann} \frac{g^4}{\ep^2}\left(1-|u|^2 \right)^2 -F(1) \int_{\partial \B} (iu , \partial_{\tau}u) $ to both sides of (\ref{lowbound3}) and using (\ref{ridboundary}) we obtain
\begin{multline}\label{lowbound4}
\E [u] \geq  \left(1-\frac{\log |\log \ep |}{|\log \ep |}\right)\int_{\A \setminus \cup_{j\in J} B_j } \xiin g^2 |\nabla u | ^2 + C \sum_j g^2 (a_j)\chiout (a_j) |d_j| |\log \ep| \\ + 2\pi H(R_*) \left( 1 + o(1) \right)\sum_{j\in J_{*}}  \chiin(a_j) d_j 
 + \frac{C\Om_1^{1/2}}{\ep}\sum_{j\in J_{+}\cup J_- } \chiin(a_j) |d_j|
  - C \frac{|\log \ep|\ep ^{\al}}{\ep ^2 \log |\log \ep| }.
\end{multline}
On the other hand, combining the upper bound to the GP energy of Proposition \ref{pro:upperbound} and the first inequality in (\ref{borneinf}) we have
\begin{equation}\label{supE}
\E [u] \leq -C \frac{\Om_1 ^2}{\ep ^2}.
\end{equation}
We deduce from the above
\begin{equation}\label{borne inf vorticite}
\sum_{j\in J_{*}}  \chiin(a_j)|d_j|  \geq C \frac{\ep}{\Om_1}\left( \frac{\Om_1 ^2 }{\ep ^2 } - \frac{|\log \ep| ^{1-\alt}}{\ep ^2 \log |\log \ep|}  \right) \geq C \frac{\Om_1 }{\ep  }.
\end{equation}
The second inequality holds true if we choose $\alt \geq 3$, which we now do. From (\ref{lowbound4}) and (\ref{supE}) we also deduce (note that the last term in (\ref{lowbound4}) is a remainder because of (\ref{borne inf vorticite}))
\begin{eqnarray}
 \sum_{j\in J_{-}} \chiin(a_j) |d_j|  &\leq& C\Om_1 ^{1/2} \sum_{j\in J_{*}}  \chiin(a_j) |d_j| \label{vorticite moins}
\\  \sum_{j\in J_{+}} \chiin(a_j) |d_j| &\leq& C\Om_1 ^{1/2} \sum_{j\in J_{*}}  \chiin(a_j) |d_j| \label{vorticite plus}
%\\ \int_{\ann \setminus \cup_{j\in J} B_j } \xi_{out} g^2 |\nabla u |^2 &\leq& C\frac{\Om_1}{\ep}  \sum_{j\in J_{*}}  \chiin(a_j) |d_j| \label{energie bord}
\end{eqnarray}
One can interpret (\ref{borne inf vorticite}), (\ref{vorticite moins}) and (\ref{vorticite plus}) as follows (recall that $\Om_1 \ll 1 $): There are at least $C \frac{\Om_1 }{\ep}$ essential vortices in $\A$ (meaning vortices with nonzero degree). Most of them are of positive degree and close to the circle $\Cet$.\\ 

\subsection{The electrostatic problem}\label{sousec:electrostatic}

A difficulty in the proof below is to define a vorticity with support on $\Cet$ starting from $u$. A possible track, following \cite{ABM}, would be to use the Jacobian estimate and a first rough upper bound to the number of vortices to obtain some compactness for the vorticity measure of $u$. Using (\ref{vorticite plus}) one would then show that the limit measure has its support on the circle of interest and obtain the lower bound by a lower semi-continuity argument. Such a strategy is difficult to adapt to our setting because the geometry of our domain strongly depends on $\ep$. To obtain non trivial limits one should rescale the annulus to work on fixed domains. Also the weights appearing in the energy would complicate the argument (recall that $g^2$ is very small close to $\dd B_{R_<}$).\\
We follow another route that does not require any compactness argument. In this subsection we prove that the left-hand side of (\ref{infkineticnew}) can be bounded below using the electrostatic energy of a suitable modification of the vorticity of $u$. The method that we use to conclude the proof of Proposition \ref{pro:kineticinf} in Subsection \ref{sousec:preuve borne inf} requires precise informations on the potential $h_*=h_{\DirC}$ generated by $\DirC$ according to (\ref{defihnu}). We provide this information in Proposition \ref{pro:electrostatic} below.\\

%Because of the shrinking width of $\A$ we would then have to rescale it cell per cell, and work at patching together the limits obtained in each cell. But the estimates on the relative fractions of the different types of vortices in the total are global. To obtain the proper limit after rescaling cell per cell one would need local estimates. This would lead to define another classification of cells and show that in most cells the estimates (\ref{vorticite moins}) and (\ref{vorticite plus}) can be localized.\\

\bigskip

Let us describe the electrostatic energy that will serve as intermediate lower bound. For technical reasons it is necessary to reduce the domain on which we work to $\ab$ where the density is large enough.\\
For any Radon measure $\nu$ supported in $\ac$ we define $\hche_{\nu}$ as the unique solution to the elliptic problem
\begin{equation}\label{defihcheck}
	\begin{cases}
		-\nabla \left( \frac{1}{g^2} \nabla \hche _{\nu}\right) = \nu \mbox{ in } \ac \\
		\hche _{\nu} = 0 \mbox{ on } \partial \ac.
	\end{cases} 
\end{equation}
and introduce
\begin{equation}\label{defiIcheck}
 \Iche (\nu) := \int_{\ac} \frac{1}{g^2} \left| \nabla \hche _{\nu} \right| ^2.
\end{equation}
We will later be interested in the minimization problem
\begin{equation}\label{defiIstarcheck}
 \Iche_* := \inf_{\nu \in \D_* ,\int \nu = 1} \Iche(\nu) = \Iche (\DirC )
\end{equation}
where the infimum is taken over the set $\D_*$ of positive Radon measures with support on the circle of radius $R_*$. The fact that $\DirC$ is the unique solution of (\ref{defiIstarcheck}) will be proven below (Proposition \ref{pro:electrostatic}).\\

The definition of the measure whose energy will be used in the lower bound goes as follows : Let $\jt$ be the modified superfluid current
\begin{equation}\label{jtilde}
\jt = \begin{cases}
							 (iu,\nabla u) \mbox{ in } \ac \cap GS_{\al} \setminus \cup_{j \in J} B_j \\
							0 \mbox{ otherwise }. 
						\end{cases}
\end{equation} 
The measure $\mut$ is the vorticity associated to $\jt$
\begin{equation}\label{mutilde}
\mut := \curl ( \jt ).
\end{equation}
The following lemma is a key ingredient in our analysis :

\begin{lem}[\textbf{Lower bound via an electrostatic energy}]\label{lem:kineticelectro}\mbox{}\\
With the above definitions, there holds for $\ep$ small enough :
\begin{equation}\label{kineticelectro}
\int _{\A \setminus \cup_{j\in J} B_j} \xiin g^2 |\nabla u| ^2 \geq (1-C |\log \ep|^{-1}) \Iche (\mut).
\end{equation}

\end{lem}

\begin{proof}[Proof of Lemma \ref{lem:kineticelectro}]

We need to define $\Gamma$ as the solution to
\begin{equation}\label{Gammasansbarre}
\begin{cases}
 -\nabla \left( \frac{1}{g^2} \nabla \Gamma \right) = 0 \mbox{ in } \ac \\
\Gamma = 0 \mbox{ on } \dd B_{\Rb} \\
\Gamma = 1 \mbox{ on } \dd \B.
\end{cases}
\end{equation}
Explicitly :
\begin{equation}\label{formuleGammasansbarre}
\Gamma (\: \vec{r} \:)=  \frac{\int_{\Rb} ^r g ^2 (s) s^{-1} ds}{\int_{\Rb} ^{1} g ^2 (s) s^{-1} ds}.
\end{equation}
Note that this function is not the same as that defined in (\ref{defiGamma}), which is not confusing because the latter will not be used again in the sequel.
We recall that in $\ac \cap GS_{\al} \setminus \cup_{j\in J} B_j$, $\xiin = 1$ and $|u|$ is close to $1$ according to Item 1 in Proposition \ref{pro:vortexballs}. Thus  
\begin{eqnarray}
\int _{\A \setminus \cup_{j\in J} B_j} \xiin g^2 |\nabla u| ^2 &\geq& \int _{\ab \cap GS_{\al} \setminus \cup_{j\in J} B_j} g^2 |\nabla u| ^2 \nonumber \\
&\geq& \left( 1 - C |\log \ep| ^{-1}\right) \int _{\ac \cap GS_{\al} \setminus \cup_{j\in J} B_j} g^2 |u|^2 |\nabla u | ^2 \nonumber \\
&\geq & \left( 1 - C |\log \ep| ^{-1}\right)\int _{\ac \cap GS_{\al} \setminus \cup_{j\in J} B_j} g^2 \left| (iu,\nabla u) \right| ^2 \nonumber \\
&=& \left( 1 - C |\log \ep| ^{-1}\right)\int _{\ac} g^2 |  \jt | ^2 \label{kineticelectro1}.  
\end{eqnarray}
We set 
\begin{equation}\label{defifmut}
f _{\mut}:= \hche_{\mut} - \frac{1}{\int_{\ac} \frac{1}{g^2}|\nabla \Gamma | ^2} \int_{\dd \B } \left( \jt \cdot \tau +\frac{1}{g ^ 2}\frac{\dd \hche_{\mut}}{\dd n}\right) \Gamma.
\end{equation}
By definition 
\begin{eqnarray}
\curl \left(\jt + \frac{1}{g^2} \nabla ^{\perp } f _{\mut}\right) &=& 0 \mbox{ in } H^{-1}(\ac) \nonumber \\
\int_{\dd \B} \left( \jt + \frac{1}{g ^2} \nabla^{\perp} f_{\mut}\right)\cdot \tau &=& 0.
\label{jthmut}
\end{eqnarray}
Hence there exists $f\in H^{1} (\ac)$ such that 
\[
\jt = - \frac{1}{g ^2} \nabla^{\perp} f_{\mut} + \nabla f.
\]
Using the fact that $f_{\mut}$ is constant on the boundary of $\ac$, we have $\int_{\ac} \nabla^{\perp} f_{\mut} \cdot \nabla f = 0$ and thus
\[
\int_{\ac} g ^2 |\jt |^2 \geq \int_{\ac} \frac{1}{g ^2 } |\nabla f_{\mut}| ^2.
\] 
Next we note that
\[
\int_{\ac} \frac{1}{g ^2 } \nabla \hche_{\mut} \cdot \nabla \Gamma = 0 
\]
because $\hche_{\mut} = 0$ on $\dd \ac = \dd B_{\Rb} \cup \dd \B$. We thus have
\[
\int_{\ac} g ^2 |\jt |^2 \geq \int_{\ac} \frac{1}{g ^2 } |\nabla f_{\mut}| ^2 \geq \int_{\ac} \frac{1}{g ^2 } |\nabla \hche_{\mut}| ^2.
\]
Combining with (\ref{kineticelectro1}) we conclude that (\ref{kineticelectro}) holds.

\end{proof}

Our next task in this subsection is to give some details on the minimization problems (\ref{defiIstar}) and (\ref{defiIstarcheck}). In particular we prove that the minimizing measure is in both cases the normalized arclength measure on $\Cet$ and we compute the associated potential explicitly. We also show that considering the problem on the reduced annulus $\ac$ does not change significantly the energy.\\
The following proposition contains probably only facts known from potential theory. Indeed, apart from the weight $g^{-2}$, the minimization problems we are considering  fall in the general context of \cite{SaTo} (see e.g. Theorem II.5.12 therein and the discussion in \cite[Example 5.1]{ABM}). We nevertheless provide a short proof for the sake of completeness. 

\begin{pro}[\textbf{The electrostatic problems}]\label{pro:electrostatic}\mbox{}
\begin{enumerate}
\item The minimization problems (\ref{defiIstar}) and (\ref{defiIstarcheck}) both admit the normalized arclength measure on $\Cet$ (denoted $\DirC$) for unique solution. Moreover there exists two positive constants $C_1$ and $C_2$ independent of $\ep$ such that
\begin{equation}\label{bornesIstar}
C_1 \leq I_* \leq C_2.
\end{equation}
\item We denote $h_* = h_{\DirC}$ and $\hche_* = \hche_{\DirC}$ the fields associated to $\DirC$ by equations (\ref{defihnu}) and (\ref{defihcheck}). Both are radial piecewise $C^1$ functions. Explicitly we have
\begin{equation}\label{formuleh*}
h_{*} (r) = h_{*} (R_*) \begin{cases} \displaystyle
															\frac{\int_{\rt} ^r g^2 (s) s^{-1}ds}{\int_{\rt} ^{R_*} g^2 (s) s^{-1}ds} \mbox{ for } \rt \leq r \leq R_* \\
															\displaystyle\frac{\int_{r} ^1 g^2 (s) s^{-1}ds}{\int_{R_*} ^1 g^2 (s) s^{-1}ds} \mbox{ for } R_* \leq r \leq 1	
														\end{cases} 
\end{equation}
and
\begin{equation}														
\hche_{*} (r) = \hche_{*} (R_*) \begin{cases}
														\displaystyle	\frac{\int_{\Rb} ^r g^2 (s) s^{-1}ds}{\int_{\Rb} ^{R_*} g^2 (s) s^{-1}ds} \mbox{ for } \Rb \leq r \leq R_*\\
														\displaystyle	\frac{\int_{r} ^{1} g^2 (s) s^{-1}ds}{\int_{R_*} ^{1} g^2 (s) s^{-1}ds} \mbox{ for } R_* \leq r \leq 1
															\end{cases}
															\label{formulehche*}															
\end{equation}
where 
\begin{eqnarray}
h_{*} (R_*) &=& I_* = \frac{\left( \int_{\rt} ^{R_*} g^2 (s) s^{-1} ds\right) \left( \int_{R_*} ^1 g^2 (s) s^{-1} ds  \right)}{ 2\pi  \int_{\rt} ^1 g^2 (s) s^{-1} ds } \label{h*R*} \\
\hche_{*} (R_*) &=& \Iche_* = \frac{ \left( \int_{\Rb} ^{R_*} g^2 (s) s^{-1} ds \right) \left( \int_{R_*} ^{1} g^2 (s) s^{-1} ds \right) }{2\pi \int_{\Rb} ^{1} g^2 (s) s^{-1} ds} \label{hche*R*}.
\end{eqnarray}
\item The following estimate holds : 
\begin{equation}\label{I-Iche}
\left| I_* - \Iche_* \right| \leq C \Om_1 ^{1/2}.
\end{equation}
\end{enumerate}

\end{pro} 

\begin{proof}

The results for the problem (\ref{defiIstarcheck}) are exactly similar to those for (\ref{defiIstar}), we thus only prove the later.\\
Let us start with the general problem of minimizing $I(\nu)$ over the set $\D_E$ of Radon measures supported on $E$, a compact subset of $\A$ :
\begin{equation}\label{defiIE}
I_E = \inf_{\nu \in \D_E, \int \nu = 1} I(\nu). 
\end{equation}
 The existence of a minimizer $\nu_E$ to such a problem is classical (see e.g. \cite{SaTo}). We denote $h_E$ the function associated to $\nu_E$ by formula (\ref{defihnu}). Next, computing the first variation of $I$, we observe that there holds, for any $\nu \in \D_E$ 
\begin{equation}\label{ELelectro}
\int_{E} \nu h_E = \lambda_E \int_E \nu
\end{equation}
where $\lambda_E$ is a Lagrange multiplier. We deduce that $h_E$ is constant on $E$. Using the above equation and (\ref{defihnu}) with $\nu = \nu_E$ we see that
\beq \label{lambdaE}
h_{|E} = \lambda_E = \int_{\A} \frac{1}{g^2} |\nabla h_E| ^2 = I_E.
\eeq
We turn to the proof of Item 1. Let $h_*$ be a solution to the minimization problem. Using (\ref{lambdaE}) and the definition of $h_*$ we have
\[
\begin{cases}
-\nabla ( \frac{1}{g^2} \nabla h_* ) = 0 \mbox{ in } \A \cap B_{R_*} \\
h_* = 0 \mbox{ on } \dd B_{\rt} \\
h_* = I_* \mbox{ on } \dd B_{R_*}.
\end{cases}
\]
This implies that $h_*$ is radial on $\A \cap B_{R_*}$. A similar argument yields that $h_*$ must be radial also on $\A \setminus B_{R_*}$ and thus radial in $\A$. Then the associated minimizing measure $\nu_{\Cet}$ is radial also. We conclude that $\DirC$ is the unique solution to the problem (\ref{defiIstar}). The bounds (\ref{bornesIstar}) are proved by noting that, using (\ref{defihnu})
\[
I_* = \int_{\A} \frac{1}{g^2} |\nabla h_*| ^2 =  \sup_{\phi \in C_c ^1 (\A)} \frac{\left| \int _{\A} \frac{1}{g ^2 }\nabla h_* \cdot \nabla \phi \right| ^2}{\int_{\A} \frac{1}{g^2} |\nabla \phi | ^2} = \frac{1}{4\pi^2 R_* ^2}  \sup_{\phi \in C_c ^1 (\A)} \frac{\left| \int _{\Cet} \phi \right| ^2}{\int_{\A} \frac{1}{g^2} |\nabla \phi | ^2}.
\] 
The upper bound follows because $g^2 \leq C (\ep |\log \ep|)^{-1}$ on $\A$ whose thickness is of order $\ep |\log \ep|$. A small computation shows that the above supremum is bounded. The lower bound is proved via a trial function for the maximization problem. For example one can use the test function used in Remark \ref{rem:com vortic}, Item 2.\\
Proving that the function defined in the right-hand side of (\ref{formuleh*}) and (\ref{h*R*}) is a solution to (\ref{defihnu}) with $\nu = \DirC$ is a straightforward computation. By uniqueness we deduce that it must be equal to $h_*$.\\
Finally (\ref{I-Iche}) follows from \eqref{bornesIstar} and the formulas (\ref{h*R*}) and (\ref{hche*R*}) once one has noted that
\[
\int _{\rt} ^{\Rb} g^2(s) s^{-1}ds \leq C \frac{|\Rb - \rt|}{\ep |\log \ep|} \leq C \Om_1 ^{1/2}.
\]
\end{proof}

\subsection{Completion of the proofs of Propositions \ref{pro:lowerbound} and \ref{pro:kineticinf}}\label{sousec:preuve borne inf}

In what follows we denote
\begin{equation}\label{degtot}
D := 2\pi \sum_{j\in \Jin } \chiin (a_j) d_j 
\end{equation}
and 
\begin{equation}\label{mustar}
\musta := D \DirC.
\end{equation}
We also decompose $\mut$ as 
\begin{equation}\label{decompmut}
\mut = \musta + \muc = D \DirC + \muc.
\end{equation}
It is useful to recall that from (\ref{borne inf vorticite}) (\ref{vorticite moins}) and (\ref{vorticite plus}) we have
\begin{equation}\label{propD}
D = 2 \pi \left( 1- C \Om_1 ^{1/2} \right)\sum_{j\in J_*} \chiin(a_j) d_j \geq C\frac{\Om_1}{\ep}.
\end{equation}

We are now ready to finish the

\begin{proof}[Proof of Proposition \ref{pro:kineticinf}]

Starting from Lemma \ref{lem:kineticelectro} there remains to evaluate $\Iche (\mut)$. The two main ingredients will be the Jacobian Estimate and the properties of $\hche_*$ proved in Proposition \ref{pro:electrostatic}.\\
We have
\begin{equation}\label{calculImut}
\Iche (\mut) = \int_{\ac} \frac{1}{g^2} |\nabla \hche_{\mut}| ^2 = \int_{\ac} \frac{1}{g^2} |\nabla \hche_{\musta}| ^2 + 2 \int_{\ac} \frac{1}{g^2} \nabla \hche_{\muc} \cdot \nabla \hche_{\musta}  + \int_{\ac} \frac{1}{g^2} |\nabla \hche_{\muc}| ^2.  
\end{equation}
Now,
\begin{equation}\label{calculImut1}
\int_{\ac} \frac{1}{g^2} |\nabla \hche_{\musta}| ^2  = D^2 \Iche (\DirC) = D^2 I_* (1 +\OO (\Om_1 ^{1/2}))
\end{equation}
by definition and use of (\ref{I-Iche}). To obtain an appropriate lower bound it is thus sufficient to compute the second term in the right-hand side of (\ref{calculImut}). Using (\ref{defihcheck}) and Proposition \ref{pro:electrostatic} we have
\begin{equation}\label{astuce}
\int_{\ac} \frac{1}{g^2} \nabla \hche_{\muc} \cdot \nabla \hche_{\musta} = \int_{\ac} \hche_{\musta} \muc = D \left( \int_{\ac} \hche_{*} \mut - D \int_{\ac}\hche_{*}  \DirC \right) = D \left( \int_{\ac} \hche_{*} \mut - D \hche_* (R_*) \right).
\end{equation}
We now compute from the definition (\ref{jtilde}) of $\jt$
\begin{eqnarray}\label{hstarmut}
\int_{\ac} \hche_{*} \mut  &= &-\int_{\ac} \jt . \nabla ^{\perp} \hche_* = -\int_{GS_{\al}} (iu,\nabla u) \cdot \nabla ^{\perp} \left(\chiin \hche_* \right)-\int_{GS_{\al}} (iu,\nabla u) \cdot \nabla ^{\perp} \left( \chiout \hche_* \right) \nonumber
\\ &+& \sum_{j\in J} \int_{B_j \cap \ab } (iu,\nabla u) \cdot \nabla ^{\perp} \hche_* \nonumber
\\ &=& \int_{GS_{\al}} \mu \chiin \hche_* -\int_{GS_{\al}} (iu,\nabla u) \cdot \nabla ^{\perp} \left(\chiout \hche_* \right) + \sum_{j\in J} \int_{B_j \cap \ac} (iu,\nabla u) \cdot \nabla ^{\perp} \hche_* .
\end{eqnarray}
Let us first show how to estimate the second term. The integral is actually located on 
\[
\ac \cap \left( AS_{\al} \cup \left( UPS_{\al} \cap GS_{\al}\right)\right)
\]
 (see the definitions of $\chiout$ and $\hche_*$). Recalling that $g^2 \geq C\ep ^{-1} |\log \ep| ^{-3}$ and $|\nabla \chiout | \leq C \left(\ep |\log \ep|\right) ^{-1}$ there we have
\begin{eqnarray}\label{contribution out}
\left| \int_{GS_{\al}} (iu,\nabla u) \cdot \nabla ^{\perp} \left(\chiout \hche_* \right) \right| &\leq& C\left( \frac{\Vert \hche_* \Vert_{L^{\infty}}}{\ep |\log \ep|}+ \Vert \nabla \hche_* \Vert_{L^{\infty}} \right) \int_{\ac \cap \left( AS_{\al} \cup \left( UPS_{\al} \cap GS_{\al}\right)\right)} |u||\nabla u|\nonumber
 \\ &\leq&    C |\log \ep | ^2  \int_{ AS_{\al} \cup \left( UPS_{\al} \cap GS_{\al}\right)} g^2 |u||\nabla u| \nonumber 
\\ &\leq& C |\log \ep| ^2   \eta \int_{AS_{\al} \cup \left( UPS_{\al} \cap GS_{\al}\right)} g ^2 |u| ^2 \nonumber
\\&+& \frac{C |\log \ep | ^2}{\eta} \int_{AS_{\al}\cup \left( UPS_{\al} \cap GS_{\al}\right)} g ^2 |\nabla u| ^2 
\end{eqnarray}
where $\eta$ is a parameter that we fix below and we have used the explicit formula (\ref{formulehche*}) for the bounds on $\hche_*$ and $|\nabla \hche_*|$. 
We recall the basic estimate
\[
 \int_{AS_{ \al}\cup \left( UPS_{\al} \cap GS_{\al}\right)} g^2 |\nabla u| ^2 \leq \frac{C}{\ep ^2}
\]
coming from (\ref{Fgfirstbound}). Also, using $g ^2 |u| ^2 \leq C (\ep |\log \ep|) ^{-1}$
\[
 \int_{AS_{ \al} \cup \left( UPS_{\al} \cap GS_{\al} \right)} g^2 |u| ^2 \leq C\ep ^{\al},
\]
as a consequence of (\ref{numberaverage}) and \eqref{numberunpleasant}. Choosing 
\[
 \eta = \ep ^{-1-\al / 2}
\]
and using (\ref{parameters}) we obtain 
\begin{equation}\label{contribution out 2}
 \left| \int_{GS_{\al}} (iu,\nabla u) \cdot \nabla ^{\perp} \left(\chiout \hche_* \right) \right|  \leq C \frac{|\log \ep|^2}{\ep |\log \ep|^{\alt / 2}}.
\end{equation}
For the third term in (\ref{hstarmut}) we use exactly the same kind of argument. We now use the smallness of the set covered by the vortex balls instead of the smallness of $AS_{\al}$. Indeed, using the division of $\A$ into $N \propto (\ep |\log \ep|)^{-1}$ cells and Item 2 of Proposition \ref{pro:vortexballs}
\begin{equation}\label{aire boules}
\left | \cup_j B_j \right | \leq C \sum_{n=1} ^N \left | \cup_{j\in J} B_j \cap \A_n \right | \leq  C \sum_{n=1} ^N \sum_{j, B_j \cap \A_n \neq \varnothing} r_j ^2  \leq C \ep |\log \ep|^{-11}.
\end{equation}
It follows that (recall that $g^2 |u|^2 \leq C (\ep |\log \ep|)^{-1}$)
\[
\int_{\cup_j B_j} g^2 |u|^2 \leq C |\log \ep|^{-12}.
\]
Using this fact and arguing as in (\ref{contribution out}) above we obtain
\begin{equation}\label{contribution boules}
 \left| \sum_{j\in J} \int_{B_j} (iu,\nabla u) \cdot \nabla ^{\perp} \hche_* \right|  \leq C \frac{1}{\ep |\log \ep|^{4}}.
\end{equation}
The first term in (\ref{hstarmut}) is estimated using the Jacobian Estimate : The function $\chiin \hche_*$ satisfies the assumptions of Proposition \ref{pro:jacest}. In particular, its support is included in $\ac \cap GS_{\al}\subset \at \cap GS_{\al} $. We thus have, using again (\ref{formulehche*}) to bound the gradient of $\hche_*$ and (\ref{Fgfirstbound}),
\begin{equation}\label{contribution in}
\int_{GS_{\al}} \mu \chiin \hche_*  = 2\pi\sum_{j\in \Jin} \chiin(a_j) d_j \hche_* (a_j) + \OO\left( \frac{1}{\ep |\log \ep|^3}\right). 
\end{equation}
Provided $\alt$ is large enough (which we are free to decide) we thus finally have, inserting (\ref{contribution out 2}), (\ref{contribution boules}) and (\ref{contribution in}) in (\ref{hstarmut})
\begin{equation}\label{hstarmut2} 
\int_{\ac} \hche_{*} \mut = 2\pi \sum_{j\in \Jin} \chiin(a_j) d_j \hche_* (a_j) + \OO\left( \frac{1}{\ep |\log \ep|^3}\right).
\end{equation}
Then, going back to (\ref{astuce}) 
%and recalling that $\xiin = 1$ on the support of $\hche_*$ (i.e. $\ac$)
\begin{equation}\label{astucefinie}
\int_{\ac} \frac{1}{g^2} \nabla \hche_{\muc} \cdot \nabla \hche_{\musta}= D \left( 2\pi \sum_{j\in \Jin} \chiin(a_j) d_j \left( \hche_* (a_j) - \hche_* (R_*)\right)\right)+ \OO \left( \frac{1}{\ep |\log \ep|^3}\right)
\end{equation} 
On the other hand by definition of $J_*$ and the explicit formula for $\hche_*$
\[
\left| \hche_* (a_j) - \hche_* (R_*)\right| \leq C \Om_1 ^{1/4} \mbox{ for any } j\in J_*.
\]
Recalling that $\hche_*(R_*) = \max _{\ac} \hche_*$ we have for any $j\in J_-$
\[
 d_j \left( \hche_* (a_j) - \hche_* (R_*)\right) \geq 0.
\]
Also, for any $j\in J_+$,
\[
d_j \left( \hche_* (a_j) - \hche_* (R_*)\right) \geq -d_j \hche_* (R_*) \geq -C d_j.
\]
It follows that %(note that $\xiin(a_j) = 1$ for any $j\in J_*$)
\begin{eqnarray}\label{astucefinie2}
\int_{\ac} \frac{1}{g^2} \nabla \hche_{\muc} \cdot \nabla \hche_{\musta} &\geq& -C D \left( \Om_1^{1/4} \sum_{j\in J_*} \chiin(a_j) d_j + \sum_{j\in J_+} \chiin(a_j) d_j \right) + \OO (\ep ^{-1} |\log \ep| ^{-3}) \nonumber \\
&\geq& - C D \left( D \Om_1 ^{1/4} +  \frac{1}{\ep |\log \ep|^3}\right) + \OO (\ep ^{-1} |\log \ep| ^{-3})
\end{eqnarray}
where we have used (\ref{propD}). Going back to (\ref{calculImut}) and (\ref{calculImut1}), combining with (\ref{kineticelectro}) and (\ref{propD}), the result is proved.

\end{proof}

With the result of Proposition \ref{pro:kineticinf} in hand it is an easy task to complete the 

\begin{proof}[Proof of Proposition \ref{pro:lowerbound}]

Collecting (\ref{infkineticnew}) and (\ref{lowbound4}) we have
\begin{eqnarray}\label{rassemble inf}
\E [u] &\geq& \left( 2 \pi \sum_{j\in J_*} \chiin (a_j) d_j \right)^2 I_* \left( 1 - C \Om_1 ^{1/4} \right) \nonumber \\
&+& 2\pi H(R_*) \left( 1 + o(1) \right)\sum_{j\in J_{*}}  \chiin(a_j) d_j \nonumber \\
&-& C \frac{\sum_{j\in J_*} \chiin(a_j) d_j}{\ep |\log \ep |^3} - C \frac{|\log \ep|\ep ^{\al}}{\ep ^2 \log |\log \ep| } \nonumber \\
&\geq & \left( 2 \pi \sum_{j\in J_*} \chiin (a_j) d_j \right)^2 I_* \left( 1 - C \Om_1 ^{1/4} \right) \nonumber \\
&+& 2\pi H(R_*) \left( 1 +o(1)\right)\sum_{j\in J_{*}}  \chiin(a_j) d_j  - C \frac{\log |\log \ep|}{\ep ^2 |\log \ep|^{1+\alt} }.
\end{eqnarray}
We have used (\ref{Omega1ter}) and (\ref{parameters}) for the second inequality. There only remains to minimize the sum of the first two terms above with respect to $2\pi \sum_{j\in J_*} \chiin(a_j)d_j$ and use (\ref{Omega1bis}) and (\ref{Omega1ter}) to obtain
\begin{equation}\label{enfin inf}
\E[u] \geq -\frac{H(R_*)^2}{4 I_*} (1+o(1)) - C \frac{\log |\log \ep|}{\ep ^2 |\log \ep|^{1+\alt} } \geq -\frac{H(R_*)^2}{4 I_*} (1+o(1)).
\end{equation}
The second inequality holds true because we choose $\alt$ large enough and $H(R_*) \propto -\Om_1  \ep ^{-1}$. This concludes the proof of the second inequality in (\ref{borneinf}).\\
Note for later use that in \eqref{rassemble inf} we have dropped the term
\[
C \sum_j g^2 (a_j) \chiout (a_j) |d_j||\log \ep| .
\]
Keeping this term and combining the lower bound we have just proved with the upper bound to the energy we obtain the estimate
\begin{equation}\label{dernier estim}
C \sum_j g^2 (a_j) \chiout (a_j) |d_j||\log \ep| \ll \frac{\Om_1 ^2}{\ep ^2}.
\end{equation}
\end{proof} 

\section{Asymptotics for the vorticity}\label{sec:vorticity}

In this Section we provide the proofs of Theorems \ref{theo:vorticity} and \ref{theo:vorticity exp}. Actually, most of the ingredients are scattered in the preceding subsections. In particular it is rather straightforward from the proofs in Subsection 3.3 to deduce

\begin{pro}[\textbf{Asymptotics for the modified vorticity}]\label{pro:asym mut}\mbox{}\\
Let $\mut$ be defined as in (\ref{mutilde}). We have for any  test function $\phi \in C_c ^1 (\ac)$ 
\begin{equation}\label{asym mut}
\left| \int_{\ac} \mut \phi + \frac{H(R_*)}{2 I_*} \int_{\ac} \DirC \phi \right| \ll \frac{\Om_1}{\ep} \left( \int_{\ac} \frac{1}{g^2} |\nabla \phi|^2 \right) ^{1/2},
\end{equation}
uniformly in $\phi$.
\end{pro}

\begin{proof}
We first remark that, combining \eqref{Omega1ter}, (\ref{bornesupenergie}), (\ref{borneinf}) and (\ref{rassemble inf})
\begin{multline}\label{calculD}
- \frac{H(R_*)^2}{4 I_*} (1+o(1)) \geq \left( 2 \pi \sum_{j\in J_*} \chiin (a_j) d_j \right)^2 I_* \left( 1 +o(1) \right) \\
+ 2\pi H(R_*) \left( 1 +o(1) \right)\sum_{j\in J_{*}}  \chiin(a_j) d_j  - o\left(\frac{\Om_1 ^2 }{\ep ^ 2 }\right) .
\end{multline}
We deduce that 
\begin{equation}\label{calculD1}
2 \pi \sum_{j\in J_*} \chiin (a_j) d_j = -\frac{H(R_*)}{2I_*} (1+o(1))
\end{equation}
and thus, going back to (\ref{propD})
\begin{equation}\label{calculD2}
D = -\frac{H(R_*)}{2I_*} (1+o(1)).
\end{equation}
Next we note that, up to now, we have neglected one useful term in the proof of the lower bound, namely the third term in (\ref{calculImut}). Keeping this term and using (\ref{calculD1}), Equation (\ref{rassemble inf}) reads
\begin{equation}\label{re borne inf}
\E [u] \geq - \frac{H(R_*) ^2}{4 I_*} (1+o(1)) + \int_{\ac} \frac{1}{g^2} |\nabla \hche_{\muc} | ^2 (1+o(1)) - o\left(\frac{\Om_1 ^2}{\ep ^ 2 }\right) .
\end{equation}
Thus (recall that $\muc = \mut - D \DirC$), using (\ref{bornesupenergie}) and the first inequality in (\ref{borneinf})
\[
\int_{\ac} \frac{1}{g^2} |\nabla \hche_{\muc} | ^2 \ll \frac{\Om_1 ^2}{\ep ^2}
\]
which implies the result via (\ref{defihcheck}) and (\ref{calculD2}).
\end{proof}

Recall that $\mut$ is defined by neglecting the superfluid current in a region where no convenient lower bound to $|u|$ is available and considering the vorticity associated to the remaining current. One can see this procedure as a regularization of $\mu$, because we actually neglect the regions where we expect the phase of $u$ to be singular (therefore $|u|$ to go to zero). Proposition \ref{pro:asym mut} states that, modulo this regularization, one can estimate the vorticity in the dual norm of 
\[
\left(\int_{\ab} g ^{-2} |\nabla \phi| ^2 \right)^{1/2}.
\]
This is the natural norm associated with the minimization problem (\ref{defiIstar}) defining the contribution of the vortices to the energy. Note that such a control is not exactly stronger than a control in $(C_c ^1 (\ab)) ^*$ norm because of the strong inhomogeneity of the weight $g^{-2}$ (see Remark \ref{rem:com vortic}). It would be stronger for example if $g^2$ was a constant because of the embedding (in two dimensions) of $H^{-1}(\ab)$ in $(C_c ^1 (\ab)) ^*$.\\
The necessity to use the $L^{\infty}$ norm of the gradient of test functions to bound the left-hand side of (\ref{resultvorticity}) appears when approximating $\mu$ by $\mut$ in the proof below (that is when justifying (\ref{sketch:mum}) rigorously). The main task is now to estimate the effect of the currents that we have neglected when defining $\mut$.

\begin{proof}[Proof of Theorem \ref{theo:vorticity}]
%In this proof we use a new radial partition of unity. The radial functions $\xiinc$ and $\xioutc$ are a partition of unity in $\ab$ associated to $\ac$ and $\ab \setminus \ac$. Their definitions and main properties are similar to those of $\xiin$ and $\xiout$. In particular we can (as we did with $\xiout$) construct the partition of unity in such a way that
%\begin{eqnarray*}
%|\nabla \xiinc | &\leq& C \ep ^{-1} |\log \ep| \\
%|\nabla \xioutc | &\leq& C \ep ^{-1} |\log \ep|.
%\end{eqnarray*}

Let $\phi \in C_c ^1 (\ab)$.
We compute
\begin{eqnarray}\label{calcul vortic}
\int_{\A} \mu \phi &=& - \int_{\A} (iu,\nabla u) \nabla ^{\perp} \left( \chiin \phi  + \chiout \phi \right) \nonumber \\
&=& - \int_{\ac \setminus \cup_{j\in J} B_j} (iu,\nabla u) \nabla ^{\perp} \left( \chiin \phi \right) + \int_{\ac  \cap \cup_{j\in J} B_j} (iu,\nabla u) \nabla ^{\perp} \left( \chiin \phi \right) \nonumber \\
&-& \int_{\A} (iu,\nabla u) \nabla ^{\perp} \left( \chiout \phi \right). 
\end{eqnarray}
The first term is estimated using Proposition \ref{pro:asym mut} (note that $\chiin \phi$ has its support included in $\ac$). By definition 
\begin{multline}\label{terme principal}
- \int_{\ac \setminus BS_{\al } \setminus \cup_{j\in J} B_j} (iu,\nabla u) \nabla ^{\perp} \left( \chiin \phi \right) = \int_{\ac} \mut \chiin  \phi = - \frac{H(R_*)}{2 I_*} \int_{\ab} \DirC \chiin \phi + o\left( \frac{\Om_1}{\ep} \right) \left( \int_{\ab} \frac{1}{g^2} |\nabla \phi |^2\right) ^{1/2}  \\
=  - \frac{H(R_*)}{2 I_*} \int_{\ab} \DirC \phi + o\left( \frac{\Om_1}{\ep} \right) \left( \int_{\ab} \frac{1}{g^2} |\nabla \phi |^2\right) ^{1/2} + \OO \left( N ^ B _{\al} \ep |\log \ep | \Vert \phi \Vert_{L ^{\infty} (\ab)} \frac{\Om_1}{\ep} \right) \\
= - \frac{H(R_*)}{2 I_*} \int_{\ac} \phi + o\left( \frac{\Om_1}{\ep} \right) \left( \int_{\ab} \frac{1}{g^2} |\nabla \phi |^2\right) ^{1/2} + o \left( \Om_1 |\log \ep| \Vert \nabla \phi \Vert_{L ^{\infty} (\ab)}\right).
\end{multline}
We have used the fact that $\chiin = 1$ on $\Cet \cap PS_{\al}$. To pass to the second line it is thus sufficient to note that the length of $\Cet \setminus PS_{\al} = \Cet \cap AS_{\al}$ is of order $N ^ B _{\al} \ep |\log \ep |$. Also, we have
\[
\Vert  \phi \Vert_{L ^{\infty} (\A)} \leq C \ep |\log \ep| \Vert \nabla \phi \Vert_{L ^{\infty} (\A)} 
\]
because $|\A| \propto \ep |\log \ep|$ and $\phi$ vanishes on $\dd \A$. Recalling (\ref{numberbad}) we obtain the third line of (\ref{terme principal}).\\
We now show that the other terms in (\ref{calcul vortic}) are remainders, arguing as when dealing with (\ref{hstarmut}) in the proof of Proposition \ref{pro:kineticinf}. \\
The second term in the right-hand side is estimated exactly as the sum of the second and the third term in the right-hand side of (\ref{hstarmut}), using the small area of the region covered by the vortex balls. The result is 
\[
\left| \int_{\ac  \cap \cup _{j\in J} B_j} (iu,\nabla u) \nabla ^{\perp} \left( \chiin \phi \right) \right| \leq C |\log \ep| ^{-3} \Vert \nabla \phi \Vert_{L^{\infty} (\ab)}.
\]
Finally, the last term in (\ref{calcul vortic}) is estimated exactly as in (\ref{contribution out}), (\ref{contribution out 2}). We obtain
\begin{equation}\label{merde}
\left| \int_{\A} (iu,\nabla u) \nabla ^{\perp} \left( \chiout \phi \right) \right| \leq C \frac{|\log \ep|^3}{|\log \ep| ^{\alt/2}} \Vert \nabla \phi \Vert_{L^{\infty}}
\end{equation}
which concludes the proof, taking $\alt$ large enough.

\end{proof}

  % redefine the command that creates the equation no.
  %\setcounter{equation}{0}  % reset counter
  % redefine the command that creates the equation no.

We conclude this section by the 

\begin{proof}[Proof of Theorem \ref{theo:vorticity exp}]
The vortex balls entering in the definition of $\mu_e$ are those defined in Proposition \ref{pro:vortexballs} from which we discard the balls that are not included in $\ab$, that is the balls labeled by $j\in J_{\mathrm{in}}$ (see Equation \eqref{Jin}):
\[
\mu_e = \sum_{j\in \Jin} 2\pi d_j \delta_{a_j}.
\]
For the statement of the Theorem we have renamed $\Jin = K$.\\
Let us first note that, using Theorem \ref{theo:vorticity}, \eqref{asympt mue} is a consequence of \eqref{mue-mui}, we thus only prove the latter.\\
From the Jacobian estimate \eqref{JE} and Lemma \ref{lem:initialbound} we have, for any $\phi \in C^1_c (\ab \cap GS_{\al})$ 
\[
\bigg|\sum_{j\in \Jin}  2 \pi d_j \phi (\avi)- \int_{GS_{\al}\cap \at} \diff \rv \: \phi \:  \curl (iu,\nabla u) \bigg| \leq  C \left\Vert \nabla \phi \right\Vert_{L^{\infty}(GS_{\al})} |\log \ep|^{-2}.  
\]
where $C$ does not depend on $\phi$. To conclude the proof we only have to extend such a statement to all test functions $\phi \in C^1_c (\ab \cap GS_{\al})$. Let us pick such a function and write 
\[
\int_{\ab} \left(\mu - \mu_e \right) \phi = \int_{\ab} \left(\mu - \mu_e \right) \chiin \phi + \int_{\ab}\mu \chiout \phi - \sum_{j\in \Jin} d_j \chiout(a_j) \phi(a_j).
\]
For the first term we can use the Jacobian estimate because $\chiin \phi$ has support in $GS_{\al}$. This yields (using also Lemma \ref{lem:initialbound})
\[
\left| \int_{\ab} \left(\mu - \mu_e \right) \chiin \phi \right| \leq C |\log \ep| ^{-2} \left\Vert \nabla \phi \right\Vert_{L ^{\infty}}.
\]
The second term is 
\[
\int_{\ab}\mu \chiout \phi = - \int_{\ab} (iu,\nabla u) \nabla^{\perp} (\chiout \phi)
\]
and has already been estimated, see \eqref{merde}. For the third term we recall that on $\ab$
\[
g^2 \sim \tfm \geq \frac{C \Om ^{1/2}}{\ep |\log \ep|}, 
\]
thus, using $\left| \phi(a_j) \right| \leq C \ep |\log \ep| \left\Vert \nabla \phi \right\Vert_{L ^{\infty}}$ and \eqref{dernier estim}
\[
\left| \sum_{j\in \Jin} d_j \chiout(a_j) \phi(a_j) \right| \leq C \ep ^2 |\log \ep| ^2 \Om_1 ^{-1/2} \sum_j g^2 (a_j) \chiout (a_j) |d_j|\left\Vert \nabla \phi \right\Vert_{L ^{\infty}} \ll \Om_1 ^{3/2}|\log \ep|\left\Vert \nabla \phi \right\Vert_{L ^{\infty}}.
\]
This concludes the proof.
\end{proof}

\section*{Appendix A : The Cost Function and the Vortex Energy}\label{ap:costfunc}
\addcontentsline{toc}{section}{Appendix A : The Cost Function and the Vortex Energy}

\renewcommand{\theequation}{A.\arabic{equation}}
\renewcommand{\thesection}{A}
\setcounter{equation}{0}
\setcounter{pro}{0}

In this appendix we study the cost function (\ref{fonctioncout}) and provide the proof of Lemma \ref{lem:couts}.\\
It is convenient to define 
\begin{equation}\label{omega TF}
\om ^{\mathrm{TF}}:= \frac{2}{3 \sqrt{\pi}\ep}
\end{equation}
and study the related function 
\beq
	\label{gainTF}
	\gaintf(r) : =\half |\log\eps| \tfm(r) + \costtf(r),
\eeq
where
\beq
	\label{exp TF potential}
	\costtf(r) : = 2 \int_{\rtf}^r \diff s \: \vec{B}_{\optphtf}(r) \cdot \vec{e}_{\theta} \tfm(r) = \eps^2 \Omega^2 \int_{\rtf}^r \diff s \: \lf[ \Omega s - \lf( [\Omega] - \optphtf \ri) s^{-1} \ri] (s^2 - \rtf^2).
\eeq
In order to investigate the behavior of the infimum of $ \gaintf $ inside the bulk, it is convenient to rescale the quantities and set
\beq \label{rescale}
	z : = \eps \Omega (r^2 - \rtf^2),
\eeq
so that $ z $ varies on a scale of order one, i.e., more precisely $ z \in [0, 2/\sqrt{\pi} ] $. With such a choice the gain function can be easily estimated:
\begin{multline} \label{calcul FTF}
 	\costtf(r) = \frac{\eps^2 \Omega^2}{2} \int_{0}^{r^2-\rtf^2} \diff t \: t  \lf[ \Omega (t + \rtf^2) - [\Omega] + \optphtf \ri] \lf(t + \rtf^2 \ri)^{-1} =	\\
	\frac{\eps^2 \Omega^2}{2} \int_{0}^{r^2-\rtf^2} \diff t \:  t \lf( \Omega t - \frac{4}{3\sqrt{\pi}\eps} + \OO(1) \ri) \lf(1 - \frac{2}{\sqrt{\pi} \eps \Omega} + t \ri)^{-1}= 	\\
	\frac{1}{2\eps} \int_{0}^{z} \diff s \:  s  \lf( s - \frac{4}{3\sqrt{\pi}} + \OO(\eps) \ri) \lf(1 - \frac{2}{\sqrt{\pi} \eps \Omega} + \frac{s}{\eps\Omega} \ri)^{-1} = 	\\
	\frac{1}{2\eps} \int_{0}^{z} \diff s \: s \lf( s - \frac{4}{3\sqrt{\pi}} \ri) + \OO(|\log\eps|) = \frac{z^2}{6 \eps} \lf( z - \frac{2}{\sqrt{\pi}} \ri) + \OO(|\log\eps|),
\end{multline}
where we have used the approximation $ [1 - \OO((\eps\Omega)^{-1})]^{-1} = 1 + \OO((\eps\Omega)^{-1}) $. 
\newline
Applying the same rescaling to the energy cost function, we thus obtain
\beq
	\label{Hrescaling}
	\gaintf(r) : = \frac{\rgaintf(z)}{12\eps},
\eeq
where (recall that $\Om = \left( 2 \left(3\pi\right)^{-1} - \Om_1 \right) \ep ^{-2} |\log \ep| ^{-1}$)
\beq
	\label{Hrescaled}
	\rgaintf(z) = z\left( \frac{2}{\pi} - 3 \Om_1 - 2 z \lf( \frac{2}{\sqrt{\pi}} - z \ri) \right)- \OO(\eps|\log\eps|).
\eeq
Let us denote, for $z\in \left[0, \frac{2}{\sqrt{\pi}} \right]$ 
\begin{equation}\label{cout renorm}
k(z):= z\left( \frac{2}{\pi} -3 \Om_1- 2 z \lf( \frac{2}{\sqrt{\pi}} - z \ri) \right).
\end{equation}
It is straightforward to see that this function takes two local maxima and two local minima in $\left[0, \frac{2}{\sqrt{\pi}} \right]$. The maxima are at $\frac{2}{\sqrt{\pi}} $ and at
\[
z_{1} = \frac{2}{3\sqrt{\pi}} - \sqrt{\frac{1}{9\pi} + \frac{\Om_1}{2}}.
\]
The minima are at $z=0$ with $k(0)=0$ and at 
\[
z_2 = \frac{2}{3\sqrt{\pi}} + \sqrt{\frac{1}{9\pi} + \frac{\Om_1}{2}}.
\]
Computing this local minimum we obtain that $k(z_2) < 0$ if and only if $\Om_1 > 0$, thus $k(z_2)$ is the absolute minimum in this regime. More precisely, for $|\Om_1| \ll 1$ we have 
\begin{equation}\label{H(Rstar) rescal}
k(z_2) = -\frac{3}{\sqrt{\pi}} \Om_1 + \OO (\Om_1 ^2) 
\end{equation}
and thus, defining $R_* >0$ by
\begin{equation}\label{AdefiRstar}
R_* ^2 : = R_h ^2 + \left(\eps \Omega \right)^{-1} z_2
\end{equation}
we obtain, for $ (\ep |\log \ep|) ^{1/2} \ll |\Om_1| \ll 1$
\begin{equation}
\gaintf( R_* ) = -\frac{\Om_1}{4\sqrt{\pi}\ep}  + \OO \left( \frac{\Om_1 ^2}{\ep} \right).
\end{equation}
Also 
\[
k''(z_2) = \frac{1}{3\sqrt{\pi}} + \OO(\Om_1) 
\]
and 
\[
k'(0) = \frac{2}{\pi} - 3 \Om_1.
\]
Recalling that $k(z)$ increases from $0$ to $z_1$ and from $z_2$ to $1$ and decreases from $z_1$ to $z_2$, we have for any $z$ and an appropriate choice of constants $k_1,k_2,k_3$ such that $|z-z_2|>  k_1\Om_1 ^{1/4}$ and $|z|> k_2 \Om_1 ^{1/2}$
\[
k (z) > k_3 \Om_1 ^{1/2}.
\]
Collecting the preceding facts we have proved

\begin{pro}[\textbf{TF vortex energy}]\label{pro:TFvortexenergy}\mbox{}	\\
Let $\Om$ be of the form
\begin{equation}\label{Omega Again}
\Om = \frac{2}{3\pi \ep ^2 |\log \ep|} - \frac{\Om_1}{\ep ^2 |\log \ep|}
\end{equation}
with $ (\ep |\log \ep|)^{1/2}\ll \Om_1 \ll 1 $. Let $\vec{r} \in \tfa$. We have
\begin{eqnarray}\label{borne cout TF 1}
\gaintf(r) &\geq& -\frac{\Om_1}{4\sqrt{\pi}\ep}  + \OO \left( \frac{\Om_1 ^2}{\ep} \right) \\
\gaintf(r) &\geq& C \frac{\Om_1 ^{1/2}}{\ep} \mbox{ if } r \geq R_h + C \ep |\log \ep| \Om_1 ^{1/2} \mbox{ and } |r-R_*|\geq C \ep |\log \ep| \Om_1 ^{1/4} \label{borne cout TF 2}
\end{eqnarray}
\end{pro}

We now compare the cost TF function $\gaintf$ to the original function appearing in our analysis, $H$. The following result is a part of the proof of Proposition A.2 in \cite{CRY} and is sufficient for our purpose.

\begin{pro}[\textbf{Comparison of the cost functions}] \label{pro:compare grho}\mbox{}	\\
Let $\om$ and $g = g_{\A,\om}$ be defined as in Proposition \ref{pro:optimalphase}. Let $\Om$ be as above and $H$ be the cost function defined in (\ref{fonctioncout}). For any $\vec{r}\in \at$ (see definition (\ref{defi annt})) there holds 		
\begin{eqnarray}\label{g-rhoTF}
\left| g^2(r) - \tfm(r) \right| &\leq& \frac{C|\log \ep| ^{5/2}}{\ep ^{1/2} } \\
\left| F(r) - \costtf(r) \right| &\leq& \frac{C}{\ep |\log \ep| } \label{F-FTF}\\
\end{eqnarray}
and thus
\begin{equation}\label{H-HTF}
\left| H(r) - \gaintf(r) \right| \leq \frac{C}{\ep |\log \ep|}.
\end{equation}		
\end{pro}

We are now equipped to present the 

\begin{proof}[Proof of Lemma \ref{lem:couts}]
The proof of (\ref{coutJout}) is based on a simple estimate from \cite[Lemma 4.1]{CRY} that we recall
\begin{equation}\label{borneF}
|F(r)| \leq C \min\left( \frac{|r-R_{<}|}{\ep} g^2 (r), 1 + \frac{C|r-1|}{\ep ^2 |\log \ep |} \right).
\end{equation}
Note that this estimate stays valid under assumptions \eqref{Omega1} to \eqref{Omega1ter}. We then write 
\[
|d_j| \frac{1}{2}  g^2 (a_j) \left| \log \ep \right| \left(1-C \frac{\log \left| \log \ep \right|}{\left|\log \ep\right|}\right) + d_j \xiin (a_j) F(a_j) \geq C|d_j|   \left( g^2 (a_j) \left| \log \ep \right| \left(1-C \frac{\log \left| \log \ep \right|}{\left|\log \ep\right|}\right) - |F(a_j)| \right).
\]
On the other hand, we have from  the definition (\ref{Jout})
\[
||a_j|-R_<| \leq C \ep |\log \ep| \Om_1 ^{1/2} \ll \ep |\log \ep|
\]
for any $j\in \Jout$. The result (\ref{coutJout}) follows using (\ref{borneF}).\\
We turn to the energetic cost of the vortices in $\Jin$. First we deduce from (\ref{F-FTF}) that
\[
F(a_j) = \costtf (a_j) + \OO\left( \frac{1}{\ep |\log \ep|}\right) 
\]
for any $j\in \Jin$. It is then straightforward (recall that $\Om_1 \gg \log |\log \ep| |\log \ep|^{-1}$) from the computation (\ref{calcul FTF}) to obtain
\[
F(a_j) < - \frac{C}{\ep |\log \ep| ^{1/2}} < 0
\]
for any $j\in \Jin$ such that $|a_j| \leq 1 - \ep |\log \ep| ^{1/2} $. On the other hand, if $|a_j| \geq 1 - \ep |\log \ep| ^{1/2} $, $g^2 (a_j) \geq C(\ep |\log \ep|)^{-1}$ and thus, using (\ref{borneF})
\[
|F(a_j)| \ll g^2 (a_j).
\]
We deduce that for any $j\in J_-$
\[
\frac{1}{2}  g^2 (a_j) \left| \log \ep \right| |d_j| + d_j \xiin (a_j) F(a_j) \geq  C  |d_j| g^2 (a_j) |\log \ep| \geq C|d_j| \frac{\Om_1 ^{1/2}}{\ep}
\]
using (\ref{g-rhoTF}). On the other hand, if $j\in J_+$ 
\[
\frac{1}{2}  g^2 (a_j) \left| \log \ep \right| |d_j| + d_j \xiin (a_j) F(a_j)  = |d_j| \left( \frac{1}{2}  g^2 (a_j) \left| \log \ep \right|  +  \xiin (a_j) F(a_j) \right)
\] 
and thus (\ref{coutJmoinsplus}) follows from Propositions \ref{pro:TFvortexenergy} and \ref{pro:compare grho} in the case where $j\in J_+$. We also use $\xiin \leq 1$. \\
There remains to show that (\ref{coutJetoile}) holds when $j \in J_*$, which is a consequence of (\ref{borne cout TF 1}) and (\ref{H-HTF}) once one has recalled that $\xiin (a_j) = 1$ for any $j\in J_*$.
\end{proof}

\section*{Appendix B : Useful results from \cite{CRY}}\label{ap:not}
\addcontentsline{toc}{section}{Appendix B : Useful results from \cite{CRY}}

\renewcommand{\theequation}{B.\arabic{equation}}
\renewcommand{\thesection}{B}
\setcounter{equation}{0}
\setcounter{pro}{0}

\setcounter{equation}{0}  % reset counter

In this appendix we gather for the convenience of the reader several results of \cite{CRY} that are essential ingredients of the proof of our main results.

We begin with estimates of the GP minimizer. The following is Proposition 2.2 of \cite{CRY}. 

	\begin{pro}[\textbf{Exponential smallness of $ \gpm $ inside the hole}]
		\label{cry pro:GP exp small}
		\mbox{}	\\
		As $ \eps \to 0 $ and for any $ \rv \in \B $,
		\beq
			\label{cry eq:exp small}
			\lf| \gpm (\rv) \ri|^2 \leq C \eps^{-1} |\log\eps|^{-1} \: \exp \lf\{ - \frac{1-r^2}{1 - \rtf^2} \ri\}.
		\eeq 
		Moreover there exists a strictly positive constant $ c $ such that for any $ \OO(\eps^{7/6}) \leq r \leq \rtf - \OO(\eps^{7/6}) $,
		\beq
			\label{cry eq:improved exp small}
			\lf| \gpm(\rv) \ri|^2 \leq C  \eps^{-1} |\log\eps|^{-1}  \: \exp \lf\{ - \frac{c}{\eps^{1/6}} \ri\}.
		\eeq
	\end{pro}

We also employ several useful properties of $g$, starting with Proposition 2.4 of \cite{CRY}.

\begin{pro}[\textbf{Preliminary estimates for $ g$}]\label{cry pro:g prelim}
		\mbox{}	\\
		As $ \eps \to 0 $ and for any $ \omega \in \Z $ such that $ |\omega| \leq \OO(\eps^{-1}) $,
		\beq
			\label{cry eq:g estimates}
			\lf\| g^2 - \tfm \ri\|_{L^2(\A)} = \OO(1),	\hspace{1,5cm}	\lf\| g \ri\|_{L^{\infty}(\A)}^2 \leq \lf\| \tfm \ri\|_{L^{\infty}(\A)} \lf(1 + \OO(\sqrt{\eps|\log\eps|} \ri).
		\eeq
	\end{pro}

We next state the exponential decay of $g$  \cite[Proposition 2.5]{CRY}.

\begin{pro}[\textbf{Exponential smallness of $ g $ inside the hole}]
		\label{cry pro:g exponential smallness}
		\mbox{}	\\
		As $ \eps \to 0 $ and for any $ \rv \in \A $
		\beq
			\label{cry eq:g exp small}
			 g^2(r) \leq C \eps^{-1} |\log\eps|^{-1} \exp \lf\{ - \frac{1-r^2}{1 - \rtf^2} \ri\}.
		\eeq 
		Moreover there exists a strictly positive constant $ c $ such that for any $ r \leq \rtf - \OO(\eps^{7/6}) $,
		\beq
			\label{cry eq:g improved exp small}
			g^2(r) \leq C \eps^{-1} |\log\eps|^{-1} \exp \lf\{ - \frac{c}{\eps^{1/6}} \ri\}.
		\eeq
	\end{pro}

Finally, it is very useful to know that $g^2$ is very close to $\tfm$ in a $L^{\infty}$ sense \cite[Proposition 2.6]{CRY}.

\begin{pro}[\textbf{Pointwise estimate for $ g $}]
		\label{cry pro:point GP dens}
		\mbox{}	\\
		As $ \eps \to 0 $ and for any $ \omega \in \Z $ with $ |\omega| \leq \OO(\eps^{-1}) $
		\beq
			\label{cry eq:pointwise bounds}
			\lf| g^2(r) - \tfm(r) \ri| \leq C \eps^{2} |\log\eps|^{2} (r^2 - \rtf^2)^{-3/2} \tfm(r)
		\eeq
		for any $ \rv \in \tfd $ such that $ r \geq \rtf + \OO(\eps^{3/2} |\log\eps|^2) $.
	\end{pro}

\vspace{1cm}
\noindent{\bf Acknowledgments.} I thank Xavier Blanc and Sylvia Serfaty for their suggestions on the manuscript, along with Michele Correggi, Jakob Yngvason and Vincent Millot for interesting discussions. The hospitality of the {\it Erwin Schr\"{o}dinger Institute} (ESI) is also gratefully acknowledged. This work is supported by {\it R\'{e}gion Ile-de-France} through a PhD grant.

\end{document}